\documentclass[a4paper,onecolumn,11pt,accepted=2022-06-06]{quantumarticle}
\pdfoutput=1
\usepackage{amsmath}
\usepackage{amssymb}
\usepackage{amsthm}
\usepackage{graphicx,color}
\usepackage{hyperref}
\usepackage{cite}

\newtheorem{thm}{Theorem}
\newtheorem{prop}{Proposition}
\newtheorem{lemma}{Lemma}
\newtheorem{corol}{Corollary}

\theoremstyle{remark}
\newtheorem{remark}{Remark} 

\theoremstyle{definition}

\numberwithin{equation}{section}

\newcommand{\tr}{\operatorname{tr}}

\newcommand{\sinc}{\operatorname{sinc}}

\newcommand{\T}{\operatorname{\mathcal{T}}}
\newcommand{\E}{\operatorname{\mathbb{E}}}

\newcommand{\Prob}{\operatorname{Prob}}
\renewcommand{\d}{\mathrm{d}}
\newcommand{\rmd}{\mathrm{d}}
\newcommand{\rmi}{\mathrm{i}}
\newcommand{\rme}{\mathrm{e}}

\newcommand{\abs}[1]{\left|{#1}\right|}

\definecolor{dgreen}{rgb}{0,0.5,0}

\definecolor{delete}{cmyk}{0.5,0,0,0}

\begin{document}
\title[One bound to rule them all]{One bound to rule them all: from Adiabatic to Zeno}
\date{May 23, 2022}
\author{Daniel Burgarth}
\orcid{0000-0003-4063-1264}
\affiliation{Center for Engineered Quantum Systems, Macquarie University, 2109 NSW, Australia}
\author{Paolo Facchi}
\orcid{0000-0001-9152-6515}
\affiliation{Dipartimento di Fisica and MECENAS, Universit\`a di Bari, I-70126 Bari, Italy}
\affiliation{INFN, Sezione di Bari, I-70126 Bari, Italy}
\author{Giovanni Gramegna}
\orcid{0000-0001-7532-1704}
\affiliation{Dipartimento di Fisica, Universit\`a di Trieste, I-34151 Trieste, Italy}
\affiliation{INFN, Sezione di Trieste, I-34151 Trieste, Italy}
\affiliation{Eberhard-Karls-Universit\"at T\"ubingen, Institut f\"ur Theoretische Physik, 72076 T\"ubingen, Germany}
\author{Kazuya Yuasa}
\orcid{0000-0001-5314-2780}
\affiliation{Department of Physics, Waseda University, Tokyo 169-8555, Japan}
\maketitle
\begin{abstract}
We derive a universal nonperturbative bound on the distance between unitary evolutions generated by time-dependent Hamiltonians in terms of the difference of their integral actions. We apply our result to provide explicit error bounds for the rotating-wave approximation and generalize it beyond the qubit case. 
We discuss the error of the rotating-wave approximation over long time and in the presence of time-dependent amplitude modulation. 
We also show how our universal bound can be used to derive and to generalize other known theorems such as the strong-coupling limit, the adiabatic theorem, and product formulas, which are relevant to quantum-control strategies including the Zeno control and the dynamical decoupling. Finally, we prove generalized versions of the Trotter product formula, extending its validity beyond the standard scaling assumption.
\end{abstract}

\section{Motivations}
Compare two unitary evolutions
\begin{equation}
U_1(t)=\T \exp\!\left(-\rmi \int_0^t \d s\, H_1(s)\right), \qquad U_2(t)=\T \exp\!\left(-\rmi \int_0^t \d s\, H_2(s)\right),
\end{equation}
where $\T$ denotes the time-ordering operator. What is the relation between these evolutions and their (possibly time-dependent) Hamiltonians $H_1(t)$ and $H_2(t)$?

If the Hamiltonians are close $H_1(t)\approx H_2(t)$, say in a time interval $t\in[0,T]$, then the evolutions they generate are close. 
The converse is not necessarily true.
The evolutions can be close $U_1(t)\approx U_2(t)$ and their distance can be small even if their Hamiltonians $H_1(t)$ and $H_2(t)$ are not close. 
A simple commutative example is a pair of Hamiltonians  
\begin{equation}
H_1(t)= H , \qquad 
H_2(t)= (1+ \kappa \cos \kappa^2 t) H,
\end{equation}
with a bounded self-adjoint operator $H=H^\dag$ and a large constant $\kappa\gg 1$. 
Their difference is very large, $H_2(t)-H_1(t)=O(\kappa)$, although highly oscillating, while the corresponding evolutions,
\begin{equation}
U_1(t) = \rme^{-\rmi t H}, \qquad 
U_2(t)=\rme^{-\rmi \left(t + \frac{1}{\kappa} \sin \kappa^2 t\right) H},
\end{equation} 
are very close, $U_1(t)\approx U_2(t)$, with a distance of order $O(1/\kappa)$. 
What happens is that the large difference of the Hamiltonians is averaged out by the fast oscillations and has negligible effects on the evolution. 
This is the essence of the averaging methods in dynamical systems~\cite{Arnold,Sanders}.

A peculiar consequence of noncommutativity is that the same phenomenon can take place also for constant \emph{time-independent} Hamiltonians, as shown in the following example. 
Consider the Hamiltonians
\begin{equation}
H_1 = \kappa Z, \qquad 
H_2 = \kappa Z + X,
\end{equation}
where $X$ and $Z$ are the first and third Pauli matrices, respectively (we will also use the second Pauli matrix $Y$), and $\kappa\gg 1$ as above. 
Both Hamiltonians are of order $O(\kappa)$ and their difference $H_2-H_1=X$ is of order $O(1)$. 
However, one can show that~\cite{ref:unity1}
\begin{equation}
\rme^{-\rmi t(\kappa Z + X)} = \rme^{-\rmi \kappa tZ } + O\!\left(\frac{1}{\kappa}\right),
\end{equation}
which would be blatantly false if $X$ and $Z$ commuted.
Again, the source of this phenomenon can be traced back to an averaging effect. 
In the rotating frame generated by $H_1$ (namely, in the interaction picture with respect to $H_1$), one gets $\hat{H}_1(t)=0$ and
\begin{equation}
\hat{H}_2(t) = \rme^{\rmi  \kappa tZ } X \rme^{-\rmi  \kappa tZ} = \cos (2 \kappa t)X - \sin (2 \kappa t)Y,
\end{equation} 
instead of $H_1$ and $H_2$, respectively.
Their difference $\hat{H}_2(t)-\hat{H}_1(t)$ is of order $O(1)$ but is fast oscillating and has negligible effects $O(1/\kappa)$ on the evolution.

The moral is that in general the distance between two Hamiltonians loosely bounds the distance between the evolutions they generate. 
A better physical quantity that controls the divergence is instead the difference of the two integral actions 
\begin{equation}
S_{21}(t) = \int_0^t \d s\,[H_2(s)-H_1(s)],
\label{eq:distact}
\end{equation}
in a suitable rotating frame.
We will give an explicit bound on the distance between two evolutions in terms of the difference between their actions~\eqref{eq:distact}, and thus
show that 
\begin{equation}
S_{21}\approx 0  \quad \Rightarrow \quad U_2\approx U_1.
\label{eq:actioncontrol}
\end{equation}
Then, we will show the effectiveness of such bound in proving a plethora of old and new results in simple ways.

\subsection{Overview of Applications and Previous Related Results}
\paragraph{Rotating-Wave Approximation.}
Coherent driving of quantum systems plays a central role in many areas of physics and chemistry. The specific task with arguably the highest demands on the accuracy of the desired operations carried out by such driving is the manipulation of two-level systems, or qubits, that form the basic elements of a quantum computer. Periodically driven two-level systems are important prototypes in diverse phenomena in nearly every subfield of physics such as optics~\cite{Eberly}, nuclear magnetic resonance~\cite{Mehring}, superconductor devices~\cite{ref:GuideSupercondQubits}, solid-state systems~\cite{ref:RWA-experiment}, and their applications to quantum information.

The rotating-wave approximation (RWA) replaces highly oscillatory components in the Schr\"odinger equation with time-averaged quantities, which often leads to much simpler and analytically tractable models, as well as capturing the essential physics. But how is it justified? At first encounter, the removal of terms from a differential equation which are usually \emph{of the same order} as the others seems ad hoc. Some geometrical interpretation (often using a lot of literal \emph{hand-waving} describing rotations of Bloch vectors) or the term \emph{off-resonance} might be provided to back up the approximation, but often these explanations have a tautological flavor. On the other hand, there are many perturbative methods which lead to a RWA\@. To name a few, there are average-Hamiltonian methods~\cite{Mehring, Haeberlen, ref:JamesAverageHamiltonian} and approaches based on the Magnus expansion and the Floquet theory~\cite{Series,ref:DiVincenzo}. While these methods are very efficient at higher orders of the RWA (such as the Bloch-Siegert shift~\cite{Eberly}), they do not give bounds on how good the naive RWA is and how it becomes exact in a certain limit. The convergence of these series is also often a subtle issue, as can be seen already for the Magnus series~\cite{ref:Blanes-Mugnus}. Recently, analytical solutions of  the time-dependent Schr\"odinger equation were discovered for linearly driven qubit systems~\cite{Xie} and were put in relationship with the Floquet theory~\cite{Schmidt}. These solutions are however rather unwieldy and make it hard to provide simple bounds for the RWA, as well as not being applicable to more general cases.

In the mathematical-physics literature, rigorous proofs of the RWA are well established~\cite{ref:RWA-Chambrion2012,ref:RWA-Augier2019,ref:RWA-RobinAugierBoscainSigalotti2020,ref:RWA-AugierBoscainSigalotti2020}, but not well known in the wider community. The basic idea is to use averaging methods developed by Krylov, Bogoliubov, and Mitropolsky (KBM) for nonlinear dynamical systems in the 1950s~\cite{Sanders}. In particular, it is shown that the solution in the RWA converges uniformly to the true solution on finite intervals of time as the drive frequency goes to infinity. This is also the idea we shall follow here. However, the KBM theory is not well suited to quantum mechanics, which is manifestly linear. In quantum theory, it should be possible to develop a self-contained approach which essentially follows similar steps to the KBM theory, but provides much better and more explicit nonperturbative bounds. The integration-by-part lemma presented in Sec.~\ref{sec:IntegrationByPartLemma} allows us to develop such bounds and to generalize the RWA further. A related but different idea is also mentioned in Ref.~\cite[Chap.~8]{ref:DAlessandroControlText2}, but error bounds are not explicitly worked out. We will also prove that the RWA is not \emph{eternal}~\cite{ref:EternalAdiabatic}: small errors accumulate over time. Such bounds are useful in quantum information and control, where increasing precision is needed to match stringent fault-tolerance thresholds, and the validity of the RWA is not always good~\cite{ref:RWA-experiment}.

\paragraph{Adiabatic Evolutions and Strong Coupling.}
Adiabatic theorems, concerning approximations of the evolutions generated by slowly varying Hamiltonians~\cite{ref:Messiah,ref:KatoAdiabatic,ref:Avron1999Gapless,ref:Joye2007,ref:Avron2012Contracting,ref:unity1}, rely on the same principle as that of the RWA, although maybe in a different flavor. In this case, the separation of the timescale on which the evolution occurs and the timescale on which the Hamiltonian changes hinders the transitions between different eigenspaces of the Hamiltonian during the evolution. At first sight, the connection between this mechanism and the RWA might seem difficult to grasp. It however becomes clear when the adiabatic evolution is studied in a suitable frame rotating with the eigenspaces of the slowly varying Hamiltonian. In this frame, one can see that the transitions between different eigenspaces are generated by fast components which are averaged out in the effective evolution. It is then intriguing to see how these apparently unrelated results may be derived using the same technique.

A version of adiabatic theorem involving a time-independent Hamiltonian also plays an important role in the control of quantum systems via the quantum Zeno dynamics~\cite{ref:QZS,ref:PaoloSaverio-QZEreview-JPA,ref:ZenoPaoloMarilena}, and precisely in its manifestation through the strong-coupling limit~\cite{ref:SchulmanJumpTimePRA,ref:ControlDecoZeno,ref:QZEExp-Ketterle,ref:QZEExp-Schafer:2014aa,ref:unity1,ref:GongPRL,ref:GongPRA,ref:EternalAdiabatic,ref:ZenoKAM}. The framework developed here can be used to reproduce these results independently. The simplicity of our approach also allows us to derive new versions of adiabatic theorem, which require less stringent assumptions on the form of the slowly varying Hamiltonians.

\paragraph{Generalized Product Formulas.} 
Trotter's product formula is widely used in physics at various levels, ranging from fundamental problems such as Feynman's path-integral formulation of quantum mechanics~\cite{Feynman,Simon} to practical ones such as Hamiltonian simulation on quantum computers~\cite{Suzuki,suzuki1991general,Sieberer19}. A matter of utmost importance for applications is the ability to control the digitization errors introduced by product formulas, and to find suitable generalizations of the formulas which can be used flexibly in particular practical problems considered. A remarkable example of this versatility is represented by symmetry protected quantum simulations~\cite{ref:ProductFormulaTaylor}, where the symmetries of the target Hamiltonian are exploited to greatly reduce the error of the simulation, which is achieved by alternating the simulation steps with unitary transformations generated by the symmetries of the system. The rapid alternation of several noncommuting Hamiltonians can be formally represented by an evolution generated by a time-varying Hamiltonian which is piecewise constant. Then, if the alternating Hamiltonians follow some particular structure, the resulting evolution can be effectively described by their average effect, a feature which can be used in practical applications of several sorts. We will show how various product formulas can be derived using our main theorem, and we will provide their explicit error bounds, which improve some of the existing ones~\cite{ref:ProductFormulaTaylor}. Another motivation for studying product formulas was recently highlighted in Ref.~\cite{ref:alphaTrot} in the attempt to establish a bridge between different control techniques such as the strong-coupling and bang-bang controls, which both yield quantum Zeno dynamics. To this aim, a ``rescaled'' version of the Trotter product formula was proved in Ref.~\cite{ref:alphaTrot}, with an analytical bound on the error which depends on the scaling parameter, although numerical evidence suggested an error uniform in the scaling parameter. Using the main theorem introduced in this paper, we will prove that the error in the rescaled product formula is indeed uniform.

\subsection{Summary of Results and Paper Outline}
This paper contains several new results together with improvements or simple proofs of known results. Our main purpose is to explore the wide range of applications of the universal bound in Lemma~\ref{lemma:divergence}\@. In particular, we will obtain the following new results (see also Fig.~\ref{ring} and Table~\ref{tab:Summary}):
\begin{figure}
\centering 
\includegraphics[width=.8\textwidth]{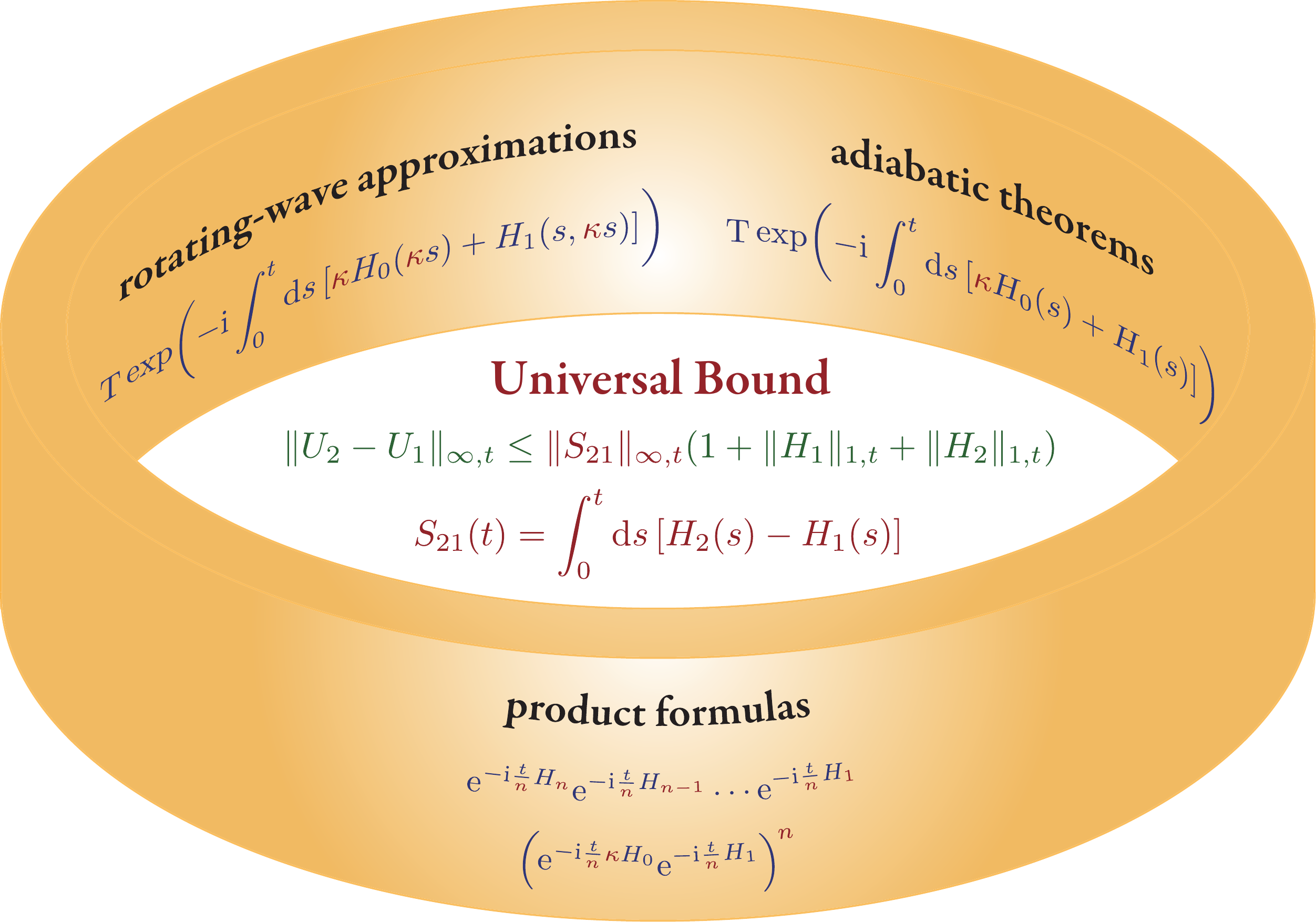}
\caption{One universal bound on the difference between two evolutions in terms of an integral action rules a variety of situations involving the separation of timescales, ranging from RWAs to adiabatic evolutions to strong-coupling limits to Zeno dynamics to dynamical decoupling to product formulas.} 
\label{ring}
\end{figure}
\begin{table}
\caption{List of evolutions whose limits as $\kappa\to+\infty$ or $n\to+\infty$ are proved in this paper, from adiabatic theorems to strong-coupling limits to RWAs to product formulas, including Zeno dynamics, dynamical decoupling, and bang-bang control. Explicit bounds on their convergence are all ruled and provided by one universal bound given in Lemma~\ref{lemma:divergence}. The evolutions are generated by Hamiltonians $H(s)$, $H_0(s)$, $H_1(s)$, etc., and include pulsed dynamics, where continuous Hamiltonian evolutions are interspersed with instantaneous unitaries $U_j$, $V_j$, etc., while their limit evolutions are generated by average Hamiltonians $\overline{H}$, Zeno Hamiltonians $H_Z(s)$, adiabatic connections $A(s)$, etc. See the relevant sections/statements for their definitions and the conditions required for the limits. Many results are new or are improvements of known results. See the explanation in the text.}
\label{tab:Summary}
\centering 
\begin{tabular}{l}
\hline
\hline
Universal Bound (Lemma~\ref{lemma:divergence}, Theorem~\ref{thm:OneMoreThm}, Corollary~\ref{cor:OneMoreThmbar})
\\
\medskip
\qquad
$\displaystyle
\| U_2 -U_1 \|_{\infty,t}  \le \| S_{21}\|_{\infty,t}( 1 +  \|H_1\|_{1,t}  + \|H_2\|_{1,t})
$,\quad
$\displaystyle
S_{21}(t) = \int_0^t \d s\,[H_2(s)-H_1(s)]
$
\\
Theorem~\ref{thm:periodic}. Eternal Approximation of Periodic Hamiltonian $H(s)=H(s+\tau)$
\\
\smallskip
\qquad
$\displaystyle
\T\rme^{-\rmi\int_0^t\rmd s\,H(\kappa s)}
\approx
\rme^{-\rmi t\overline{H}_\kappa}
\vphantom{\int_0^t}
$,\quad for any long times $t$
\\
Section~\ref{sec:RWA}. Rotating-Wave Approximation
\\
\smallskip
\qquad
$\displaystyle
\T\rme^{-\rmi\int_0^t\rmd s\,[\kappa H_0(\kappa s)+H_1(s,\kappa s)]}
\approx
\T\rme^{-\rmi\int_0^t\rmd s\,\kappa H_0(\kappa s)}
\T\rme^{-\rmi\int_0^t\rmd s\,\overline{H}(s)}
\vphantom{\int_0^t}
$
\\
Theorem~\ref{thm:StrongCouplingLimit}. Strong-Coupling Limit
\\
\smallskip
\qquad
$\displaystyle
\rme^{-\rmi t(\kappa H_0+H_1)}
\approx
\rme^{-\rmi t(\kappa H_0+H_Z)}
\vphantom{\int_0^t}
$
\\
Theorem~\ref{thm:AdiabaticTheorem}. Adiabatic Theorem
\\
\smallskip
\qquad
$\displaystyle
\T\rme^{-\rmi\int_0^t\rmd s\,\kappa H_0(s)}
\approx
\T\rme^{-\rmi\int_0^t\rmd s\,[\kappa H_0(s)+A(s)]}
\vphantom{\int_0^t}
$
\\
Theorem~\ref{thm:genadiabatic}. Generalized Adiabatic Theorem
\\
\qquad
$\displaystyle
\T\rme^{-\rmi\int_0^t\rmd s\,H_\kappa(s)}
\approx
\T\rme^{-\rmi\int_0^t\rmd s\,[\kappa H_0(s)+G_{\kappa,Z}(s)+A(s)]}
\vphantom{\int_0^t}
$,\quad for $\frac{1}{\kappa}H_\kappa(t)\to H_0(t)$
\\
Application of Theorem~\ref{thm:genadiabatic}. Strong-Coupling Limit for Time-Dependent Hamiltonians
\\
\smallskip
\qquad
$\displaystyle
\T\rme^{-\rmi\int_0^t\rmd s\,[\kappa H_0(s)+H_1(s)]}
\approx
\T\rme^{-\rmi\int_0^t\rmd s\,[\kappa H_0(s)+H_Z(s)+A(s)]}
\vphantom{\int_0^t}
$
\\
Theorem~\ref{thm:ATrott} and Corollary~\ref{cor:RandomTrott}. Ergodic-Mean and Random Trotter Formulas
\\
\qquad
$\displaystyle
\Bigl(\rme^{-\rmi\frac{t}{n}H_n}\rme^{-\rmi\frac{t}{n}H_{n-1}}\cdots \,\rme^{-\rmi\frac{t}{n}H_1} \Bigr)^n
\to
\rme^{-\rmi t\overline{H}}
\vphantom{\int_0^t}
$
\\
\smallskip
\qquad
$\displaystyle
\Bigl(
\rme^{-\rmi\frac{t}{np}H_p}\rme^{-\rmi\frac{t}{np}H_{p-1}}\cdots \,\rme^{-\rmi\frac{t}{np}H_1}
\Bigr)^n
\to
\rme^{-\rmi t \overline{H}}
\vphantom{\int_0^t}
$
\\
Corollary~\ref{thm:Ukicks}. Frequent Unitary Kicks
\\
\smallskip
\qquad
$\displaystyle
U_{n+1}\rme^{-\rmi\frac{t}{n}H} U_n\rme^{-\rmi\frac{t}{n}H}\cdots\,U_2\rme^{-\rmi\frac{t}{n}H} U_1
\approx
U_{n+1}\cdots U_1\rme^{-\rmi t \overline{H}}
\vphantom{\int_0^t}
$
\\
Application of Corollary~\ref{thm:Ukicks}. Dynamical Decoupling
\\
\smallskip
\qquad
$\displaystyle
\Bigl( V_p^\dag\rme^{-\rmi\frac{t}{n p}H} V_p\cdots 	V_1^\dag\rme^{-\rmi\frac{t}{n p}H}V_1 \Bigr)^n 
\to
\rme^{-\rmi t \overline{H}}
\vphantom{\int_0^t}
$
\\
Corollary~\ref{cor:UnitaryKick}. Bang-Bang Control
\\
\smallskip
\qquad
$\displaystyle
\Bigl(U\rme^{-\rmi\frac{t}{n}H}\Bigr)^n
\approx
U^n \rme^{-\rmi tH_Z}
\vphantom{\int_0^t}
$
\\
Theorem~\ref{GTF}. Generalized Trotter Formula
\\
\medskip
\qquad
$\displaystyle
\Bigl(
\rme^{-\rmi \frac{t}{n}\kappa H_0}\rme^{-\rmi \frac{t}{n}H}
\Bigr)^n
\approx
\rme^{-\rmi t(\kappa H_0+H_1)}
\vphantom{\int_0^t}
$,\quad for $\kappa\le\frac{\theta}{\eta t}n$
\\
\hline
\hline
\end{tabular}
\end{table}
\begin{itemize}
\item Theorem~\ref{thm:OneMoreThm} is a specific instance of the  averaging method for quantum systems with explicit nonperturbative bound.
\item Theorem~\ref{thm:periodic} provides an averaged generator with eternal validity for periodic Hamiltonians.
\item Section~\ref{sec:RWA} applies our bound to the important case of the RWA for a qubit, develops nonperturbative bounds, and settles the question of the long-time (in)validity of the RWA\@. It also provides several generalizations beyond the qubit case.
\item Theorem~\ref{thm:AdiabaticTheorem} provides a concise and explicit bound on the adiabatic theorem.
\item Theorem~\ref{thm:genadiabatic} generalizes the adiabatic theorem to a larger class of Hamiltonians.
\item Remark~\ref{rmk:timedepstrong} generalizes the strong-coupling limit to time-dependent Hamiltonians.

\item Theorem~\ref{thm:ATrott} generalizes Trotter's product formula to ergodic sequences.
\item An interesting application of Theorem~\ref{thm:ATrott} is the random Trotter formula in Corollary~\ref{cor:RandomTrott}.
\item Bound~(\ref{eq:DD}) on the dynamical decoupling is exponentially better in $t$ than a previous bound~\cite{ref:unity2}.
\item Bound~(\ref{eqn:BoundUnitaryKick}) in 
Corollary~\ref{cor:UnitaryKick} improves a bound on the unitary kicks obtained in Ref.~\cite{ref:ProductFormulaTaylor} from $O(\frac{1}{n}\log n)$ to $O(\frac{1}{n})$. 
\item Theorem~\ref{GTF} solves a numerical conjecture of Ref.~\cite{ref:alphaTrot} by another generalization of Trotter's formula.
\end{itemize}

The rest of the paper is structured as follows. 
In Sec.~\ref{sec:IntegrationByPartLemma}, we derive the universal bound, which will be used throughout the paper in various applications. 
The main result is split in Lemma~\ref{lemma:divergence}, Theorem~\ref{thm:OneMoreThm}, and Corollary~\ref{cor:OneMoreThmbar}\@.
Lemma~\ref{lemma:divergence} gives a bound on the distance between unitary evolutions generated by general time-dependent Hamiltonians.
It is applied to time-dependent Hamiltonians depending on a common control parameter, and a bound on the convergence error in the distance in the limit of the parameter is provided in Theorem~\ref{thm:OneMoreThm}\@.
A particular case in which one of the two time-dependent Hamiltonians is independent of the control parameter is stated in Corollary~\ref{cor:OneMoreThmbar}\@.
In Sec.~\ref{sec:periodic}, the application of the universal bound to periodic Hamiltonians is discussed, along with an eternal approximation of the evolution valid for any long times. 
In Sec.~\ref{sec:RWA}, we apply our results to the RWAs, first considering the standard case with a qubit, and then discussing generalizations beyond the qubit and when there are multiple timescales brought into play. 
In Sec.~\ref{sec:AdiabaticTh}, we use the universal bound for adiabatic evolutions, including the adiabatic theorem, strong-coupling limit, and their generalizations. 
In Sec.~\ref{sec:productFormulas}, we focus on the applications involving product formulas, including the standard Trotter product formula, random Trotter formula, dynamical decoupling, bang-bang control, etc.
Finally, in Sec.~\ref{sec:Conclusions}, we conclude the paper with some comments.
The main results, Lemma~\ref{lemma:divergence}, Theorem~\ref{thm:OneMoreThm}, and Corollary~\ref{cor:OneMoreThmbar}, are all valid for just locally integrable generators (which are not necessarily continuous).
This is explicitly demonstrated in Appendix~\ref{app:locInt}\@.
A few basic lemmas relevant to the ergodic means are collected in Appendix~\ref{app:Ergodic}\@.
Moreover, a converse of~(\ref{eq:actioncontrol}) is shown in Appendix~\ref{app:canongauge}, and in Appendix~\ref{app:isospectral} we provide an interesting general lower bound for the long-term divergence between evolutions generated by non-isospectral Hamiltonians.

\section{The Bound to Rule Them All}
\label{sec:IntegrationByPartLemma}
Let us start by stating our main instrument, which will be used throughout the paper to estimate a variety of limit quantum evolutions.

In the following, we will always consider operators on a separable (\emph{not} necessarily finite-dimensional) Hilbert space $\mathcal{H}$. In particular, our main object of investigations will be unitary propagators $U(t)$ on $\mathcal{H}$ generated by locally integrable (\emph{not} necessarily continuous) time-dependent bounded Hamiltonians $H(t)$.

Given an operator-valued function $A(t)$ on $\mathbb{R}$, consider its $L^\infty$ and $L^1$ norms
\begin{equation}
\label{eq:Linf-L1def}
\|A\|_{\infty,t}=\sup_{s\in[0,t]}\|A(s)\|,\qquad
\|A\|_{1,t}=\int_0^t\d s\,\|A(s)\|,
\end{equation}
for all $t\ge 0$, where we use the spectral norm $\|A(s)\|$ for operators.
The time-dependent operator $A(t)$ is bounded if $\|A(t)\|$ is finite for every $t$, and it is locally integrable if $\|A\|_{1,t}$ is finite for all $t$.
[Notice that, if $A(t)$ is not continuous in $t$, $\|A\|_{\infty,t}$ might diverge even if $A(t)$ is bounded for all $t$.]

Now we state and prove the universal bound that will be used throughout the paper. 
\begin{lemma}[Integration-by-part lemma]
\label{lemma:divergence}
Consider two families of locally integrable time-dependent Hamiltonians $t\in \mathbb{R} \mapsto H_j(t)$, with $H_j(t)$ bounded and self-adjoint for all $t\in \mathbb{R}$ and $j=1,2$.
Let $t\mapsto U_j(t)$ be the unitary propagator generated by $H_j(t)$,
\begin{equation}
U_j(t) =\T\exp\!\left(
-\rmi\int_0^t\d s\, H_j(s)\right),
\end{equation}
for $j=1,2$, and define the integral action
\begin{equation}
\label{eq:S21def}
S_{21}(t) = \int_0^t \d s\,[H_2(s)-H_1(s)].
\end{equation}
Then, one has for $t\ge 0$
\begin{equation}
\label{eq:divergence}
U_2(t)-U_1(t) = -\rmi S_{21}(t) U_2(t) - \int_0^t \d s\, U_1(t) U_1(s)^\dag[H_1(s) S_{21}(s) - S_{21}(s) H_2(s)]U_2(s),
\end{equation}
so that
\begin{equation}
	\|U_2(t)-U_1(t)\| 
\le \| S_{21}(t)\|  + \int_0^t \d s\,\|S_{21}(s)\|(\|H_1(s)\|+\|H_2(s)\|), 
\label{eqn:IntermediateBound}
\end{equation}
and the following bound holds
\begin{equation}
\label{eq:divergencebound}
\| U_2 -U_1 \|_{\infty,t}  \le \| S_{21}\|_{\infty,t}( 1 +  \|H_1\|_{1,t}  + \|H_2\|_{1,t}).
\end{equation}
\end{lemma}
\begin{proof}
If $H_j(t)$ is continuous, then $\dot{U}_j(t) = - \rmi H_j(t) U_j(t)$, for $j=1,2$, and the proof is just an integration by parts.
Write
\begin{align}
U_2(t)-U_1(t) 
&=\int_0^t \d s\, U_1(t) \frac{\d}{\d s}[U_1(s)^\dag U_2(s)]
\nonumber\\
&=-\rmi \int_0^t \d s\, U_1(t)  U_1(s)^\dag[H_2(s)-H_1(s)]U_2(s) .
\label{eq:U2U1-H2H1}
\end{align}
Since $H_2(s)-H_1(s) = \dot{S}_{21}(s)$, by integrating by parts one has
\begin{align}
U_2(t)-U_1(t) &= -\rmi \int_0^t \d s\, U_1(t)  U_1(s)^\dag    \dot{S}_{21}(s) U_2(s) 
\nonumber\\
&= -\rmi S_{21}(t) U_2(t) + \rmi \int_0^t \d s\, U_1(t)[\dot{U}_1(s)^\dag S_{21}(s) U_2(s) + U_1(s)^\dag    S_{21}(s) \dot{U}_2(s)]
\nonumber
\displaybreak[0]
\\
&= -\rmi S_{21}(t) U_2(t) - \int_0^t \d s\, U_1(t) U_1(s)^\dag[H_1(s) S_{21}(s) - S_{21}(s) H_2(s)]U_2(s).
\end{align}
This is~\eqref{eq:divergence}.
Since integration by parts holds for absolutely continuous functions, the result remains valid for locally integrable $H_j(t)$.
See Theorem~\ref{thm:5} for an explicit derivation.
By taking the norm of the divergence~\eqref{eq:divergence}, one gets~\eqref{eqn:IntermediateBound}, that implies
\begin{equation}
\|U_2(t)-U_1(t)\| 
\le \| S_{21}(t)\|  + \| S_{21}\|_{\infty,t}(\|H_1\|_{1,t}  + \|H_2\|_{1,t}),	
\end{equation}
and the bound~\eqref{eq:divergencebound} follows.
\end{proof}

\begin{remark}[Constant Hamiltonians]
\label{rem:Constant}
If $H_j(t)= H_j$ ($j=1,2$) are independent of time, then $S_{21}(t) = t (H_2-H_1)$, and
\begin{equation}
\|S_{21}\|_{\infty,t} = t \|H_2-H_1\|, \qquad  \|H_j\|_{1,t} = t \|H_j\|,
\end{equation}
so that the bound~(\ref{eq:divergencebound}) yields
\begin{equation}
\| U_2 -U_1 \|_{\infty,t}  \le t \|H_2-H_1\|( 1 +  t \|H_1\| + t \|H_2\|).
\end{equation}
Notice, however, that in such a situation with constant Hamiltonians one can get a sharper bound directly from~\eqref{eq:U2U1-H2H1} as
\begin{equation}
\|U_2-U_1\|_{\infty,t}  \le \| H_2 - H_1 \|_{1,t} =  t \|H_2-H_1\|.
\label{eq:simplebound}
\end{equation}
As we will momentarily see, the bound~\eqref{eq:divergencebound} is very useful for time-dependent Hamiltonians when they tend to compensate on average so that even if the difference between the two Hamiltonians $H_2(t)-H_1(t)$ might be large their action is small,
\begin{equation}
S_{21}(t) = \int_0^t \d s\, H_2(s) - \int_0^t \d s\, H_1(s) \approx 0.
\end{equation}
\end{remark}

\begin{remark}[$p$-norm]
\label{rmk:pnorm}
Since
\begin{equation}
\|H_j\|_{1,t} \le t \|H_j\|_{\infty,t},
\end{equation}
for $j=1,2$, the bound~(\ref{eq:divergencebound}) can be further bounded by
\begin{equation}
\|U_2-U_1\|_{\infty,t}  \le \|S_{21}\|_{\infty,t}( 1 +  t \|H_1\|_{\infty,t} + t \|H_2\|_{\infty,t}).
\end{equation}
More generally, let
\begin{equation}
\|H_j\|_{p,t}  = \biggl( \int_0^t \d s\, \|H_j(s)\|^p \biggr)^{1/p}<+\infty,
\end{equation} 
with $1\le p\le \infty$. Then, by the H{\"o}lder inequality
\begin{equation}
\|H_j\|_{1,t} \le t^{1-\frac{1}{p}} \|H_j\|_{p,t},
\end{equation}
one gets
\begin{equation}
\| U_2 -U_1 \|_{\infty,t}  \le \|S_{21}\|_{\infty,t}\left( 1 +  t^{1-\frac{1}{p}} \|H_1\|_{p,t} + t^{1-\frac{1}{p}} \|H_2\|_{p,t}\right).
\end{equation}
\end{remark}

According to Remark~\ref{rem:Constant}, it would be useful for constant Hamiltonians to consider their actions in a rotating frame with respect to an arbitrary reference Hamiltonian $H_0(t)$, so that the relevant action in the rotating frame may be small as a consequence of an averaging mechanism. 
\begin{prop}[Rotating frame]
	\label{lem:divergenceR}
	Let $U_0(t)$ be the unitary propagator generated by some reference Hamiltonian $H_0(t)$ and define 
	\begin{equation}
		\label{eq:S21defR}
		\hat{S}_{21}(t) = \int_0^t \d s\, U_0(s)^\dag[H_2(s)-H_1(s)]U_0(s).
	\end{equation}
	Then,
	\begin{equation}
		\label{eq:divergenceboundR}
		\| U_2 -U_1 \|_{\infty,t}  \le \| \hat{S}_{21}\|_{\infty,t}(1 +  \|H_1- H_0\|_{1,t}  + \|H_2-H_0\|_{1,t}).
	\end{equation}
\end{prop}
\begin{proof}
	This follows simply by applying Lemma~\ref{lemma:divergence} to 
	$\hat{U}_j(t) = U_0(t)^\dag U_j(t)$ ($j=1,2$), whose generators are given by
	\begin{equation}\label{eq:refHamiltonian}
		\hat{H}_{j}(t)  = U_0(t)^\dag[H_j(t) - H_0(t)]U_0(t), 
	\end{equation}
	and by noticing that
	\begin{equation}
		\|\hat{U}_2(t) - \hat{U}_1(t)\| =  \| U_0(t)^\dag[U_2(t) - U_1(t)]\| = \| U_2(t) - U_1(t) \|,
	\end{equation}
	and $\|\hat{H}_{j}(t) \| = \| H_j(t) - H_0(t) \|$.
\end{proof}
\begin{remark}[Gauge freedom]
\label{rmk:gauge}
The freedom in the choice of the gauge Hamiltonian $H_0(t)$ makes this bound very useful in many applications, beyond the above-mentioned case of constant Hamiltonians.
In particular, by choosing as a reference Hamiltonian the average
\begin{equation}
H_0(t) = \frac{1}{2}[H_1(t) + H_2(t)],
\end{equation}
one gets
\begin{equation}
H_2(t)- H_0(t)  = -[H_1(t) - H_0(t)]  =  \frac{1}{2}[H_2(t) - H_1(t)],
\end{equation}
whence the bound~(\ref{eq:divergenceboundR}) is reduced to
\begin{equation}
\label{eq:divergenceboundS}
\| U_2 -U_1 \|_{\infty,t}  \le \| \hat{S}_{21}\|_{\infty,t}(1 +  \|H_2 - H_1 \|_{1,t}).
\end{equation}
This is a more symmetric version of Lemma~\ref{lemma:divergence}, which involves only the difference $H_2(t)-H_1(t)$ of the Hamiltonians and the action $\hat{S}_{21}(t)$ in the average frame as defined in~\eqref{eq:S21defR}.
For further discussion about rotating frames, and a proof of a converse of the inequality~\eqref{eq:divergenceboundR}, see Appendix~\ref{app:canongauge}\@.
\end{remark}

\begin{remark}[Unbounded Hamiltonians]
\label{rmk:UnboundedH}
Notice that the bound~\eqref{eq:divergenceboundR} can be easily extended to \emph{unbounded} Hamiltonians $H_1(t)$ and $H_2(t)$ whose difference
$H_2(t)-H_1(t)$ is bounded, by suitably choosing the (unbounded) reference Hamiltonian~$H_0(t)$ so that $H_j(t)-H_0(t)$, with $j=1,2$, are both bounded. Typical examples are the controlled Schr\"odinger operators in Sec.~4.1 of Ref.~\cite{ref:RWA-Chambrion2012}.
\end{remark}

An immediate consequence of Lemma~\ref{lemma:divergence} in the case of time-dependent Hamiltonians which depend on a control parameter $\kappa$ is the following convergence result, which we state as a Theorem due to its importance in applications.
\begin{thm}	
\label{thm:OneMoreThm}
Consider two families of integrable time-dependent Hamiltonians $t\in[0,T]\mapsto H_\kappa (t)$ and $t\in[0,T] \mapsto \overline{H}_\kappa (t)$, with $H_\kappa (t)$ and  $\overline{H}_\kappa (t)$ self-adjoint and bounded for all $t\in[0,T]$ and $\kappa\in A$, with $A\subset \mathbb{R}$ a set with a limit point $\kappa_0$ (possibly $\kappa_0=\infty$). 
Let $t\mapsto U_\kappa (t)$ and  $t\mapsto\overline{U}_\kappa(t)$ be the unitary propagators generated by $H_\kappa (t)$ and
$\overline{H}_\kappa (t)$, respectively,
\begin{equation}
U_\kappa (t) =\T\exp\!\left(
-\rmi\int_0^t\d s\, H_\kappa (s)\right),
\quad 
\overline{U}_\kappa (t) =\T\exp\!\left(
-\rmi\int_0^t\d s\, \overline{H}_\kappa (s)\right).
\end{equation}
Assume that 
\begin{equation}
\|S_\kappa\|_{\infty,T}(1+\| H_\kappa \|_{1,T} + \| \overline{H}_\kappa \|_{1,T})  \to 0 , \quad \text{as}\quad \kappa\to \kappa_0,
\label{eqn:ErgodicMeannew}
\end{equation}
where
\begin{equation}
S_\kappa (t)= \int_0^t \d s\,[H_\kappa (s) - \overline{H}_\kappa (s)].
\end{equation}
Then, one gets
\begin{equation}
U_\kappa (t) - \overline{U}_\kappa (t) \to 0, \quad\text{as}\quad \kappa \to \kappa_0,
\label{eq:adiabaticnew1}
\end{equation}
uniformly for $t\in[0,T]$.
The convergence error is bounded by
\begin{equation}
\label{eq:divergenceboundbarKK}
\| U_\kappa (t)  - \overline{U}_\kappa (t)\|   \le \| S_{\kappa}\|_{\infty,T}( 1 +  \| H_\kappa \|_{1,T}  + \|\overline{H}_\kappa \|_{1,T}).
\end{equation}
\end{thm}
\begin{proof}
From  Lemma~\ref{lemma:divergence} with $H_1(t)= \overline{H}_\kappa (t)$ and $H_2(t)= H_\kappa (t)$, one gets
\begin{equation}
\| U_\kappa (t)  - \overline{U}_\kappa (t)\|   
\le\| U_\kappa -\overline{U}_\kappa \|_{\infty,t}  \le \| S_\kappa \|_{\infty,t}( 1  + \|H_\kappa \|_{1,t}  +  \|\overline{H}_\kappa \|_{1,t})  \to 0,
\end{equation}
as $\kappa\to \kappa_0$.
\end{proof}

If the second Hamiltonian $\overline{H}_\kappa (t)$ in Theorem~\ref{thm:OneMoreThm} is independent of $\kappa$, we get the following result.
\begin{corol}	
\label{cor:OneMoreThmbar}
Consider an integrable time-dependent Hamiltonian $t\in[0,T] \mapsto H_\kappa (t)$, with $H_\kappa (t)$ self-adjoint and bounded for all $t\in[0,T]$ and $\kappa\in A$, with $A\subset\mathbb{R}$ a set with a limit point $\kappa_0$. 
Let $t\mapsto U_\kappa (t)$ be the unitary propagator generated by $H_\kappa (t)$,
\begin{equation}
U_\kappa (t) =\T\exp\!\left(
-\rmi\int_0^t\d s\, H_\kappa (s)\right).
\end{equation}
Assume that there exists an integrable time-dependent Hamiltonian $t \mapsto \overline{H}(t)$, with $\overline{H}(t)$ self-adjoint and bounded, such that
\begin{equation}
\|S_\kappa \|_{\infty,T}(1+\|H_\kappa \|_{1,T})  \to 0 , \quad \text{as}\quad \kappa\to \kappa_0,
\label{eqn:ErgodicMeannewbar}
\end{equation}
where
\begin{equation}
S_\kappa (t)= \int_0^t \d s\, [H_\kappa (s) - \overline{H}(s)].
\end{equation}
Then, one gets
\begin{equation}
U_\kappa (t) \to \overline{U}(t), \quad\text{as}\quad \kappa \to \kappa_0,
\label{eq:adiabaticnew1bar}
\end{equation}
uniformly for $t\in[0,T]$, where $t\mapsto \overline{U}(t)$ is the unitary propagator generated by $\overline{H}(t)$,
\begin{equation}
\overline{U}(t) = \T\exp\!\left(-\rmi\int_0^t\d s\, \overline{H}(s)\right).
\label{eq:adiabaticnew2bar}
\end{equation}
The convergence error is bounded by
\begin{equation}
\label{eq:divergenceboundbar}
\| U_\kappa (t)  - \overline{U}(t) \|   \le \| S_{\kappa}\|_{\infty,T} ( 1 +  \|H_\kappa \|_{1,T}  + \|\overline{H}\|_{1,T}).
\end{equation}
\end{corol}

\begin{remark}[$\|H_\kappa \|_{1,T}$ can be unbounded in $\kappa$]
\label{rmk:6}
Assumption~(\ref{eqn:ErgodicMeannewbar}) implies the condition
\begin{equation}
S_\kappa (t)\to0, \quad \text{as}\quad \kappa\to \kappa_0,
\label{eqn:ErgodicMeannew2}
\end{equation}
uniformly for $t\in[0,T]$.
Moreover, condition~\eqref{eqn:ErgodicMeannewbar} follows from assumption~\eqref{eqn:ErgodicMeannew2} if 
$\|H_\kappa \|_{1,T}$ is  bounded in $\kappa$, i.e.,
\begin{equation}
\sup_\kappa \|H_\kappa \|_{1,T} <+\infty,
\end{equation}
which happens, for instance, if $\sup_\kappa \|H_\kappa \|_{\infty,T} <+\infty$.
However, this condition is not necessary for Corollary~\ref{cor:OneMoreThmbar} to hold.
For example, 
\begin{equation}
H_\kappa (t)= \kappa^{1/3} \sin(\kappa t) H
\end{equation} 
has an unbounded norm 
\begin{equation}
\|H_\kappa \|_{1,t}\ge 2 \kappa^{1/3} \|H\| (t/\pi -1/\kappa)\to+\infty,\quad\text{as}\quad \kappa\to +\infty,
\end{equation} 
but the action $S_\kappa (t)$ with $\overline{H}(t)=0$ vanishes as $S_\kappa (t) = 2 \kappa^{-2/3} \sin^2(\kappa t/2) H \to 0$, and 
\begin{equation}
\|S_\kappa \|_{\infty,t} \|H_\kappa \|_{1,t} \le t \|S_\kappa \|_{\infty,t} \|H_\kappa \|_{\infty,t}  \le 2 t \kappa^{-1/3} \|H\|^2 \to 0,
\end{equation} 
as $\kappa\rightarrow+\infty$.
Therefore, Corollary~\ref{cor:OneMoreThmbar} applies.
\end{remark}

\begin{remark}[Rotating frame]
\label{rmk:effectiveH}
Most of the applications we are going to consider make use of the strategy introduced in Proposition~\ref{lem:divergenceR}, that consists in identifying a suitable rotating frame $U_{0,\kappa}(t)$ such that the averaged action~\eqref{eq:S21defR} vanishes as $\kappa\to \kappa_0$.
\begin{enumerate}
\item 
Let $H_{0,\kappa}(t)$ be the reference Hamiltonian generating $U_{0,\kappa}(t)$ and apply Theorem~\ref{thm:OneMoreThm} to the evolutions $\hat{U}_\kappa (t) = U_{0,\kappa}(t)^\dag U_\kappa (t)$ and $\hat{\overline{U}}_\kappa (t) = U_{0,\kappa}(t)^\dag \overline{U}_\kappa (t)$. 
If 
\begin{equation}
\|\hat{S}_{\kappa}\|_{\infty,T} 
(1+ \| H_\kappa - H_{0,\kappa} \|_{1,T} + \| \overline{H}_\kappa -  H_{0,\kappa} \|_{1,T})  \to 0 , \quad\text{as}\quad \kappa\to \kappa_0,
\end{equation}
where
\begin{equation}
\hat{S}_\kappa (t)= \int_0^t \d s\,  U_{0,\kappa}(s)^\dag [H_\kappa (s) - \overline{H}_\kappa (s)] U_{0,\kappa}(s),
\end{equation}
one gets, using~(\ref{eq:divergenceboundR}), that
\begin{equation}
\label{eq:divergenceboundbarKKe}
\|U_\kappa (t) - \overline{U}_\kappa (t)\| \le \|\hat{S}_{\kappa}\|_{\infty,T} 
(1+ \| H_\kappa - H_{0,\kappa} \|_{1,T} + \| \overline{H}_\kappa -  H_{0,\kappa} \|_{1,T})  \to 0 , 
\end{equation}
as $\kappa\to \kappa_0$,  for all $t\in[0,T]$.

\item
Moreover, let 
\begin{equation}
\label{eq:tildeHKAe}
\hat{H}_\kappa (t) = U_{0,\kappa}(t)^\dag[H_\kappa (t)   - H_{0,\kappa}(t)]U_{0,\kappa}(t)
\end{equation}
be the generator of $\hat{U}_\kappa (t) = U_{0,\kappa}(t)^\dag U_\kappa (t)$. Then, Corollary~\ref{cor:OneMoreThmbar} applied to $\hat{U}_\kappa (t)$
implies that as $\kappa\to \kappa_0$ the Hamiltonian 
\begin{equation}
H_\kappa (t) = H_{0,\kappa}(t) + U_{0,\kappa}(t)\hat{H}_\kappa (t)U_{0,\kappa}(t)^\dag
\end{equation} 
can be replaced by the effective Hamiltonian 
\begin{equation}
H^{\mathrm{eff}}_\kappa (t) = H_{0,\kappa}(t) + U_{0,\kappa}(t) \overline{H}(t) U_{0,\kappa}(t)^\dag,
\label{eq:effH}
\end{equation}
provided that
\begin{equation}
\| \hat{S}_{\kappa}\|_{\infty,T}( 1 +  \|\hat{H}_\kappa \|_{1,T})\to0,\quad\text{as}\quad \kappa\to \kappa_0,
\label{eqn:CondReference2}
\end{equation}
where
\begin{equation}
\hat{S}_\kappa (t)= \int_0^t \d s\, [\hat{H}_\kappa (s)  - \overline{H}(s)].
\label{eq:actionKbar}
\end{equation}
Indeed, the divergence between the evolutions $U_\kappa (t)$ and  $U^{\mathrm{eff}}_\kappa (t) = U_{0,\kappa}(t)\overline{U}(t)$ generated by $H_\kappa(t)$ and $H_\kappa^{\mathrm{eff}}(t)$, respectively, vanishes uniformly in time, with an error bounded by
\begin{equation}
\label{eq:divergenceboundbarAe}
\| U_\kappa (t)  - U^{\mathrm{eff}}_\kappa (t)\|   \le \| \hat{S}_{\kappa}\|_{\infty,T}( 1 +  \|\hat{H}_\kappa \|_{1,T}  + \|\overline{H}\|_{1,T})\to0,
\end{equation}
as $\kappa\to \kappa_0$,  for all $t\in[0,T]$.
\end{enumerate}
\end{remark}

\subsection{Periodic Hamiltonians}
\label{sec:periodic}
In most applications, the evolution $U_\kappa(t)$ is generated by a periodic Hamiltonian and the parameter $\kappa$ is related to its frequency, that can be very large. In such situations, we can get an even better control on the error between the true evolution $U_\kappa(t)$ and a unitary group $\exp(-\rmi t \overline{H}_\kappa)$ generated by a time-independent Hamiltonian $\overline{H}_\kappa$ [which is not necessarily the average of $H_\kappa(t)$]. In particular, we will show that by a suitable choice of $\overline{H}_\kappa$ the error can be uniformly bounded \emph{for arbitrarily large times}.

Consider a 1-periodic Hamiltonian
\begin{equation}
	H(t) = H(t+1),
\end{equation}
and the dynamics generated by $H_\kappa (t) = H(\kappa t)$, where $\kappa$ is large.
We can approximate the evolution $U_\kappa (t)$ generated by $H_\kappa(t)$ with the unitary group $\overline{U}(t)= \rme^{-\rmi t \overline{H}}$ generated by the (constant) average Hamiltonian
\begin{equation}
	\overline{H} 
	= \int_0^1 \d s\, H(s).
	\label{eq:avgHam}
\end{equation}  
Indeed, the relevant integral action reads
\begin{align}
	S_\kappa (t) 
	&=\int_0^t \d s\, [H_\kappa (s)- \overline{H}]
	\nonumber\\
	&=\frac{1}{\kappa} \int_0^{\kappa t} \d s\, H(s)- t \overline{H} 
	\nonumber\\
	&=\frac{\lfloor \kappa t \rfloor}{\kappa} \overline{H} +
	\frac{1}{\kappa} \int_0^{\{\kappa t\}} \d s\, H(s) - t \overline{H}
\nonumber
\displaybreak[0]
\\
	&=- \frac{\{ \kappa t \}}{\kappa} \overline{H} +
	\frac{1}{\kappa} \int_0^{\{\kappa t\}} \d s\, H(s), 
\end{align}
where $\lfloor x\rfloor\in\mathbb{N}$ is the integer part of $x$, i.e.\ the greatest integer less than or equal to $x$, and $\{x\}=x-\lfloor x\rfloor\in [0,1)$  is its fractional part. 
Therefore, we get $\|S_\kappa (t)\|  \leq 2 \|H\|_{1,1} /\kappa$ uniformly in $t$.
Since
\begin{equation}
	\| H_\kappa \|_{1,T} = \frac{1}{\kappa} \int_0^{\kappa T} \d s\, \|H(s)\| 
	\leq \left(T + \frac{1}{\kappa}\right)  \|H\|_{1,1}
\end{equation}
and $\|\overline{H}\|_{1,T}\leq T \|H\|_{1,1}$, from Corollary~\ref{cor:OneMoreThmbar} we get
\begin{equation}
	\| U_\kappa - \overline{U} \|_{\infty,T} 
	\leq  \frac{2}{\kappa} \|H\|_{1,1} \left[
	1 + \left(2T+\frac{1}{\kappa}\right)\|H\|_{1,1}
	\right]
	\to 0,
	\label{eq:boundper}
\end{equation}
as $\kappa\to+\infty$, so that $U_\kappa (t) \to \overline{U}(t)$ uniformly for $t\in[0,T]$.

Notice, however, that in general the distance~\eqref{eq:boundper} between the approximation $\overline{U}(t)$ and the true evolution $U_\kappa (t)$ does not remain small uniformly in time and it can increase with~$T$. 
In the following theorem, we will show how one can exploit the periodicity in order to give a uniform bound which is \emph{eternal} in time.
To this end, recall the Floquet theorem~\cite{Teschl}. Since $H_\kappa (t)=H(\kappa t)$ is ($1/\kappa$)-periodic, one has
\begin{equation}
	U_\kappa (t)
	=U_\kappa \!\left(\frac{\{\kappa t\}}{\kappa}\right) 
	\left[U_\kappa \!\left(\frac{1}{\kappa}\right)\right]^{\lfloor  \kappa t\rfloor}.
	\label{eq:floq}
\end{equation}
The strategy is to take, in place of the average Hamiltonian $\overline{H}$ in~\eqref{eq:avgHam}, a constant ($\kappa$-dependent) Hamiltonian $\overline{H}_\kappa = \overline{H} + O(1/\kappa)$ which generates \emph{exactly the same} evolution as $U_\kappa(t)$ at all periods $t=n/\kappa$ with $n\in\mathbb{N}$, namely, $U_\kappa (n/\kappa) =[U_\kappa (1/\kappa)]^n= \rme^{-\rmi\frac{n}{\kappa} \overline{H}_\kappa}$. 
This allows us to reduce the distance between the approximation $\overline{U}_\kappa(t)=\rme^{-\rmi t\overline{H}_\kappa}$ and the true evolution $U_\kappa(t)$ for any $t$ to the distance during only one period $0\le t<1/\kappa$.
\begin{thm}[Eternal approximation of periodic Hamiltonian]
	\label{thm:periodic}
	Consider an integrable time-dependent Hamiltonian $t\in\mathbb{R} \mapsto H(t)$, with $H(t)$ self-adjoint, bounded, and 1-periodic, 
	\begin{equation}
		H(t+1) = H(t),
	\end{equation} 
	for all $t$.
	For all $\kappa >0$, let $t\mapsto U_\kappa (t)$ be the unitary propagator generated by $H_\kappa (t)= H(\kappa t)$,
	\begin{equation}
		U_\kappa (t) =\T\exp\!\left(
		-\rmi\int_0^t\d s\, H(\kappa s)\right),
	\end{equation}
	and let $\overline{H}_\kappa$ be a bounded self-adjoint operator such that
	\begin{equation}
		\exp \!\left(
		-\frac{\rmi}{\kappa} \overline{H}_\kappa \right) =
		U_\kappa \!\left(\frac{1}{\kappa}\right) .
		\label{eq:coincide0}
	\end{equation}
	Then, one gets
	\begin{equation}
		\sup_{t\in\mathbb{R}}\| U_\kappa (t) - \rme^{-\rmi t \overline{H}_\kappa}\| = O\!\left(\frac{1}{\kappa} \right), \quad\text{as}\quad \kappa \to +\infty,
	\end{equation}
	the convergence error being bounded by
	\begin{equation}
		\label{eq:divergenceboundper}
		\| U_\kappa (t)  - \rme^{-\rmi t \overline{H}_\kappa} \|   \leq \frac{\theta}{\kappa}\|H\|_{1,1}\left( 1 +  \frac{\theta}{\kappa}\|H\|_{1,1}\right),
	\end{equation}
	with $\theta\leq 1+2\log2$.
\end{thm}
\begin{proof}
	The identity~\eqref{eq:coincide0} forces the approximate evolution $\overline{U}_\kappa (t) = \rme^{-\rmi t \overline{H}_\kappa}$ to coincide with $U_\kappa (t)$ after one period at $t=1/\kappa$, 
	namely $U_\kappa (1/\kappa) = \overline{U}_\kappa (1/\kappa)$,
	and thus, by the Floquet theorem~(\ref{eq:floq}), at all periods, $U_\kappa(n/\kappa)=\overline{U}_\kappa(n/\kappa)$ with $n\in \mathbb{N}$.
	The approximate Hamiltonian $\overline{H}_\kappa$ is given by a logarithm of the true evolution over one period, $\overline{H}_\kappa =\rmi \kappa \log U_\kappa (1/\kappa)$, and it is a perturbation of $\overline{H}$.
	To see this, note that
	\begin{equation}
		U_\kappa (t) =\T\exp\!\left(
		-\rmi\int_0^t\d s\, H_\kappa (s)\right) 
		= \T\exp\!\left(
		-\frac{\rmi}{\kappa}\int_0^{\kappa t}\d s\, H(s)\right), 
		\label{eq:rescale}
	\end{equation}
	and hence,
	\begin{equation}
		U_\kappa \!\left(\frac{1}{\kappa}\right)
		= \T\exp\!\left(
		-\frac{\rmi}{\kappa}\int_0^1\d s\, H(s)\right)
		= \exp \!\left(
		-\frac{\rmi}{\kappa} \overline{H}_\kappa \right)
		= \overline{U}_\kappa \!\left(\frac{1}{\kappa}\right).
		\label{eq:coincide}
	\end{equation}
	Then, for $\kappa>\|H\|_{1,1}$,
	\begin{align}
		\overline{H}_\kappa 
		& = \rmi \kappa \log U_\kappa \!\left(\frac{1}{\kappa}\right) 
		\nonumber\\
		&=	\int_0^1 \d s\, H(s) -\frac{\rmi}{2\kappa}\int_0^1 \d s \int_0^s \d u\, [H(s), H(u)] + O\!\left(\frac{1}{\kappa^2} \right) 
		\nonumber\\
		& = \overline{H} + O\!\left(\frac{1}{\kappa} \right) .
		\label{eq:Magnus}
	\end{align}
	This is nothing but the Magnus expansion, which in fact can be proved to be valid for $\kappa>\|H\|_{1,1}/\pi$. See Theorem~9 of Ref.~\cite{ref:Blanes-Mugnus}. 
	This fact, combined with the bound~\eqref{eq:boundper}, ensures that the evolution $\overline{U}_\kappa(t)$ generated by $\overline{H}_\kappa$ also approximates the true evolution $U_\kappa(t)$ at $O(1/\kappa)$ at least for $t\in[0,T]$.
	Moreover, from~\eqref{eq:floq} and~\eqref{eq:coincide}, we have
	\begin{equation}
		U_\kappa (t) - \overline{U}_\kappa (t) = \left[ U_\kappa \!\left(\frac{\{\kappa t\}}{\kappa}\right) - \overline{U}_\kappa \!\left(\frac{\{\kappa t\}}{\kappa}\right) \right]\overline{U}_\kappa \!\left(\frac{\lfloor \kappa t \rfloor}{\kappa}\right),
	\end{equation}
	whence
	\begin{equation}
		\| U_\kappa (t) - \overline{U}_\kappa (t) \| = \left\| U_\kappa \!\left(\frac{\{\kappa t\}}{\kappa}\right) - \overline{U}_\kappa \!\left(\frac{\{\kappa t\}}{\kappa}\right) \right\| ,
	\end{equation}
	so that the distance is traced back to one period 
	$0\leq \{\kappa t\}/\kappa < 1/\kappa$ and is thus bounded uniformly in time $t$.

	Let us prove the explicit error bound~\eqref{eq:divergenceboundper}. From~\eqref{eq:rescale}, we get
	\begin{equation}
		U_\kappa \!\left(\frac{\{\kappa t\}}{\kappa}\right) = \T\exp\!\left(
		-\frac{\rmi}{\kappa}\int_0^{\{\kappa t\}}\d s\, H(s)\right) ,
	\end{equation}
	while 
	\begin{equation}
		\overline{U}_\kappa \!\left(\frac{\{\kappa t\}}{\kappa}\right) = 
		\exp \!\left( -\rmi \frac{\{\kappa t\}}{\kappa} \overline{H}_\kappa \right) .
	\end{equation}
	By using Lemma~\ref{lemma:divergence} for the evolutions $U_1(t)$ and $U_2(t)$ generated by $H_1 =  \overline{H}_\kappa /\kappa$ and $H_2(t)= H(t)/\kappa$, respectively, one has
	\begin{align}
		\left\| U_\kappa \!\left(\frac{\{\kappa t\}}{\kappa}\right) - \overline{U}_\kappa \!\left(\frac{\{\kappa t\}}{\kappa}\right) \right\|
		&=\| U_2(\{\kappa t\}) -U_1(\{\kappa t\}) \|
		\nonumber\\
		&\le\| U_2 -U_1 \|_{\infty,1}
		\vphantom{\frac{1}{\kappa}}
		\nonumber\\
		&\leq \| S_{21}\|_{\infty,1}\left( 1 +  \frac{1}{\kappa}\|H\|_{1,1}  + \frac{1}{\kappa}\|\overline{H}_\kappa \|\right),
	\end{align}
	where the action is given by
	\begin{equation}
		S_{21}(t)=\frac{1}{\kappa} \int_0^t \d s\, [H(s) - \overline{H}_\kappa],
	\end{equation}
	and is bounded by 
	\begin{equation}
		\| S_{21}\|_{\infty,1} \leq \frac{1}{\kappa}	(\|H\|_{1,1} + \|\overline{H}_\kappa \|).
	\end{equation}
	Moreover,
	\begin{equation}
		\|\overline{H}_\kappa \| = \kappa \left\| \log\!\left(1 + \left[U_\kappa \!\left(\frac{1}{\kappa}\right)  -1 \right]\right)\right\|\leq 
		- \kappa \log\!\left(1 - \left\| U_\kappa \!\left(\frac{1}{\kappa}\right)  -1 \right\| \right),
	\end{equation}
	for $\| U_\kappa (1/\kappa) -1 \|<1$.
	Now, notice that
	\begin{equation}
		\left\| U_\kappa \!\left(\frac{1}{\kappa}\right) -1 \right\| 
		= \frac{1}{\kappa}\,\biggl\|\int_0^1 \d s \, H(s) U_\kappa \!\left(\frac{s}{\kappa}\right) \biggr\|\leq \frac{1}{\kappa} \|H\|_{1,1},
	\end{equation}
	whence
	\begin{equation}
		\|\overline{H}_\kappa \| \leq 
		- \kappa \log\!\left(1 - \frac{1}{\kappa}\|H\|_{1,1} \right) \leq (2\log 2)\|H\|_{1,1},
	\end{equation}
	for $\kappa \geq 2 \|H\|_{1,1}$.
	Therefore,
	\begin{equation}
		\| U_\kappa (t) - \overline{U}_\kappa (t) \| 
		= \left\| U_\kappa \!\left(\frac{\{\kappa t\}}{\kappa}\right) - \overline{U}_\kappa \!\left(\frac{\{\kappa t\}}{\kappa}\right) \right\| 
		\leq \frac{\theta}{\kappa}\|H\|_{1,1} \left( 1 +  \frac{\theta}{\kappa}\|H\|_{1,1}\right),
	\end{equation}
	with $\theta=1+2\log2$.
	Since the right-hand side is larger than 2 for 
	$\kappa < 2 \|H\|_{1,1}$, this bound can be trivially extended to all $\kappa >0$.
\end{proof}

\begin{remark}[Large-time approximation]
\label{rmk:largetime}
In general, one does not have an explicit analytic expression for the solution $\overline{H}_\kappa$ of equation~\eqref{eq:coincide0}, and for sufficiently large $\kappa$ one relies on the series expansion of the logarithm as in~\eqref{eq:Magnus}. Notice, however, that for practical purposes one might be content to consider an asymptotic expansion of $\overline{H}_\kappa$ up to a given order $O(1/\kappa^\ell)$ for some $\ell\geq 2$, and 
approximate the evolution with the unitary group generated by a Hamiltonian
\begin{equation}
	\widetilde{H}_\kappa = \overline{H}_\kappa + O\!\left(\frac{1}{\kappa^\ell}\right).
\end{equation}
In such a case, in general the approximation $\exp(-\rmi t \widetilde{H}_\kappa)$ is no longer eternal (but see Remark~\ref{rmk:isospectral} for a notable exception). However,
it can be proved to work up to a time $T=O(\kappa^{\ell -1})$, which becomes larger and larger for higher and higher $\ell$. Indeed,
from~\eqref{eq:simplebound} we get
\begin{equation}
	\| \rme^{-\rmi t \widetilde{H}_\kappa}  - \rme^{-\rmi t \overline{H}_\kappa} \| \leq  t \| \widetilde{H}_\kappa  - \overline{H}_\kappa \| = O\!\left(\frac{t}{\kappa^\ell}\right),
\end{equation}
and by the triangle inequality and~\eqref{eq:divergenceboundper},
\begin{equation}
	\sup_{t\in[0,T]}\| U_\kappa (t)  - \rme^{-\rmi t \widetilde{H}_\kappa} \|   \leq \frac{\theta}{\kappa}\|H\|_{1,1}\left( 1 +  \frac{\theta}{\kappa}\|H\|_{1,1}\right) + O\!\left(\frac{T}{\kappa^\ell}\right).
	\end{equation}
Therefore,
\begin{equation}
	\sup_{t\in[0,T]} \| U_\kappa (t)  - \rme^{-\rmi t \widetilde{H}_\kappa} \|   = O\!\left(\frac{1}{\kappa}\right),
\quad \text{for} \quad T=O(\kappa^{\ell -1}), 
\label{eq:largetime}
	\end{equation}
which should be compared with the approximation~\eqref{eq:boundper},
\begin{equation}
	\sup_{t\in[0,T]} \| U_\kappa (t)  - \rme^{-\rmi t \overline{H}} \|   = O\!\left(\frac{1}{\kappa}\right),
\quad \text{for} \quad T=O(1),
	\end{equation}	
given by the average Hamiltonian $\overline{H}$.
\end{remark}

\begin{remark}[Isospectral perturbations]
	\label{rmk:isospectral}
	Notice that 
	any \emph{isospectral} ($1/\kappa$)-perturbation of  $\overline{H}_\kappa$ in Theorem~\ref{thm:periodic}
	does the job. Indeed, let 
	\begin{equation}
		\label{eq:overoverH}
		\overline{\overline H}_\kappa = W_\kappa \overline{H}_\kappa W_\kappa^\dag,
	\end{equation}
	with $W_\kappa = 1 + O(1/\kappa)$ and unitary. Then,
	\begin{equation}
		\bigl\| \rme^{-\rmi t \overline{\overline H}_\kappa} - \rme^{-\rmi t \overline{H}_\kappa}\bigr\|	= \bigl\|\bigl[ W_\kappa, \rme^{-\rmi t \overline{H}_\kappa}\bigr] W_\kappa^\dag \bigr\|	  = \bigl\| \bigl[ W_\kappa -1 , \rme^{-\rmi t \overline{H}_\kappa} \bigr] \bigr\| \leq 2 \| W_\kappa -1 \| = O\!\left(\frac{1}{\kappa}\right),
	\end{equation}
	uniformly in time $t$.
	Therefore, by the triangle inequality
	\begin{equation}
		\sup_{t\in\mathbb{R}}\,\bigl\| U_\kappa (t) - \rme^{-\rmi t \overline{\overline{H}}_\kappa}\bigr\| = O\!\left(\frac{1}{\kappa} \right), \quad\text{as}\quad \kappa \to +\infty.
	\end{equation}
	Moreover, if $\overline{H}_\kappa$  has a pure point spectrum, the isospectral Hamiltonians in~\eqref{eq:overoverH} are the only generators of eternal approximations as shown in Proposition~\ref{prop:isospectral} in Appendix~\ref{app:isospectral}\@.
	The proof is based on the idea that evolutions which are not isospectral eventually diverge. See Ref.~\cite[Eq.~(36) of the Supplemental Material]{ref:ZenoKAM}.
\end{remark}

\begin{remark}[General period]
	A periodic Hamiltonian with a general period $\tau>0$, 
	\begin{equation}
		H(t) = H(t+\tau),
	\end{equation}
	is reduced to a 1-periodic Hamiltonian $\widetilde{H}(s)=\widetilde{H}(s+1)$ by scaling time as $t\to s=t/\tau$ and
	\begin{equation}
		H(t) \to \widetilde{H}(s)=\tau H(\tau s).
	\end{equation}
	The evolution operator $U_\kappa(t)$ generated by $H(\kappa t)$ can then be obtained by the evolution operator $\widetilde{U}_\kappa(t)$ generated by $\widetilde{H}(\kappa s)$ as $U_\kappa(t)=\widetilde{U}_\kappa(t/\tau)$. The average Hamiltonian~\eqref{eq:avgHam} is replaced by
	\begin{equation}
		\overline{H} 
		=\frac{1}{\tau}\int_0^\tau \d s\, H(s),
		\label{eq:avgHamtau}
	\end{equation} 
	and the distance between the group $\overline{U}(t)=\rme^{-\rmi t \overline{H}}$ and the true evolution $U_\kappa(t)$ is bounded by
	\begin{align}
		\| U_\kappa - \overline{U} \|_{\infty,T}
		&= \bigl\|\widetilde{U}_\kappa -\overline{\widetilde{U}}\bigr\|_{\infty,T/\tau}
		\nonumber
		\\
		&\leq\frac{2}{\kappa } \|\widetilde{H}\|_{1,1} \left[
		1 + \left(\frac{2 T}{\tau} +\frac{1}{\kappa}\right)\|\widetilde{H}\|_{1,1} \right]  
		\nonumber
		\displaybreak[0]
		\\
		&=\frac{2}{\kappa } \|H\|_{1,\tau} \left[
		1 + \left(\frac{2 T}{\tau} +\frac{1}{\kappa}\right)\|H\|_{1,\tau} \right].
		\label{eq:boundpertau}
	\end{align}
	The eternal bound~\eqref{eq:divergenceboundper} is translated into
	\begin{equation}
		\label{eq:divergenceboundpertau}
		\| U_\kappa (t)  - \rme^{-\rmi t \overline{H}_\kappa} \|   \leq \frac{\theta}{\kappa}\|H\|_{1,\tau} \left( 1 +  \frac{\theta}{\kappa}\|H\|_{1,\tau}\right),
	\end{equation}
	where $\overline{H}_\kappa $ is a solution of the equation
	\begin{equation}
		\exp \!\left(
		-\rmi\frac{\tau}{\kappa} \overline{H}_\kappa \right) =
		U_\kappa \!\left(\frac{\tau}{\kappa}\right) ,
		\label{eq:coincidetau}
	\end{equation}
	that is
	\begin{equation}
		\overline{H}_\kappa 
		= \frac{\rmi\kappa}{\tau} \log U_\kappa \!\left(\frac{\tau}{\kappa}\right)
		=\overline{H}
		-\frac{\rmi}{2\kappa\tau}\int_0^{\tau} \d s \int_0^s \d u\, [H(s), H(u)] + O\!\left(\frac{1}{\kappa^2} \right),
		\label{eqn:MagnusHtau}
	\end{equation}
	for $\kappa>\|H\|_{1,\tau}/\pi$.
\end{remark}

\section{Rotating-Wave Approximation}
\label{sec:RWA}
\subsection{Qubit Example}\label{qubitexample}
Consider the Hamiltonian of a two-level atom with a natural frequency $\omega_0$ (almost) resonantly coupled with an oscillating laser field,
\begin{equation}
H(t)=\frac{1}{2}\omega_0Z+g\cos(\omega t)X,
\label{eq:origH}
\end{equation}
where $\omega>0$ is the frequency of the drive, and we denote the detuning by $\delta=\omega_0-\omega$. 
The time-dependent part of this Hamiltonian, $g\cos(\omega t)X$, can be decomposed into two components, one containing the ``co-rotating'' terms 
\begin{equation}
	H_{\mathrm{co}}(t)=\frac{1}{2}g( \rme^{-\rmi \omega t} \sigma_+ + \rme^{\rmi \omega t} \sigma_-),
\end{equation}
and the other containing the ``counter-rotating'' terms
\begin{equation}
	H_{\mathrm{counter}}(t)=\frac{1}{2}g( \rme^{\rmi \omega t} \sigma_+ + \rme^{-\rmi \omega t} \sigma_- ),
\end{equation}
where $\sigma_\pm=(X \pm \mathrm{i}Y)/2$ are the ladder operators. 
This terminology can be understood by looking at these two components in the reference frame rotating with the constant part of the Hamiltonian, $U_{0,\omega_0}(t)=\rme^{- \frac{\rmi}{2}\omega_0 tZ}$:
\begin{align}
	U_{0,\omega_0}(t)^\dagger H_{\mathrm{co}}(t)U_{0,\omega_0}(t)&=\frac{1}{2}g( \rme^{\rmi \delta t} \sigma_+ + \rme^{-\rmi \delta t} \sigma_-),\\
	U_{0,\omega_0}(t)^\dagger H_{\mathrm{counter}}(t)U_{0,\omega_0}(t)&=\frac{1}{2}g( \rme^{\rmi(\omega_0+\omega) t} \sigma_+ + \rme^{-\rmi (\omega_0+\omega) t} \sigma_-).
\end{align}
These show that, if the detuning $\delta$ is small, i.e.~$\omega_0\approx\omega$, the co-rotating terms rotate more or less with the reference frame, while the counter-rotating terms rotate in the opposite direction. 
The RWA drops the counter-rotating terms and approximates the evolution by
\begin{equation}\label{eq:RWAapprox}
	U(t)\approx U_\mathrm{RWA}(t),
\end{equation}
where $U_\mathrm{RWA}(t)$ is generated by
\begin{equation}
	\label{eq:RWAHam}
	H_\mathrm{RWA}(t)=\frac{1}{2}\omega_0Z +\frac{1}{2}g( \rme^{-\rmi \omega t} \sigma_+ + \rme^{\rmi \omega t} \sigma_-),
\end{equation}
with the co-rotating terms.
The evolution $U_\mathrm{RWA}(t)$ generated by $H_\mathrm{RWA}(t)$ is indeed easily dealt with compared to the original $U(t)$ generated by $H(t)$, since in the reference frame rotating with
\begin{equation}
H_{0,\omega}= \frac{1}{2}\omega Z, \qquad  U_{0,\omega}(t) = \rme^{-\frac{\rmi}{2}\omega tZ},
\label{eqn:RWA_Reference}
\end{equation} 
the Hamiltonian $H_\mathrm{RWA}(t)$ in~(\ref{eq:RWAHam}) in the RWA is transformed to a constant one as 
\begin{equation}
\hat{H}_\mathrm{RWA}
=U_{0,\omega}(t)^\dag[H_\mathrm{RWA}-H_{0,\omega}]U_{0,\omega}(t)=\frac{1}{2}\delta Z + \frac{1}{2}gX.
\label{eqn:HRWA_Reference}
\end{equation}
Then, the evolution it generates is simply given by $\hat{U}_\mathrm{RWA}(t)=U_{0,\omega}(t)^\dag U_\mathrm{RWA}(t)=\rme^{-\rmi t\hat{H}_\mathrm{RWA}}$, and the evolution in the original frame is explicitly obtained as $U_\mathrm{RWA}(t)=\rme^{-\frac{\rmi}{2}\omega tZ}\rme^{-\rmi t\hat{H}_\mathrm{RWA}}$.

The validity of the approximation~\eqref{eq:RWAapprox} can be verified by using Corollary~\ref{cor:OneMoreThmbar} with $\kappa_0 = +\infty$.
It also provides an upper bound on the error associated with the approximation.
More specifically, we take the strategy described in Remark~\ref{rmk:effectiveH}. 
We work in the reference frame rotating with the driving frequency $\omega$, specified by~(\ref{eqn:RWA_Reference}).
The generator of $\hat{U}_\omega(t)=U_{0,\omega}(t)^\dag U(t)$ in the rotating frame reads 
\begin{align}
\hat{H}_\omega(t)
&=U_{0,\omega}(t)^\dag[H(t)   - H_{0,\omega}]U_{0,\omega}(t)
\nonumber\\
&=\frac{1}{2}(\omega_0-\omega)Z + g \cos(\omega t) \rme^{\frac{\rmi}{2}\omega tZ} X \rme^{-\frac{\rmi}{2}\omega tZ}
\nonumber\\ 
&=\frac{1}{2}\delta Z + g \cos(\omega t)[\cos(\omega t) X-\sin(\omega t)Y]
\nonumber\\ 
&=\frac{1}{2}\delta Z +
\frac{1}{2}g[(1+\cos2\omega t)X  - \sin (2 \omega t) Y].
\end{align}
Observe here that in the limit of large $\omega$ the integral action of $\hat{H}_\omega(t)$ converges to
\begin{equation}
\int_0^t \d s\, \hat{H}_\omega(s)
\to t \overline{H},\quad\text{as}\quad\omega\to+\infty,
\end{equation}
where
\begin{equation}
\overline{H}
= \lim_{\tau\to+\infty} \frac{1}{\tau}\int_0^\tau \d s\, \hat{H}_\omega(s)
= \frac{\omega}{\pi} \int_0^{\pi/\omega} \d s\, \hat{H}_\omega(s)
= \frac{1}{2}\delta Z + \frac{1}{2}gX 
\label{eqn:AveHamRWAqubit}
\end{equation}
is the time-average of $\hat{H}_\omega(t)$, and is actually $\overline{H}=\hat{H}_\mathrm{RWA}$ introduced in~(\ref{eqn:HRWA_Reference}). 
Therefore, the action defined by
\begin{equation}
\hat{S}_\omega (t)
=\int_0^t\d s\,[\hat{H}_\omega (s)- \overline{H}]
=\frac{g}{4\omega}[\sin (2 \omega t) X  -(1-\cos2\omega t)Y]
\end{equation}
vanishes as $\omega\to+\infty$. 
The condition~(\ref{eqn:CondReference2}) is satisfied, and by Corollary~\ref{cor:OneMoreThmbar} the evolution $\hat{U}_\omega(t)$ is well approximated by the evolution generated by the average Hamiltonian $\overline{H}$ in~(\ref{eqn:AveHamRWAqubit}).
In the original frame,
\begin{equation}
U(t)-U_\mathrm{RWA}(t)\to0,\quad\text{as}\quad\omega\to+\infty,
\end{equation}
where $U_\mathrm{RWA}(t)$ is generated by an effective Hamiltonian
\begin{align}
H_\mathrm{RWA}(t)
&=H_{0,\omega} + U_{0,\omega}(t)\overline{H}U_{0,\omega}(t)^\dag
\nonumber\\
&=\frac{1}{2}(\omega+\delta) Z +\frac{1}{2}g\rme^{-\frac{\rmi}{2}\omega tZ} X \rme^{\frac{\rmi}{2}\omega tZ} 
\nonumber\\
&=\frac{1}{2}\omega_0Z+\frac{1}{2}g[\cos(\omega t) X + \sin(\omega t) Y],
\end{align}
which is exactly the RWA Hamiltonian~\eqref{eq:RWAHam}.

A bound on the error of the approximation follows from~\eqref{eq:divergenceboundbarAe}.
By noting
\begin{equation}
\|\hat{S}_\omega\|_{\infty,t}\le\frac{|g|}{2\omega},\qquad
\|\hat{H}_\omega\|_{1,t}\le\frac{1}{2}\sqrt{\delta^2+4g^2}\,t,\qquad
\|\overline{H}\|_{1,t}=\frac{1}{2}\sqrt{\delta^2+g^2}\,t,
\end{equation}
the error within the time range $t\in[0,T]$ is estimated by
\begin{equation}
\|U(t)-U_\mathrm{RWA}(t)\|
\le\|\hat{S}_{\omega}\|_{\infty,T}(1+\|\hat{H}_\omega\|_{1,T}+\|\overline{H}\|_{1,T})
\le\frac{|g|}{2\omega}\left(
1+\sqrt{\delta^2+4 g^2}\,T
\right).
\label{eq:errRWA}
\end{equation}
See Fig.~\ref{fig:RWA} for a pictorial representation of the RWA\@.
\begin{figure}
\centering 
\includegraphics[width=.8\textwidth]{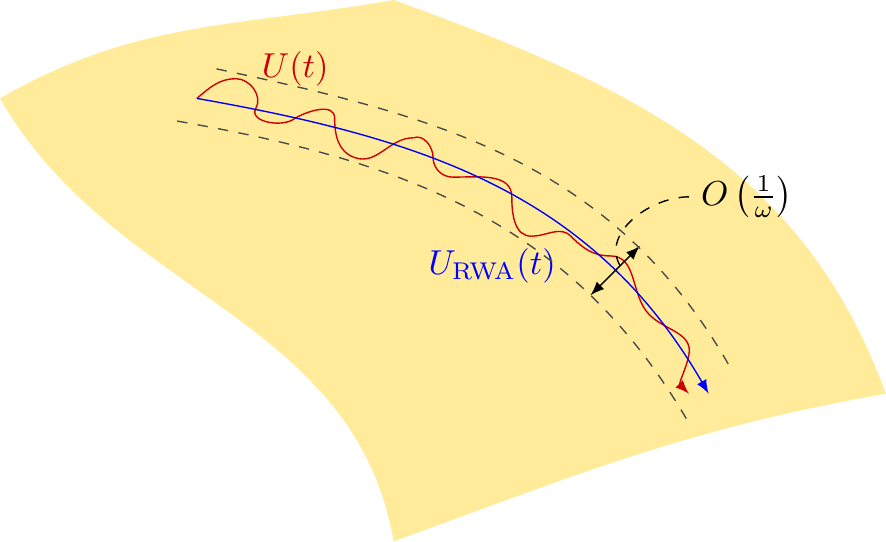}
\caption{Cartoon of the RWA\@. While the true evolution $U(t)$ oscillates quickly with a high frequency $\omega$, the averaged one $U_\mathrm{RWA}(t)$ is simpler and only $O(1/\omega)$ away from the exact dynamics. For the simple but ubiquitous single-qubit case, an explicit nonperturbative bound is provided by~(\ref{eq:errRWA}).} 
\label{fig:RWA}
\end{figure}

\begin{remark}
	The error bound in~\eqref{eq:errRWA} increases with time $T$. In fact, the RWA approximation in general is not uniform, and the error accumulates over time. If one seeks for an eternal approximation, one has to look for the $\omega$-dependent perturbation $\overline{H}_\omega$ of $\overline{H}$, introduced in Theorem~\ref{thm:periodic}, which exactly reproduces the evolution generated by $\hat{H}_\omega(t)$ at every period $t=n\pi/\omega$ with $n\in\mathbb{N}$. According to~\eqref{eqn:MagnusHtau}, we get 
	\begin{align}
		\overline{H}_\omega 
		&=\frac{\rmi\omega}{\pi}\log \hat{U}_\omega\!\left(\frac{\pi}{\omega} \right) 
		\nonumber\\
		&=
		\overline{H}		
		-\frac{\rmi\pi}{2\omega}\int_0^1 \d s\int_0^s \d u\, [\hat{H}_\omega(\pi s/\omega), \hat{H}_\omega(\pi u/\omega)] + O\!\left(\frac{1}{\omega^2} \right) 
		\nonumber\\
		&=\frac{1}{2}\delta Z + \frac{1}{2}gX 
		+ \frac{g}{4\omega} \left( \delta X - \frac{1}{2} g Z \right)	+ O\!\left(\frac{1}{\omega^2} \right),
		\label{eqn:MugnusOmega}
	\end{align}
	for $\|\hat{H}_\omega\|_{1,\pi/\omega}<\pi$. 	
	Since $\|\hat{H}_\omega\|_{1,\pi/\omega}\le \frac{\pi}{2\omega}\sqrt{\delta^2+4g^2}$, 
	we can take, e.g.~$\omega>\frac{\pi}{2}\sqrt{\delta^2+4g^2}$. 
	Because  $\overline{H}$ in~(\ref{eqn:AveHamRWAqubit}) and $\overline{H}_\omega$ in~(\ref{eqn:MugnusOmega}) have different eigenvalues, the evolution generated by $\overline{H}$  eventually diverges from $U(t)$ as discussed in  Remark~\ref{rmk:isospectral}.
\end{remark}
\begin{remark}
If one tries to approximate $H(t)$ in~(\ref{eq:origH}) by $H_\mathrm{cRWA}(t)=\frac{1}{2}\omega_0Z +H_\mathrm{counter}(t)$ instead of $H_\mathrm{RWA}(t)=\frac{1}{2}\omega_0Z +H_\mathrm{co}(t)$ in~(\ref{eq:RWAHam}),  one needs to go to the reference frame rotating with $U_{0,-\omega}(t)$ instead of $U_{0,\omega}(t)$ in~(\ref{eqn:RWA_Reference}).
Then, the counterparts of $\|\hat{H}_\omega\|_{1,t}$ and $\|\overline{H}\|_{1,t}$ in this case become $O(\omega)$ and the error of the approximation remains $O(1)$.
For the same reason, simply taking the time-average of $H(t)$ as $\overline{H}$ fails to approximate the original evolution $U(t)$.
\end{remark}

\subsection{Generalization beyond the Qubit Scenario}\label{sec:genRWA}
The above procedure for the RWA can be easily generalized to systems beyond the qubit.
Let us consider a time-dependent Hamiltonian of the form
\begin{equation}\label{eq:genH}
H(t)=\kappa H_0 +H_1(\kappa t),
\end{equation}
with large $\kappa$.
In the qubit example considered above, $\kappa =\omega$, $H_0= \frac{1}{2} Z$, and $H_1(\tau)= \frac{1}{2}\delta Z + g \cos \tau \, X$. Notice, however, that in general $H_1(\tau)$ is \emph{not} assumed to be periodic.

In the reference frame rotating with the Hamiltonian $\kappa H_0 $, the generator of the evolution is given by
\begin{equation}
\hat{H}_\kappa (t) 
=  \rme^{\rmi \kappa t H_0} H_1(\kappa t)
\rme^{-\rmi \kappa t H_0} ,
\end{equation}
which is a function of $\tau=\kappa t$.
Assuming that the long-time average of $\hat{H}_\kappa (t)$ converges to some limit, namely,
\begin{equation}
\lim_{\tau\to+\infty}\frac{1}{\tau} \int_0^\tau \d s\, \rme^{\rmi s H_0} H_1(s)
\rme^{-\rmi s H_0}
=\overline{H},
\end{equation}
we have that the action vanishes,
\begin{equation}
\hat{S}_\kappa (t) = \int_0^t \d s\, [\hat{H}_\kappa (s) - \overline{H}]
= \frac{1}{\kappa} \int_0^{\kappa t} \d s\,[\rme^{\rmi s H_0} H_1(s) \rme^{-\rmi s H_0} - \overline{H}]\to 0,
\label{eqn:ActionRWAGen}
\end{equation} 
as $\kappa\to+\infty$.
Then, according to~(\ref{eq:divergenceboundbarAe}), the divergence between the evolution $U(t)$ generated by the Hamiltonian $H(t)$ and the unitary $\rme^{-\rmi \kappa tH_0}\rme^{-\rmi t\overline{H}}$ is bounded by
\begin{equation}
\|U(t)-\rme^{-\rmi \kappa tH_0}\rme^{-\rmi t\overline{H}}\|
\le\| \hat{S}_{\kappa}\|_{\infty,T}\left( 
1 +  \frac{1}{\kappa}\|H_1\|_{1,\kappa T}  + \|\overline{H}\|_{1,T}
\right).
\end{equation}
Therefore, if the condition
\begin{equation}
\|\hat{S}_\kappa \|_{\infty,T}
\left(1+\frac{1}{\kappa}\|H_1\|_{1,\kappa T}\right)
\to 0, \quad \text{as}\quad \kappa\to+\infty
\end{equation}
is satisfied, the evolution $U(t)$ is well approximated by $\rme^{-\rmi \kappa tH_0}\rme^{-\rmi t\overline{H}}$,
\begin{equation}
U(t)-\rme^{-\rmi \kappa tH_0}\rme^{-\rmi t\overline{H}}\to0,\quad\text{as}\quad \kappa\to+\infty.
\end{equation}
In particular, if $\|H_1(t)\|\le M$ uniformly for all $t$ for some $M\ge0$, we have
$\|H_1\|_{1,T}\le MT$, $\|\overline{H}\|_{1,T}\le MT$, and hence
\begin{equation}
\|U(t)  - \rme^{-\rmi \kappa t  H_0} \rme^{-\rmi t \overline{H}} \|  
\le  (1 + 2 MT) \| \hat{S}_\kappa \|_{\infty,T} \to 0. 
\label{eqn:BoundRWAGen}
\end{equation}

\subsection{Rotating-Wave Approximation for Systems with Two Driving Timescales}
\label{sec:RWAmodulated}
\subsubsection{Time-Dependent Drive Envelope for a Qubit}
\label{sec:timeRWA}
We can also discuss the RWA under situations where the drive is modulated with a time-dependent envelope~\cite{ref:DiVincenzo}. 
As a simple example, let us look at the evolution generated by the Hamiltonian
\begin{equation}
H(t)  = \frac{1}{2}\omega_0Z + g(t) \cos (\omega t)X.
\label{eqn:RWAmodulated}
\end{equation}
Compared with the above qubit example~(\ref{eq:origH}), the drive here is modulated in time with an envelope function $g(t)$.
In the reference frame rotating with $H_{0,\omega}$ and $U_{0,\omega}(t)$ in~(\ref{eqn:RWA_Reference}), the generator of the evolution is given by
\begin{align}
\hat{H}_\omega(t)
&=U_{0,\omega}(t)^\dag[H(t)   - H_{0,\omega}]U_{0,\omega}(t)
\nonumber\\
&=\frac{1}{2}\delta Z +
\frac{1}{2}g(t)[(1+\cos2\omega t)X  - \sin (2 \omega t) Y],
\end{align}
with $\delta=\omega_0-\omega$.
We are going to show that this evolution with large $\omega$ is well approximated by the evolution generated by the time-dependent Hamiltonian
\begin{equation}
\overline{H}(t)=\frac{1}{2}\delta Z+\frac{1}{2}g(t)X.
\end{equation}
To this end, let us consider the integral action
\begin{equation}
\hat{S}_\omega(t)=\int_0^t \d s\,[\hat{H}_\omega(s)-\overline{H}(s)]=\frac{1}{2}\int_0^t \d s\,g(s)[\cos(2\omega s) X -\sin(2\omega s) Y].
\label{eq:3.30}
\end{equation}
Let us assume that $g(t)$ is of class $C^{n}$. By performing $n$ integration by parts, we get
\begin{equation}
\|\hat{S}_\omega\|_{\infty,t}
\le
\sum_{k=1}^n
\frac{1}{(2\omega)^k}
\|g^{(k-1)}\|_{\infty,t}
+
\frac{1}{2(2\omega)^n}\|g^{(n)}\|_{1,t},
\end{equation}
which is small for large $\omega$.
Since
\begin{equation}
\|\hat{H}_\omega \|_{1,t} \le \frac{1}{2} \sqrt{\delta^2+4\|g\|_{\infty,t}^2} \,t,
\qquad 
\| \overline{H} \|_{1,t} \le \frac{1}{2} \sqrt{\delta^2+\|g\|_{\infty,t}^2} \,t,
\end{equation}
the bound~\eqref{eq:divergenceboundbarAe} gives
\begin{align}\label{eq:BoundTimeDependentDrive}
\| U - U_\mathrm{RWA} \|_{\infty,T}  
&\le
\left(
\sum_{k=1}^n
\frac{1}{(2\omega)^k}
\|g^{(k-1)}\|_{\infty,T}
+
\frac{1}{2(2\omega)^n}\|g^{(n)}\|_{1,T}
\right) \left(
1+\sqrt{\delta^2+4 \|g\|_{\infty,T}^2}\,T
\right)
\nonumber\\
&\to 0,
\end{align}
as $\omega\to+\infty$, where $U(t)$ and $U_\mathrm{RWA}(t)$ are unitaries generated by the Hamiltonian $H(t)$ in~(\ref{eqn:RWAmodulated}) and the Hamiltonian 
\begin{align}
H_\mathrm{RWA}(t)
&=H_{0,\omega} + U_{0,\omega}(t)\overline{H}(t)U_{0,\omega}(t)^\dag
\nonumber\\
&=\frac{1}{2} \omega_0Z +\frac{1}{2} g(t)( \rme^{-\rmi \omega t} \sigma_+ + \rme^{\rmi \omega t} \sigma_-)
\label{eq:RWAHamt}
\end{align}
in the RWA, respectively.

\begin{remark}[Piecewise-constant drive]
	The bound~\eqref{eq:BoundTimeDependentDrive} holds under the assumption that $g(t)$ is $n$ times continuously differentiable. 
	In practical situations, it can happen that the drive envelope $g(t)$ is not differentiable but is piecewise constant as
		\begin{equation}
		g(t)= g_{j}, \quad\text{for}
\quad t\in [t_{j-1}, t_j) \quad(j=1,\ldots,N),
	\end{equation}
where $0=t_0<t_1<\cdots<t_N=T$.		
	Then, for $t_{k}<t<t_{k+1}$, the action~\eqref{eq:3.30} reads
	\begin{align}
		\hat{S}_\omega(t)
		={}&\frac{1}{2}\sum_{j=1}^kg_j\int_{t_{j-1}}^{t_j}\d s\,[\cos(2\omega s) X -\sin(2\omega s) Y]
		+\frac{1}{2}g_{k+1}\int_{t_k}^t\d s\,[\cos(2\omega s) X -\sin(2\omega s) Y]
				\nonumber\\
		={}&{-\frac{g_1}{4\omega} Y} 
		-\frac{1}{4\omega} \sum_{j=1}^{k}(g_{j+1}-g_{j})[\sin(2\omega t_j)X +\cos(2\omega t_j)Y]
		\nonumber\\
		&{}+\frac{g_{k+1}}{4\omega} [\sin(2\omega t)X+\cos(2\omega t)Y].
	\end{align}
	The bound~(\ref{eq:BoundTimeDependentDrive}) in this case is replaced by
	\begin{equation}
		\| U - U_\mathrm{RWA} \|_{\infty,T}  
		\le \frac{1}{2\omega}
		\left(\max_{1\le j\le N}|g_j|+\frac{1}{2}\sum_{j=1}^{N-1} \abs{g_{j+1}-g_j}\right) \left(
		1+\sqrt{\delta^2+4 \max_{1\le j \le N}\abs{g_j}^2}\,T
		\right).
	\end{equation}
\end{remark}

\subsubsection{General System with Two Driving Timescales}
\label{sec:timegenRWA}
The strategy employed above for the qubit driven by a time-dependent drive envelope can be applied to systems beyond qubits.
Consider a Hamiltonian with two driving timescales of the form
\begin{equation}
H(t)=\kappa H_0 (\kappa t)+H_1(t, \kappa t).
\end{equation}
In the rotating frame of $\kappa H_0 (\kappa t)$, one has
\begin{equation}
\hat{H}_\kappa (t)  =  W(\kappa t)^\dag H_1(t, \kappa t)W(\kappa t) = \widetilde{H}(t, \kappa t) ,
\end{equation}
where
\begin{equation}
W(t) = \T\exp\!\left(
-\rmi\int_0^t\d s\, H_0(s)\right).
\end{equation}
The essential idea in the previous section was to average over the fast time-dependence in the drive.
To do this in the integral action, integration by parts was performed.
Suppose that the average over the fast variable has a limit,
\begin{equation}
\overline{H}(t,\tau) = \frac{1}{\tau} \int_0^\tau \d s\, \widetilde{H}(t, s)  
\to \overline{H}(t), \quad\text{as}\quad \tau\to+\infty,
\end{equation}
uniformly for $t\in[0,T]$.
Now, note that
\begin{align}
s\frac{\d}{\d s}\overline{H}(s,\kappa s)
&=s\partial_1\overline{H}(s,\kappa s)+\kappa s\partial_2\overline{H}(s,\kappa s)
\nonumber
\\
&=s\partial_1\overline{H}(s,\kappa s)-\overline{H}(s,\kappa s)+\widetilde{H}(s,\kappa s),
\end{align}
where $\partial_1$ and $\partial_2$ are partial derivatives of $\overline{H}(s_1,s_2)$ with respect to $s_1$ and $s_2$, respectively. Then,
\begin{equation}
\frac{\d}{\d s}[s\overline{H}(s,\kappa s)]
=
s\partial_1\overline{H}(s,\kappa s)
+\widetilde{H}(s,\kappa s),
\end{equation}
which yields the following formula of integration by parts,
\begin{equation}
\int_0^t \d s\,\hat{H}_\kappa (s)
=\int_0^t \d s\,\widetilde{H}(s,\kappa s)
=t\overline{H}(t,\kappa t)-\int_0^t \d s\, s\partial_1\overline{H}(s,\kappa s).
\end{equation}
Using this formula, we estimate the integral action as
\begin{align}
\hat{S}_\kappa (t) 
&=\int_0^t \d s\,[\hat{H}_\kappa (s) - \overline{H}(s)] 
\nonumber\\
&=t[\overline{H}(t,\kappa t) - \overline{H}(t)]
-\int_0^t \d s\,s[\partial_1\overline{H}(s,\kappa s) - \partial_s\overline{H}(s)],
\end{align}
where $\partial_s$ is the derivative with respect to $s$.
Thus, under the additional assumption that
\begin{equation}
\partial_1\overline{H}(t,\tau) 
\to\partial_t\overline{H}(t),\quad\text{as}\quad \tau\to+\infty,
\end{equation}
uniformly for $t\in[0,T]$, the action vanishes $\|\hat{S}_\kappa \|_{\infty,T} \to 0$, as $\kappa\to+\infty$.
Since
\begin{equation}
\|\hat{H}_\kappa (t) \|= \|  W(\kappa t)^\dag H_1(t, \kappa t)W(\kappa t)\| = \|  H_1(t , \kappa t)\|,
\end{equation}
assumption~\eqref{eqn:CondReference2} is satisfied if for instance $\|H_1(t, s)\| \le M$ uniformly for all $t$ and $s$. 
In such a case, Corollary~\ref{cor:OneMoreThmbar} applies, and one has
\begin{equation}\label{eq:GenGenRWA}
U(t) - \T\exp\!\left(
-\rmi\int_0^{\kappa t}\d s\, H_0(s)
\right)
\T\exp\!\left(
-\rmi\int_0^{t}\d s\, \overline{H}(s)
\right)\to0,
\end{equation}
as $\kappa\to+\infty$, uniformly on finite time intervals.

\section{Strong-Coupling Limit and Adiabatic Theorem}
\label{sec:AdiabaticTh}
An adiabatic theorem in quantum mechanics describes the evolution of a quantum system under a slowly driven Hamiltonian~\cite{ref:Messiah}.
This can be formally expressed by considering an evolution operator $U(t)$ satisfying
\begin{equation}
\label{eq:slowEvolution}
\rmi \frac{\d}{\d t} U(t) = H\!\left(\frac{t}{T} \right) U(t), \quad t\in[0,T],
\end{equation}
for large time $T\to+\infty$, with a family of self-adjoint operators $\{H(s)\}_{s\in[0,1]}$. 
Upon rescaling time to $s=t/T$, $s\in[0,1]$ (``macroscopic'' time), and setting $U_T(s)= U(s T)$, the above equation becomes
\begin{equation}\label{eq:slowStrong}
\rmi\frac{\d}{\d s} U_T(s) = TH (s)U_T(s), \quad s\in[0,1],
\end{equation}
that is a strong-coupling limit $T\to+\infty$~\cite{ref:KatoAdiabatic}. This simple observation establishes an interesting link between slow evolutions for long times and evolutions with strong couplings.
In this section, we will show how Corollary~\ref{cor:OneMoreThmbar} is useful to deal with both situations.

\subsection{Strong-Coupling Limit}
\label{sec:strongCoupling}
Let us start with the simplest case of the strong-coupling limit, i.e., the limit $\kappa\to+\infty$ of the evolution generated by a time-independent Hamiltonian $\kappa H_0+ H_1$.
This limit gives rise to the separation of timescales in the evolution, and the transitions between eigenspaces of $H_0$ are suppressed. 
Still, if the dimension of an eigenspace is greater than 1, the system can evolve unitarily within the eigenspace.
This is a manifestation of the quantum Zeno dynamics~\cite{ref:SchulmanJumpTimePRA,ref:QZS,ref:ControlDecoZeno,ref:QZEExp-Ketterle,ref:PaoloSaverio-QZEreview-JPA,ref:ZenoPaoloMarilena,ref:QZEExp-Schafer:2014aa,ref:unity1,ref:GongPRL,ref:GongPRA,ref:EternalAdiabatic,ref:ZenoKAM}, and the evolution within the eigenspaces is generated by a ``Zeno Hamiltonian''~\cite{ref:QZS,ref:PaoloSaverio-QZEreview-JPA}.
This limit can be treated by our method and yields the following theorem.
\begin{thm}[Strong-coupling limit]
\label{thm:StrongCouplingLimit}
Given two bounded self-adjoint operators $H_0$ and $H_1$, assume that $H_0$ has the finite spectral representation
\begin{equation}
H_0=\sum_{\ell=1}^m E_\ell P_\ell,
\end{equation}
where $E_k\neq E_\ell$ for $k\neq\ell$, $P_\ell=P_\ell^\dag$, $P_kP_\ell=\delta_{k\ell}P_\ell$ for all $k$ and $\ell$, and $\sum_{\ell=1}^m P_\ell = 1$. 
Let
\begin{equation}
\eta=\mathop{\min_{k,\ell}}_{k\neq\ell}|E_k-E_\ell|
\end{equation}
be the minimal spectral gap of $H_0$, and
\begin{equation}
H_Z=\sum_{\ell=1}^m P_\ell H_1P_\ell 
\label{eqn:ZenoHcont}
\end{equation}
be the Zeno Hamiltonian. Then,
\begin{equation}
\|\rme^{-\rmi t(\kappa H_0 +H_1)} - \rme^{-\rmi t(\kappa H_0  + H_Z)}\|
\le
\frac{2\sqrt{m}}{\kappa\eta}
\|H_1\|
(1+2T\|H_1\|),
\label{eqn:StronCouplingBound}
\end{equation}
for all $t\in [0,T]$. 
\end{thm}
\begin{proof}
This is a particular case of the scenario discussed in Sec.~\ref{sec:genRWA}.
By using Lemma~\ref{lem:ContinuousErgodic} in Appendix~\ref{app:Ergodic}, the time-average of the Hamiltonian in the rotating frame is shown to converge to the Zeno Hamiltonian,
\begin{equation}
\frac{1}{\tau}\int_0^\tau \d s\, \rme^{\rmi s H_0} H_1
\rme^{-\rmi s H_0} \to H_Z,  \quad\text{as}\quad \tau\to+\infty,
\end{equation}
and the integral action~(\ref{eqn:ActionRWAGen}) is explicitly given by
\begin{equation}
\hat{S}_\kappa (t)
=\frac{1}{\kappa} \int_0^{\kappa t} \d s\,(\rme^{\rmi s H_0} H_1
\rme^{-\rmi s H_0}  - H_Z)
={\sum_{k,\ell}}'
\frac{\rme^{\rmi \kappa t(E_k-E_\ell)}-1}{\rmi\kappa(E_k-E_\ell)}
P_kH_1P_\ell,
\end{equation} 
where $\sum'_{k,\ell}$ represents summation over the pair $(k,\ell)$ excluding terms with $k=\ell$. Using~\eqref{eq:boundSumProj} of Lemma~\ref{lem:ProjectedSum}, this action is bounded by
\begin{equation}
\|\hat{S}_\kappa \|_{\infty,t}
\le
\frac{2\sqrt{m}}{\kappa\eta}
\|H_1\|.
\end{equation}
Noting $\|H_Z\|\le\|H_1\|$, the bound~(\ref{eqn:BoundRWAGen}) applies and reads~(\ref{eqn:StronCouplingBound}).
\end{proof}
The bound~\eqref{eqn:StronCouplingBound} is comparable with the slightly better bound obtained~in Ref.~\cite{ref:ZenoKAM} by a different method. However, this theorem can immediately be  generalized to a time-dependent Hamiltonian $H_1(t)$, a case that cannot be dealt with the methods of Ref.~\cite{ref:ZenoKAM}.
We just get a time-dependent Zeno Hamiltonian $H_Z(t)$ projected exactly in the same way as~(\ref{eqn:ZenoHcont}), with the same error~\eqref{eqn:StronCouplingBound} with $\|H_1\|_{\infty,T}$ in place of $\|H_1\|$.

\subsection{Adiabatic Theorem}
\label{sec:AdiabaticTh1}
The main concern in adiabatic theorems is to characterize how likely it is for a system starting from an eigenspace of a time-dependent Hamiltonian generating the evolution to remain in the same eigenspace during the evolution, provided that the variation of the Hamiltonian is sufficiently slow.  By virtue of Eqs.~\eqref{eq:slowEvolution} and~\eqref{eq:slowStrong}, the limit evolution operator when the Hamiltonian variation is slow is equivalently described by 
\begin{equation}
U_\kappa (t)=\T\exp\!\left(-\rmi\int_0^t\d s\,\kappa H_0 (s)\right),
\end{equation}
with $\kappa\to+\infty$.  In the geometric approach introduced by Kato~\cite{ref:KatoAdiabatic}, the transitions between different eigenspaces of the Hamiltonian can be controlled by bounding the distance between the exact evolution $U_\kappa(t)$ and an ``adiabatic evolution'' which sends each initial eigenspace of $H_0(0)$ to the corresponding eigenspace of $H_0(t)$ at $t\in[0,T]$ with dynamical phases within each eigenspace (see Remark~\ref{rmk:AdThm} below).
The following theorem is proved by using Corollary~\ref{cor:OneMoreThmbar}. 
Since our purpose is to show how our universal bound can be applied  to prove also the adiabatic theorem, we do not intend to use the weakest assumptions possible. 
The readers interested in a broader overview on adiabatic theorems and more refined bounds tailored to specific situations might refer to Ref.~\cite{Lidar}. 
\begin{thm}[Adiabatic theorem]
\label{thm:AdiabaticTheorem}
Assume that for every $t\in[0,T]$ the  self-adjoint operator $H_0(t)$ has the finite spectral representation 
\begin{equation}
H_0(t)=\sum_{\ell=1}^m E_\ell(t)P_\ell(t),
\label{eqn:SpectralRepresentationH0t}
\end{equation}
where $\{E_\ell(t)\}\subset\mathbb{R}$, while $P_\ell(t)=P_\ell(t)^\dag$ and $P_k(t)P_\ell(t)=\delta_{k\ell}P_\ell(t)$ for all $k$ and $\ell$.
Assume  that $E_\ell(t)$ are $C^1$ and $P_\ell(t)$ are $C^2$, and that there are no level crossings, i.e.,
\begin{equation}
|\omega_{k,\ell} (t)| = |E_k(t) - E_\ell (t)| >0, \quad  \text{for}\quad t\in[0,T],
\label{eqn:NoLevelCrossing}
\end{equation}
for $k\neq \ell$.
Then, one has
\begin{align}
&\left\|
U_\kappa (t) - 
\T\exp\!\left(
-\rmi\int_0^t\d s\,[\kappa H_0 (s)+A(s)]
\right)
\right\|
\nonumber\\
&\qquad
\le
\frac{\sqrt{m}}{\kappa\eta}(1+T\|A\|_{\infty,T})
\left[
\left(
2+\frac{\eta'}{\eta}T
\right)\|A\|_{\infty,T}
+ T\|\dot{A}\|_{\infty,T} 
\right],
\label{eqn:AdiabaticTheorem}
\end{align}
for all $t\in [0,T]$, where
\begin{equation}
A(t)= \sum_{\ell=1}^m\frac{\rmi}{2}[\dot{P}_\ell(t),P_\ell(t)]
\label{eqn:AdiabaticTransporterGenerator}
\end{equation}
is the generator of the adiabatic transporter, and
\begin{equation}
\eta = \mathop{\min_{k,\ell}}_{k\neq\ell} \min_{t\in[0,T]} |\omega_{k,\ell}(t)|, \qquad \eta' = \mathop{\max_{k,\ell}}_{k\neq\ell} \max_{t\in[0,T]} |\dot{\omega}_{k,\ell}(t)|
\label{eqn:SpectralGap}
\end{equation}
are the minimal spectral gap and the maximal spectral slope, respectively.
\end{thm}
\begin{remark}
\label{rmk:AdThm}
The second term inside the norm  in~\eqref{eqn:AdiabaticTheorem} can be written as 
\begin{equation}
\T\exp\!\left(
-\rmi\int_0^t\d s\,[\kappa H_0 (s)+A(s)]
\right)=W(t) \rme^{ -\rmi  \kappa t \overline{H}_W(t)},
\end{equation}
where
\begin{equation}
W(t)=\T\exp\!\left(
-\rmi\int_0^t\d s\,A(s)
\right) 
\label{eqn:Transporter}
\end{equation}
is the unitary adiabatic transporter~\cite{ref:KatoAdiabatic}, which transports $P_\ell(0)$ to $P_\ell(t)$ through
\begin{equation}
P_\ell(t)=W(t)P_\ell(0)W(t)^\dag,
\end{equation}
and
\begin{equation}
H_W(t)=W(t)^\dag H_0(t)W(t) = \sum_{\ell=1}^m E_\ell(t)P_\ell(0)
\end{equation}
yields the dynamical phases within each eigenspace.
Note that since $[H_W(t),H_W(t')] = 0$ we have
\begin{equation}
\T\exp\!\left( -\rmi  \kappa \int_0^t \d s\, H_W(s) \right) 
= \exp\!\left( -\rmi  \kappa \int_0^t \d s\, H_W(s) \right) 
= \rme^{-\rmi \kappa t \overline{H}_W(t)},
\end{equation}
where 
\begin{equation}
\overline{H}_W(t)
= \frac{1}{t} \int_0^t \d s\, H_W(s) 
= \sum_{\ell=1}^m  \overline{E}_\ell(t)P_\ell(0).
\label{eq:dynphase}
\end{equation}
\end{remark}
\begin{proof}[Proof of Theorem~\ref{thm:AdiabaticTheorem}]
By going to the adiabatic rotating frame of $W(t)$ and by removing the dynamical phases, we get
\begin{equation}
V_\kappa (t) 
= \rme^{ \rmi  \kappa t \overline{H}_W(t)} W(t)^\dag U_\kappa (t) 
= \T\exp\!\left( -\rmi \int_0^t \d s\, \hat{H}_\kappa (s) \right), 
\label{eqn:AdiabaticRotatingW}
\end{equation}
with
\begin{equation}
\hat{H}_\kappa (t)
= -\rme^{ \rmi  \kappa t \overline{H}_W(t)} W(t)^\dag A(t) W(t) \rme^{ -\rmi  \kappa t \overline{H}_W(t)}
= -{\sum_{k,\ell}}' \rme^{ \rmi  \kappa \sigma_{k,\ell}(t)} A_{k,\ell}(t),
\label{eq:addec1}
\end{equation}
where
\begin{gather}
\sigma_{k,\ell}(t)
=\int_0^t\d s\,\omega_{k,\ell}(s),
\qquad
A_{k,\ell}(t)
=P_k(0) W(t)^\dag A(t) W(t)  P_\ell(0),
\label{eq:addec2}
\end{gather}
and $\sum'_{k,\ell}$ represents the summation over the pair $(k,\ell)$ excluding terms with $k=\ell$.
Notice that $A_{\ell,\ell}(t) = W(t)^\dag P_\ell(t)  A(t) P_\ell (t) W(t) = 0$, and thus the diagonal terms in~\eqref{eq:addec1} do not contribute.

By the assumption~(\ref{eqn:NoLevelCrossing}) that there are no level crossings, $\sigma_{k,\ell}(t)$ is strictly monotonic and invertible, and we get
\begin{align}
S_\kappa (t)
=\int_0^t \d s\, \hat{H}_\kappa (s)
&= -{\sum_{k,\ell}}' \int_0^t \d s\, \rme^{ \rmi  \kappa \sigma_{k,\ell}(s) } A_{k,\ell}(s)
\nonumber
\\
&=-{\sum_{k,\ell}}' 
\int_0^{\sigma_{k,\ell}(t)} \d v\,\rme^{ \rmi  \kappa v } \frac{1}{\omega_{k,\ell}(\sigma_{k,\ell}^{-1}(v))}A_{k,\ell}(\sigma_{k,\ell}^{-1}(v))
\nonumber
\displaybreak[0]
\\
&=-{\sum_{k,\ell}}' 
\int_0^{\sigma_{k,\ell}(t)} \d v\,  \rme^{ \rmi  \kappa v } \widetilde{A}_{k,\ell}(v) \to 0,
\quad \text{as}\quad \kappa\to+\infty,
\end{align}
by the Riemann-Lebesgue lemma~\cite{Strichartz}.
Moreover,
\begin{equation}
\|\hat{H}_\kappa (t)\| 
=\|\rme^{ \rmi  \kappa t \overline{H}_W(t)} W(t)^\dag A(t) W(t) \rme^{ -\rmi  \kappa t \overline{H}_W(t)}\| 
=\|A(t)\| 
\le\|A\|_{\infty,T},
\end{equation}
so that $\hat{H}_\kappa (t)$ is uniformly bounded in $\kappa$, and Corollary~\ref{cor:OneMoreThmbar} applies with $\overline{H}(t)=0$, giving
\begin{equation}
\|U_\kappa (t) - W(t) \rme^{ -\rmi  \kappa t \overline{H}_W(t)}\|
\le\|S_\kappa \|_{\infty,T}(1 +   T\|A\|_{\infty,T}) \to 0 ,
\label{eqn:AdiabaticConvergence}
\end{equation}
for all $t\in[0,T]$.

Let us bound the error.
Since $E_\ell(t)$ is $C^1$ and $P_\ell(t)$ is $C^2$, $A_{k,\ell}(s)$ is differentiable, and by integration by parts we get
\begin{align}
\int_0^t\d s\,\rme^{\rmi \kappa\sigma_{k,\ell}(s)}A_{k,\ell}(s) 
={}&
\frac{1}{\rmi \kappa}
\left(
\frac{
\rme^{ \rmi  \kappa\sigma_{k,\ell}(t)}
}{\omega_{k,\ell}(t)}A_{k,\ell}(t)
-
\frac{1}{\omega_{k,\ell}(0)}A_{k,\ell}(0)
\right)
\nonumber\\
&{}
-
\frac{1}{\rmi \kappa}
\int_0^t\d s\,  
\frac{
\rme^{ \rmi  \kappa\sigma_{k,\ell}(s)}
}{\omega_{k,\ell}(s)}
\left(
\dot{A}_{k,\ell}(s)
-
\frac{
\dot{\omega}_{k,\ell}(s)
}{\omega_{k,\ell}(s)}
A_{k,\ell}(s)
\right).
\end{align}
Note here that $\dot{A}_{k,\ell}(t) = P_k(0) W^\dag(t) \dot{A}(t) W(t)  P_\ell(0)$. 
Then, using Lemma~\ref{lem:ProjectedSum} in Appendix~\ref{app:Ergodic}, the action $S_\kappa (t)$ is bounded by
\begin{equation}
\|S_\kappa \|_{\infty,T} 
\le \frac{\sqrt{m}}{\kappa\eta}\left[
\left(
2
+\frac{\eta'}{\eta} T 
\right)\|A\|_{\infty,T} 
+ T\|\dot{A}\|_{\infty,T} 
\right].
\end{equation}
By noting 
\begin{equation}
W(t) \rme^{ -\rmi  \kappa t \overline{H}_W(t)}
=
\T\exp\!\left(
-\rmi\int_0^t\d s\,[\kappa H_0 (s)+A(s)]
\right),
\end{equation}
the bound in~(\ref{eqn:AdiabaticConvergence}) yields the adiabatic theorem~(\ref{eqn:AdiabaticTheorem}).
\end{proof}

An explicit error bound for an adiabatic approximation, which looks similar to the bound~(\ref{eqn:AdiabaticTheorem}), can be found in Ref.~\cite{ref:JansenRuskaiSeiler}, but it concerns a different scenario.
It bounds the probability of leakage from a collection of eigenspaces to the rest of the spectrum of the Hamiltonian, while Theorem~\ref{thm:AdiabaticTheorem} bounds the distance from the adiabatic evolution confined within every eigenspace. The two bounds are not directly comparable.

\subsection{Generalized Adiabatic Theorem}
\label{sec:AdiabaticZeno}
The adiabatic theorem proved in Theorem~\ref{thm:AdiabaticTheorem} in the previous section can be generalized to the case of a more general time-dependent Hamiltonian $H_\kappa(t)$ which assumes the form $\kappa H_0(t)$ only asymptotically for large $\kappa$. The possibility of such a generalization was mentioned by Kato in Ref.~\cite{ref:KatoAdiabatic} but not worked out explicitly.
\begin{thm}[Generalized adiabatic theorem]
\label{thm:genadiabatic}
Let $t\in[0,T]\mapsto H_\kappa (t)$ be an integrable Hamiltonian, with  $H_\kappa (t)$ self-adjoint and bounded for all $t\in[0,T]$ and all $\kappa>0$. 
Let 
\begin{equation}
U_\kappa (t)=\T\exp\!\left(-\rmi\int_0^t\d s\, H_\kappa (s)\right)
\end{equation}
be the evolution generated by $H_\kappa(t)$.
Assume that for all $t\in[0,T]$ there exists the limit
\begin{equation}
H_0(t) = \lim_{\kappa\to+\infty} \frac{1}{\kappa} H_\kappa(t),
\end{equation}
where $H_0(t)$ has the properties of the Hamiltonian in Theorem~\ref{thm:AdiabaticTheorem}. Assume that 
\begin{equation}
G_\kappa (t) =  H_\kappa(t) - \kappa H_0(t)
\end{equation}
is differentiable and  $\|G_\kappa\|_{\infty,T},\|\dot{G}_\kappa\|_{\infty,T} = o (\sqrt{\kappa})$.
Then, we have
\begin{equation}
U_\kappa (t) 
- 
\T\exp\!\left(
-\rmi\int_0^t\d s\,[\kappa H_0 (s)+G_{\kappa,Z}(s)+A(s)]
\right)\to 0,
\end{equation}
as $\kappa\to+\infty$, uniformly for $t\in[0,T]$, where
$A(t)$ is the adiabatic connection defined in~(\ref{eqn:AdiabaticTransporterGenerator}), and
\begin{equation}
G_{\kappa,Z}(t)=\sum_{\ell=1}^m P_\ell(t)G_\kappa(t)P_\ell(t)
\end{equation}
is the time-dependent Zeno Hamiltonian.
The convergence error is bounded by
\begin{align}
&\left\|
U_\kappa (t) 
- 
\T\exp\!\left(
-\rmi\int_0^t\d s\,[\kappa H_0 (s)+G_{\kappa,Z}(s)+A(s)]
\right)
\right\|
\nonumber
\\
&\quad
\le
\frac{\sqrt{m}}{\kappa\eta}
(
1
+T
\|A\|_{\infty,T}
+2T\|G_\kappa\|_{\infty,T}
)
\nonumber
\\
&\qquad
{}\times
\left[
\left(
2+  \frac{\eta'}{\eta} T
\right)
(\|A\|_{\infty,T}+\|G_\kappa\|_{\infty,T})
+ T
 (
 \|\dot{A}\|_{\infty,T} 
 + \|\dot{G}_\kappa\|_{\infty,T} 
 + 2\|A\|_{\infty,T}\|G_\kappa\|_{\infty,T}
)
\right],
\label{eqn:GeneralizedAdiabaticThm}
\end{align}
where $\eta$ and $\eta'$ are the minimal spectral gap and the maximal spectral slope, respectively, defined in~(\ref{eqn:SpectralGap}).
\end{thm}
\begin{proof}
We prove this in the same way as in the proof of Theorem~\ref{thm:AdiabaticTheorem}.
We go to the adiabatic rotating frame of $W(t)$ and remove the dynamical phases,
\begin{equation}
V_\kappa (t) 
= \rme^{ \rmi  \kappa t \overline{H}_W(t)} W(t)^\dag U_\kappa (t) 
= \T\exp\!\left( -\rmi \int_0^t \d s\, \hat{H}_\kappa (s) \right).
\label{eqn:AdiabaticRotatingW2}
\end{equation}
In the present case, the Hamiltonian $\hat{H}_\kappa (t)$, at variance with that in~\eqref{eq:addec1}, reads
\begin{equation}
\hat{H}_\kappa (t)
= -\rme^{ \rmi  \kappa t \overline{H}_W(t)} W(t)^\dag[A(t)-G_\kappa(t)] W(t) \rme^{ -\rmi  \kappa t \overline{H}_W(t)}
= -\sum_{k, \ell}\rme^{ \rmi  \kappa \sigma_{k,\ell}(t)} A_{k,\ell}(t),
\end{equation}
where $W(t)$, $\overline{H}_W(t)$, and $\sigma_{k,\ell}(t)$ are the same as those in~(\ref{eqn:Transporter}),~(\ref{eq:dynphase}), and~(\ref{eq:addec2}), respectively, while $A_{k,\ell}(t)$ in~(\ref{eq:addec2}) is replaced by
\begin{equation}
A_{k,\ell}(t)
=P_k(0) W(t)^\dag[A(t)-G_\kappa(t)]W(t)  P_\ell(0).
\end{equation}
Now, we consider the integral action
\begin{equation}
S_\kappa (t)
=\int_0^t\d s\,[\hat{H}_\kappa (s)-W(s)^\dag G_{\kappa,Z}(s)W(s)]
=-{\sum_{k,\ell}}'\int_0^t \d s\, \rme^{ \rmi  \kappa \sigma_{k,\ell}(s) } A_{k,\ell}(s).
\end{equation}
We can bound it in the same way as done in Theorem~\ref{thm:AdiabaticTheorem}, but now we have
$\dot{A}_{k,\ell}(t) = P_k(0) W^\dag(t)\{\dot{A}(t)-\dot{G}_\kappa(t)-\rmi[A(t),G_\kappa(t)]\}W(t)  P_\ell(0)$. 
By noting that
\begin{align}
W(t) \rme^{ -\rmi  \kappa t \overline{H}_W(t)}
&\T\exp\!\left(
-\rmi\int_0^t\d s\,W(s)^\dag G_{\kappa,Z}(s) W(s)\right)
\nonumber
\\
={}&\T\exp\!\left(
-\rmi\int_0^t\d s\,[\kappa H_0 (s)+G_{\kappa,Z}(s)+A(s)]
\right),
\end{align}
and $\|G_{\kappa,Z}(t)\|\leq\|G_\kappa(t)\|$, we get the bound~(\ref{eqn:GeneralizedAdiabaticThm}).
\end{proof}
\begin{remark}[Strong-coupling limit with time-dependent Hamiltonians]
\label{rmk:timedepstrong}
An immediate application of Theorem~\ref{thm:genadiabatic} is to the case where a time-dependent perturbation $H_1(t)$ is added to the strong driving $\kappa H_0(t)$ considered in Theorem~\ref{thm:AdiabaticTheorem},
\begin{equation}
U_\kappa (t)=\T\exp\!\left(-\rmi\int_0^t\d s\,[\kappa H_0 (s)+H_1(s)]\right).
\label{eq:4.42}
\end{equation}
This is also regarded as a generalization of the strong-coupling limit proved in Theorem~\ref{thm:StrongCouplingLimit} in Sec.~\ref{sec:strongCoupling} to the case where the Hamiltonians are time-dependent.
In this case, $G_\kappa(t)=H_1(t)$, and one gets
\begin{align}
&\left\|
U_\kappa (t) 
- 
\T\exp\!\left(
-\rmi\int_0^t\d s\,[\kappa H_0 (s)+H_Z(s)+A(s)]
\right)
\right\|
\nonumber
\\
&\quad
\le
\frac{\sqrt{m}}{\kappa\eta}
(
1
+T
\|A\|_{\infty,T}
+2T\|H_1\|_{\infty,T}
)
\nonumber\\
&\qquad
{}\times
\left[
\left(
2+  \frac{\eta'}{\eta} T
\right)
(\|A\|_{\infty,T}+\|H_1\|_{\infty,T})
+ T
 (
 \|\dot{A}\|_{\infty,T} 
 + \|\dot{H}_1\|_{\infty,T} 
 + 2\|A\|_{\infty,T}\|H_1\|_{\infty,T}
)
\right],
\label{eqn:GeneralizedAdiabaticThmH1}
\end{align}
with a time-dependent Zeno Hamiltonian 
\begin{equation}
H_Z(t)=\sum_{\ell=1}^m P_\ell(t)H_1(t)P_\ell(t).
\end{equation}
\end{remark}
The example~\eqref{eq:4.42} represents a strong-coupling implementation of the quantum Zeno dynamics of a time-dependent Hamiltonian, which has potential applications in holonomic quantum computation~\cite{ref:Zeno_unchained}.

\subsection{Comments}
\begin{itemize}
\item
In Secs.~\ref{sec:periodic},~\ref{qubitexample}, and~\ref{sec:genRWA}, we dealt with Hamiltonians of the type
\begin{equation}
H_\kappa (t) = H(\kappa t),
\end{equation}
in particular to discuss the RWA in Sec.~\ref{sec:RWA}.
The dynamics is approximated by an average Hamiltonian
\begin{equation}
\overline{H} 
= \lim_{\kappa\to+ \infty} 
\frac{1}{t}\int_0^t \d s\,  H_\kappa (s) 
= \lim_{\kappa\to+ \infty} 
\frac{1}{t}\int_0^t \d s\,  H(\kappa s) 
= \lim_{\tau \to+ \infty} \frac{1}{\tau} \int_0^\tau \d s\,  H(s),
\end{equation}
which is independent of time.
In the situation where the drive envelope is modulated in time, the system is driven with two different timescales,
\begin{equation}
H_\kappa (t) = H(t, \kappa t).
\label{eq:2times}
\end{equation}
We analyzed such situations in Sec.~\ref{sec:RWAmodulated}, and the effective evolution is generated by a time-dependent Hamiltonian $\overline{H}(t)$ obtained by averaging $H(t, \kappa t)$ over the fast variable,
\begin{equation}
\overline{H}(t) = \lim_{\tau\to+\infty} \frac{1}{\tau} \int_0^\tau \d s\, H(t,s).
\end{equation}

\item
A further generalization is provided in terms of two non-isochronous timescales,
\begin{equation}
H_\kappa (t) = H(t, \kappa \sigma(t) ), 
\end{equation}
with $\dot{\sigma}(t)>0$ for $t\in[0,T]$, so that $\sigma(t)$ is strictly increasing and invertible.
This case, however, is essentially equivalent to the previous one up to a time-reparametrization $s=\sigma(t)$.
Indeed, the Schr\"odinger equation
\begin{equation}
\frac{\d U_\kappa (t)}{\d t} = -\rmi H(t, \kappa \sigma(t) ) U_\kappa (t)
\end{equation}
is equivalent to
\begin{equation}
\frac{\d\widetilde{U}_\kappa (s)}{\d s} = -\rmi \widetilde{H}(s, \kappa s ) \widetilde{U}_\kappa (s),
\end{equation}
with $s=\sigma(t)$ and
\begin{equation}
\widetilde{H}(u, v) = \frac{1}{\dot{\sigma}(\sigma^{-1}(u))} H(\sigma^{-1}(u) , v ).
\end{equation}

\item
The adiabatic theorems proved in Secs.~\ref{sec:AdiabaticTh1} and~\ref{sec:AdiabaticZeno} were reduced to analyzing Hamiltonians of the form
\begin{equation}
H_\kappa (t) = \sum_\ell H_\ell (t, \kappa \sigma_\ell(t) ).
\end{equation}
This cannot be turned into the form~\eqref{eq:2times} by a single time-reparametrization.

\item
The effectiveness of Corollary~\ref{cor:OneMoreThmbar} shows that the introduction of a particular form of $H_\kappa (t)= H(t,\kappa t)$ in terms of a two-variable Hamiltonian $H(u,v)$ (and its average with respect to the fast variable $v$) is something of a red herring.
No particular structure of $H_\kappa (t)$ is in fact required. What is really needed is only that the limit
\begin{equation}
\int_0^t \d s\, H_\kappa (s) \to \int_0^t \d s\, \overline{H} (s), \quad \text{as}\quad \kappa \to+\infty
\end{equation}
exists uniformly in time, as well as a growth condition on $\|H_\kappa \|_{1,T}$ for large $\kappa$.
\end{itemize}

\section{Product Formulas}
\label{sec:productFormulas}
The applications considered so far involve continuous Hamiltonians, but as shown in Appendix~\ref{app:locInt}, our main tool still works for Hamiltonians which are not continuous, as long as they are locally integrable. In this section, we show how this is relevant in proving various product formulas, where the evolution stems from the alternation of noncommuting Hamiltonians.

\subsection{Ergodic-Mean Trotter Formula}
The standard Trotter product formulas require repetitive operations. For instance, for a sequence of $p$ bounded Hamiltonians $(H_1,\ldots, H_p)$, one has~\cite{Suzuki}
\begin{equation}
\label{eq:TPF}
\Bigl(\rme^{-\rmi\frac{t}{n}H_p}\rme^{-\rmi\frac{t}{n}H_{p-1}}\cdots\,\rme^{-\rmi\frac{t}{n}H_1}\Bigr)^n\to \rme^{-\rmi t \sum_{j=1}^p H_{j}},\quad \text{as}\quad n\to+\infty.
\end{equation}
For generic infinite sequences $(H_1, H_2, \ldots{})$, however, it is easy to construct counterexamples to such convergence (see the following Remark ~\ref{rmk:noTrot}). 
Here, we show that, with an additional assumption of existence, an ergodic-mean Trotter formula can be established (see also Ref.~\cite{ref:Bernad}).
This has a variety of applications and many known results can be reproduced on the basis of Theorem~\ref{thm:ATrott}, as we will see in the following subsections (some of the results are even improved and/or new).
\begin{thm}[Ergodic-mean Trotter formula]
\label{thm:ATrott}
Given a sequence of bounded self-adjoint operators $(H_n)_{n\geq 1}$ with $\|H_n\|\le M$, consider the unitary product
\begin{equation}
W_n(t) =
\rme^{-\rmi\frac{t}{n}H_n}\rme^{-\rmi\frac{t}{n}H_{n-1}}\cdots \,\rme^{-\rmi\frac{t}{n}H_1}.
\label{eq:meanTrot}
\end{equation}
Suppose that 
\begin{equation}
\overline{H}_n = \frac{1}{n} \sum_{j=1}^n H_j \to \overline{H}, \quad \text{as}\quad n\to+\infty,
	\label{eq:limkicks1}
\end{equation}
for some bounded self-adjoint operator $\overline{H}$.  
Then,
\begin{equation}
W_n(t)  \to  \rme^{-\rmi t \overline{H}},\quad \text{as}\quad n\to+\infty,
\end{equation}
uniformly for $t$ in compact intervals.
The convergence error is bounded by
\begin{equation}
\| W_n(t) - \rme^{-\rmi t \overline{H}}\| \leq \left(M_n+\frac{2 M}{n}\right)t (1+2tM) ,
\label{eq:Trotbound}
\end{equation}
where
$M_n = \max_{1\leq j\leq n}\{\frac{j}{n}\|\overline{H}_j-\overline{H}\|\}$.
In particular, if one assumes that
\begin{equation}
\|\overline{H}_n-\overline{H}\| \leq \frac{\theta M}{n^\alpha},
\label{eq:convratemean1}
\end{equation}
for some $\theta \ge 0$ and $\alpha>0$ and for all integer $n \geq 1$, then
\begin{equation}
\| W_n(t) - \rme^{-\rmi t \overline{H}}\| \leq \left(\frac{\theta}{n^{\alpha_1}}+\frac{2}{n}\right)tM(1+2tM) ,
\label{eq:thetabound}
\end{equation}
where $\alpha_1 = \min\{\alpha,1\}$.
\end{thm}
\begin{proof}
The proof is an application of Corollary~\ref{cor:OneMoreThmbar} to the unitary evolution
\begin{equation}
W_n(t) = \T\exp\!\left(-\rmi\int_0^t \d s\, h_n(s)\right)
\end{equation}
generated by the piecewise-constant Hamiltonian
\begin{equation}
h_n(s)
= H_j ,\quad\text{for}
\quad s\in \left[(j-1)\frac{t}{n},\,j\frac{t}{n}\right)
\quad(j=1,\ldots,n).
\end{equation}
We have
\begin{equation}
\|h_n\|_{1,t}\leq tM,\qquad
\|\overline{H}\|_{1,t}\le tM.
\end{equation}
For any $s\in[0,t]$, let $s = (m - r)t/n$, with $m=\lceil ns/t\rceil$ and $r=\lceil ns/t\rceil- ns/t$, where $\lceil x\rceil\in\mathbb{N}$ represents the least integer greater than or equal to $x$. Then,
\begin{align}
S_n(s)
&=\int_0^s\rmd u\,[h_n(u)- \overline{H}]
\nonumber
\\
&=
\frac{t}{n}\sum_{j=1}^{m}
H_j
- \frac{rt}{n} H_m
-s \overline{H}
\nonumber\\
&=\frac{mt}{n}\overline{H}_m - s \overline{H} - \frac{rt}{n} H_m
\vphantom{\sum_{j=1}^{m}}
\nonumber\\
&=s(\overline{H}_m - \overline{H}) +   \frac{rt}{n} (  \overline{H}_m - H_m).
\label{eqn:MeanErgodicAction}
\end{align}
Therefore,
\begin{equation}
\| S_n(s)\| \leq s \bigl\|\overline{H}_{\lceil n s/t \rceil} - \overline{H}\bigr\| + \frac{2t}{n} M ,
\label{eqn:MeanErgodicActionBound}
\end{equation}
whence
\begin{align}
\| S_n\|_{\infty,t} &= \sup_{s \in[0,t]}\| S_n(s)\|
\nonumber\\
&\leq 	t \sup_{\sigma \in[0,1]} \left\{\sigma\bigl\|\overline{H}_{\lceil n\sigma \rceil} - \overline{H}\bigr\|\right\} + \frac{2t}{n} M 
\nonumber\\
&= t \max_{1\leq j\leq n}\left\{\frac{j}{n}\|\overline{H}_j-\overline{H}\| \right\} + \frac{2t}{n} M,
\end{align}
which implies that
\begin{equation}
\| S_n\|_{\infty,t} \to 0, \quad \text{as}\quad n\to+\infty.
\end{equation}
Thus, by Corollary~\ref{cor:OneMoreThmbar}, we  get
\begin{equation}
\| W_n(t) - \rme^{-\rmi t \overline{H}}\| \leq \| S_n\|_{\infty,t} (1+ 2 t M) \to 0 ,
\label{eqn:MeanErgodicConvergence}
\end{equation}
as $n\to+\infty$, with a rate bounded by~\eqref{eq:Trotbound}.

In particular, under the assumption~\eqref{eq:convratemean1}, one gets
\begin{equation}
	M_n = \max_{1\leq j\leq n}\left\{\frac{j}{n}\|\overline{H}_j-\overline{H}\| \right\}
	\leq \max_{1\leq j\leq n}\left\{\frac{j}{n}\frac{\theta M}{j^\alpha} \right\},
\end{equation}
that is 
\begin{equation}
	M_n \leq  \frac{\theta M}{n} \max_{1\leq j\leq n} j^{1-\alpha} = 	\frac{\theta M}{n^{\min\{\alpha,1\}}}.	
	\end{equation}
By plugging it into~\eqref{eq:Trotbound}, one gets
the bound~\eqref{eq:thetabound}.
\end{proof}
See Fig.~\ref{fig:ErgodicMeanTrotter} for a pictorial representation giving the intuition behind the ergodic-mean Trotter formula.
\begin{figure}
\centering
\includegraphics[width=0.9\textwidth]{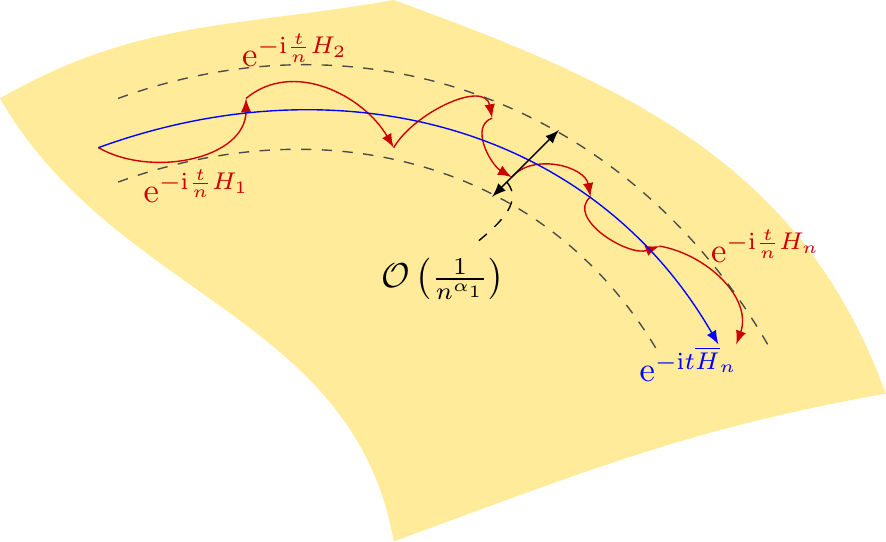}
\caption{Cartoon of product formula. While the true evolution generated by $n$ noncommuting unitaries oscillates quickly, the averaged one is simpler and only $O(1/n^{\alpha_1})$ away from the exact dynamics. For the simple but ubiquitous periodic case, an explicit nonperturbative bound is provided by~(\ref{eq:TPFrate}). }
\label{fig:ErgodicMeanTrotter}
\end{figure}

\begin{remark}[Counterexample]
\label{rmk:noTrot}
	Notice that if the sequence $(H_1, H_2,\ldots)$ does not have an ergodic mean, the product formula $W_n(t)$ in~\eqref{eq:meanTrot} may not converge. 
	Indeed, suppose that there are two subsequences of $(H_n)$, say $(H_{n_k})_{k\geq 1}$ and $(H_{n'_k})_{k\geq 1}$, with different ergodic means, namely,
	\begin{equation}
		\overline{H}_{n_k}\to \overline{H}_1 \quad \text{and}\quad \overline{H}_{n'_k}\to \overline{H}_2, \quad \text{as}\quad k\to+\infty,
	\end{equation}
with $\overline{H}_1\neq \overline{H}_2$.
Then, by Theorem~\ref{thm:ATrott} we get that
\begin{equation}
W_{n_k}(t)  \to  \rme^{-\rmi t \overline{H}_1} \quad \text{and}\quad
W_{n'_k}(t)  \to  \rme^{-\rmi t \overline{H}_2},
\quad \text{as}\quad k\to+\infty,
\end{equation}
so that the product $W_n(t)$ does not converge.

A simple example is given by  $H_n = X$ if $\lfloor \log_{10}(n) \rfloor$ is even, and $H_n = Y$ if $\lfloor \log_{10}(n) \rfloor$ is odd, so that the sequence alternates between $X$ and $Y$ with a slower and slower switching frequency, exponentially decreasing.  
It is easy to see that $\overline{H}_{10^{2k}-1}= \frac{1}{11} X + \frac{10}{11} Y$ for all $k\geq 1$, while $\overline{H}_{10^{2k+1}} \to\frac{1}{11} Y + \frac{10}{11} X$ as $k\to+\infty$, and thus the Trotter product~\eqref{eq:meanTrot} does not converge. 
\end{remark}

\begin{remark}[Trotter product formula]
When the sequence of Hamiltonians $H_j$ is periodic,
\begin{equation}
H_{j+p} = H_j, \quad \text{for all $j \geq 1$,} 
\end{equation}
for some $p\geq 2$, then one gets back to the standard Trotter formula~\eqref{eq:TPF}. 
Set 
\begin{equation}
M= \max_{1\leq j\leq p} \|H_j\|.
\end{equation}
In this case, the limit~\eqref{eq:limkicks1} exists, and
\begin{equation}
\overline{H} = \overline{H}_p = \frac{1}{p} \sum_{j=1}^p H_j.
\end{equation}
In more detail, by noting
\begin{equation}
\sum_{j=1}^nH_j
=
\left\lfloor\frac{n}{p}\right\rfloor\sum_{j=1}^pH_j
+
\sum_{j=1}^{\{\frac{n}{p}\}p}H_j,
\end{equation}
one has
\begin{align}
\overline{H}_n-\overline{H}
&=\left\lfloor\frac{n}{p}\right\rfloor\frac{p}{n}\overline{H}
+\frac{1}{n}\sum_{j=1}^{\{\frac{n}{p}\}p}H_j
-\overline{H}
\nonumber
\\
&=-\left\{\frac{n}{p}\right\}\frac{p}{n}\overline{H}
+\frac{1}{n}\sum_{j=1}^{\{\frac{n}{p}\}p}H_j,
\end{align}
and hence,
\begin{equation}
\|\overline{H}_n - \overline{H}\|
\leq2\left\{\frac{n}{p}\right\} \frac{p}{n} M
	\leq \frac{2p M}{n}, 
\end{equation}
which has the form~\eqref{eq:convratemean1} with $\theta=2 p$ and $\alpha=1$.
Therefore, Theorem~\ref{thm:ATrott} applies, and the Trotter formula~\eqref{eq:TPF} holds with a  rate
\begin{equation}
\Bigl\| 
\Bigl(
\rme^{-\rmi\frac{t}{np}H_p}\rme^{-\rmi\frac{t}{np}H_{p-1}}\cdots \,\rme^{-\rmi\frac{t}{np}H_1}
\Bigr)^n - \rme^{-\rmi t \overline{H}} \Bigr\| \leq \frac{2}{n}\left(1+\frac{1}{p}\right)  t M (1+ 2 t M).
\label{eq:TPFrate}
\end{equation}
This shows the same scaling as the bound derived in Ref.~\cite{Suzuki}.	
This bound can be generalized and improved considering commutator scalings as done in Ref.~\cite{Childs21}.
\end{remark}

\begin{remark}[Eternal Trotterization]
Let us look at the behavior of the Trotter product formula for large times.
Consider, for simplicity, the case $p=2$ and fix the value $\tau>0$ of the time step, so that the total time is $t= n \tau$.
The bound~\eqref{eq:TPFrate} reads
\begin{equation}
	\Bigl\|\Bigl(\rme^{-\rmi \frac{
	\tau}{2}H_2} \rme^{-\rmi \frac{
	\tau}{2}H_1}\Bigr)^n - \rme^{-\rmi n \tau \overline{H}}\Bigr\| \leq 3 \tau M (1+2 n \tau M),
\end{equation}	
with $\overline{H} = \frac{1}{2} (H_1+ H_2)$, that is
\begin{equation}
	\Bigl(\rme^{-\rmi \frac{
	\tau}{2}H_2} \rme^{-\rmi \frac{
	\tau}{2}H_1}\Bigr)^n - \rme^{-\rmi n \tau \overline{H}}= O(\tau M), \quad \text{for}\quad n = O\!\left(\frac{1}{\tau M}\right).
	\label{eq:nOrd}
\end{equation}	
Therefore, for small $\tau$, the Trotter formula is well approximated by the unitary group generated by the average Hamiltonian $\overline{H}$ up to a number $n$ of periods of $O(1/\tau)$, that is up to a total time $t=n\tau$ of $O(1)$.

By using Theorem~\ref{thm:periodic}, we can find a generator of a unitary group which approximates the product formula \emph{uniformly} in $n$. The Trotter formula is obtained at $t=n\tau$ by the propagator
\begin{equation}
	U_\tau(t) = \T\exp\!\left(-\rmi \int_0^t \d s\, h(s)\right),
\end{equation}
generated by the $\tau$-periodic Hamiltonian 
\begin{equation}
	h(t)= \begin{cases}
		H_1, & t\in [0,\frac{\tau}{2}),\\
		H_2, & t\in [\frac{\tau}{2},\tau),
	\end{cases}
	\qquad h(t+\tau)=h(t) .
\end{equation}
The equation~\eqref{eq:coincidetau} reads $\rme^{-\rmi \tau \overline{H}_\tau} = U_\tau(\tau)$, that is
\begin{equation}
	\rme^{-\rmi \tau \overline{H}_\tau} = \rme^{-\rmi \frac{
	\tau}{2}H_2} \rme^{-\rmi \frac{
	\tau}{2}H_1},
\end{equation}
whose solution, for sufficiently small $\tau$, is
\begin{equation}
	\overline{H}_\tau = \frac{\rmi}{\tau} \log\!\left(\rme^{-\rmi \frac{
	\tau}{2}H_2} \rme^{-\rmi \frac{
	\tau}{2}H_1}\right) = \frac{1}{2} (H_1+H_2) + \rmi \frac{\tau}{8} [H_1,H_2] + O(\tau^2 M^3).
\end{equation}	
 The bound~\eqref{eq:divergenceboundpertau} reads
\begin{equation}
	\|U_\tau(t) - \rme^{-\rmi t \overline{H}_\tau}\| \leq \theta \| h\|_{1,\tau} (1+\theta\|h\|_{1,\tau}),
\end{equation}	
for all $t\geq0$, and since
\begin{equation}
	\|h\|_{1,\tau}=\int_0^\tau \d s\,\|h(s)\|= \frac{\tau}{2}(\|H_1\|+\|H_2\|) \leq \tau M,
\end{equation}
one has
\begin{equation}
	\|U_\tau(t) - \rme^{-\rmi t \overline{H}_\tau}\| \leq \theta \tau M (1+\theta \tau M),
\end{equation}	
for all $t\geq0$. 
By setting $t=n\tau$, we finally get the sought eternal bound
\begin{equation}
	\Bigl\|\Bigl(\rme^{-\rmi \frac{
	\tau}{2}H_2} \rme^{-\rmi \frac{
	\tau}{2}H_1}\Bigr)^n - \rme^{-\rmi n \tau \overline{H}_\tau}\Bigr\| \leq \theta \tau M (1+\theta \tau M),
	\quad \forall n\in\mathbb{N}.
\end{equation}

If instead one uses the Hamiltonian
\begin{equation}
	\widetilde{H}_\tau = \frac{1}{2} (H_1+H_2) + \rmi \frac{\tau}{8} [H_1,H_2],
\end{equation}
then its distance from the eternal generator is
\begin{equation}
	\widetilde{H}_\tau = \overline{H}_\tau + O(\tau^2 M^3),
\end{equation}
and the approximation is no more eternal. However, according to Remark~\ref{rmk:largetime}, one gets
\begin{equation}
	\Bigl(\rme^{-\rmi \frac{
	\tau}{2}H_2} \rme^{-\rmi \frac{
	\tau}{2}H_1}\Bigr)^n - \rme^{-\rmi n \tau \widetilde{H}_\tau}= O(\tau M), \quad \text{for}\quad n = O\!\left(\frac{1}{\tau^2 M^2}\right),
\end{equation}	
which is valid for a  number $n$ of periods larger than~\eqref{eq:nOrd}, that is up to a total time $t=n\tau$ of $O(1/\tau)$.
\end{remark}

\subsection{Random Trotter Formula}
\label{sec:RandomTrotter}
The crucial ingredient in the ergodic-mean Trotter formula is the convergence~\eqref{eq:limkicks1} of the arithmetic mean to a limit.
If the sequence of Hamiltonians $H_1, H_2, \ldots$ is drawn at random from a given distribution with finite variance $\sigma_H^2$, by the law of large numbers one gets
that~\eqref{eq:limkicks1} holds in probability:
\begin{equation}
\overline{H}_n 
\stackrel{P}{\longrightarrow} \overline{H}, \quad \text{as}\quad n\to+\infty,
	\label{eq:limkicks1prob}
\end{equation}
where $\overline{H}=\E[H_j]$ is the mean of the distribution, that is
\begin{equation}
	\lim_{n\to+\infty}\Prob\bigl(\| \overline{H}_n - \overline{H}\| > \varepsilon\bigr)= 0, \quad  \text{for all } \varepsilon >0.
\end{equation}
As a consequence, in such a situation one will get 
	\begin{equation}
W_n(t)  \stackrel{P}{\longrightarrow}  \rme^{-\rmi t \overline{H}},\quad \text{as}\quad n\to+\infty,
\end{equation} 
and the theorem will hold in probability. This is the content of the following corollary, which is a probabilistic version of Theorem~\ref{thm:ATrott}, with $\theta = \sigma_H /M$ and $\alpha=1/2$.

Similar stochastic convergence theorems have been known in the mathematical-physics literature.
See for instance Ref.~\cite{ref:RandomTrotterKurtz1972}. 
However, because they are framed in a much more general setting, there are no explicit bounds on the convergence speed provided. 
On the other hand, in the context of quantum information, there has been renewed interest in randomized Trotter evolutions~\cite{ref:RandomCampbellPRL2019,ref:RandomChildsQuantum2019,ref:RandomCampbellQuantum2020,ref:KuengRandomTrotter}. 
In particular, Ref.~\cite{ref:RandomCampbellPRL2019} provides an elegant way to obtain bounds on the average evolution, and Ref.~\cite{ref:KuengRandomTrotter} computes an error bound for the expected difference of the individual realizations to the limit. 
In our framework, we obtain a similar bound (though it is not optimal) with a different proof strategy.

\begin{corol}[Random Trotter formula]
\label{cor:RandomTrott}
Suppose that each operator of a sequence of bounded self-adjoint operators $(H_n)_{n\geq 1}$ is sampled independently and randomly from an identical distribution (i.i.d.\ sampling), and consider the unitary product
\begin{equation}
W_n(t) =
\rme^{-\rmi\frac{t}{n}H_n}\rme^{-\rmi\frac{t}{n}H_{n-1}}\cdots \,\rme^{-\rmi\frac{t}{n}H_1}.
\label{eq:RandomTrot}
\end{equation}
Assume that $\|H_j\|\le M$ for all $j\geq 1$, and the distribution has a finite mean $\overline{H}=\E[H_j]$ and variance $\sigma_H^2=\E\bigl[\|H_j-\overline{H}\|_2^2\bigr]$, where $\|A\|_2=\sqrt{\tr(A^2)}$ is the Hilbert-Schmidt norm.
Then, for any $\varepsilon>0$,
\begin{equation}
\Prob\bigl(\| \overline{H}_n - \overline{H}\| > \varepsilon\bigr) \leq \frac{1}{\varepsilon}\frac{\sigma_H}{n^{1/2}},
\label{eq:weaklaw}
\end{equation}
and
\begin{equation}
\Prob\Bigl(\| W_n(t) - \rme^{-\rmi t \overline{H}}\|>\varepsilon\Bigr)
\le \frac{1}{\varepsilon} \left(\frac{\sigma_H}{M} \frac{1}{n^{1/2}}+\frac{2}{n}\right)tM(1+2tM),
\label{eq:randombound}
\end{equation}
for $t\ge0$.
\end{corol}
\begin{proof}
The first bound~\eqref{eq:weaklaw} is the weak law of large numbers for i.i.d.\ random operators. Indeed, we get
\begin{align}
\Bigl(
\E\bigl[\|\overline{H}_n-\overline{H}\|\bigr]
\Bigr)^2
&\le
\E\bigl[\|\overline{H}_n-\overline{H}\|^2\bigr]
\vphantom{\frac{1}{n^2}}
\nonumber\\
&\le
\E\bigl[\|\overline{H}_n-\overline{H}\|_2^2\bigr]
\vphantom{\frac{1}{n^2}}
\nonumber\\
&=
\E\bigl[\tr\{(\overline{H}_n-\overline{H})^2\}\bigr]
\vphantom{\frac{1}{n^2}}
\nonumber\\
&=
\frac{1}{n^2}
\sum_{i=1}^n \sum_{j=1}^n\E\bigl[\tr\{(H_i-\overline{H}) (H_j-\overline{H})\}\bigr]
\nonumber
\displaybreak[0]
\\
&=
\frac{1}{n^2}
\sum_{j=1}^n\E\bigl[\tr\{(H_j-\overline{H})^2\}\bigr]
\nonumber
\displaybreak[0]
\\
&=
\frac{1}{n}
\E\bigl[\|H_j-\overline{H}\|_2^2\bigr]=
\frac{1}{n}
\sigma_H^2,
\end{align}
that is
\begin{equation}
\E\bigl[\|\overline{H}_n-\overline{H}\|\bigr] \leq \frac{\sigma_H}{n^{1/2}},
\label{eqn:variance}
\end{equation}
where the i.i.d.\ property
\begin{equation}
\E\bigl[(H_i-\overline{H})(H_j-\overline{H})\bigr]=\delta_{ij}\E\bigl[(H_j-\overline{H})^2\bigr]
\end{equation}
has been used. By Markov's inequality, for any $\varepsilon>0$ we get
\begin{equation}
\Prob\bigl(\| \overline{H}_n - \overline{H}\| > \varepsilon\bigr) \leq 
\frac{1}{\varepsilon} \E\bigl[\|\overline{H}_n-\overline{H}\|\bigr] \leq
\frac{1}{\varepsilon}\frac{\sigma_H}{n^{1/2}},	
\end{equation} 
that is inequality~\eqref{eq:weaklaw}.

Now, we proceed as in Theorem~\ref{thm:ATrott} and  bound the action $S_n(s)$ by~\eqref{eqn:MeanErgodicActionBound}, i.e.,
\begin{equation}
\|S_n(s)\|
\leq s\bigl\|\overline{H}_{\lceil ns/t \rceil} - \overline{H}\bigr\| + \frac{2t}{n} M ,
\end{equation}
for $s\in[0,t]$. By taking the expectation value 
and using~\eqref{eqn:variance}, we get
\begin{align}
\E\bigl[
\|S_n(s)\|
\bigr]
&\le
s\E\bigl[\bigl\|\overline{H}_{\lceil ns/t \rceil} - \overline{H}\bigr\|\bigr] + \frac{2t}{n} M
\nonumber\\
&\le\frac{s}{\lceil ns/t\rceil^{1/2}}\sigma_H + \frac{2t}{n} M
\nonumber\\
&\le\frac{s^{1/2}t^{1/2}}{n^{1/2}}\sigma_H + \frac{2t}{n} M
\nonumber\\
&\le\frac{t}{n^{1/2}}\sigma_H + \frac{2t}{n} M.
\end{align}
By the intermediate bound~\eqref{eqn:IntermediateBound} in Lemma~\ref{lemma:divergence}, we have
\begin{equation}
\|W_n(t) - \rme^{-\rmi t \overline{H}}\|
\le\|S_n(t)\| + 2M\int_0^t \d s\,\|S_n(s)\|,
\label{eqn:IntermediateBoundRandom}
\end{equation}
whence, by taking the expectation value,
\begin{align}
\E\Bigl[
\|W_n(t) - \rme^{-\rmi t \overline{H}}\|
\Bigr]
&\le\E\bigl[\|S_n(t)\|\bigr] + 2M\int_0^t \d s\E\bigl[\|S_n(s)\|\bigr]
\nonumber\\
&\leq \left(\frac{\sigma_H}{n^{1/2}}+\frac{2 M}{n}\right)t (1+2tM).
\end{align}
The bound~\eqref{eq:randombound} follows by Markov's inequality.
\end{proof}

\subsection{Frequent Unitary Kicks}
Consider the evolution of a system generated by a time-independent Hamiltonian $H$ interrupted by a sequence of $n+1$ unitary kicks $U_1, \ldots, U_{n+1}$,
\begin{equation}
	 W_n(t) =
	 U_{n+1}\rme^{-\rmi\frac{t}{n}H} U_n\rme^{-\rmi\frac{t}{n}H}\cdots\,U_2\rme^{-\rmi\frac{t}{n}H} U_1 .
\label{eq:multikicks}
\end{equation}
The evolution lasts for a time $t$ and there is a kick every time $t/n$. We are interested in the evolution for large $n$.

We can rewrite~\eqref{eq:multikicks} in a more convenient form
\begin{align}
W_n(t)
&=V_{n+1} (V_n^\dag\rme^{-\rmi\frac{t}{n}H} V_n)\cdots (V_2^\dag\rme^{-\rmi\frac{t}{n}H}V_2) 	 (V_1^\dag\rme^{-\rmi\frac{t}{n}H}V_1)
\nonumber\\
&=V_{n+1}\rme^{-\rmi\frac{t}{n}V_n^\dag H V_n} \cdots\,\rme^{-\rmi\frac{t}{n}V_2^\dag H V_2} \rme^{-\rmi\frac{t}{n}V_1^\dag H V_1} ,
\end{align}
where $V_j=U_j\cdots U_1$ ($j=1,\ldots,n+1$), which, up to the last factor $V_{n+1}$, reduces to an isospectral case of the product formula~\eqref{eq:meanTrot}. 
Therefore, for large $n$ the evolution  can be approximated by the unitary group generated by the average Hamiltonian, namely,
\begin{equation}
W_n(t)\approx V_{n+1} \rme^{-\rmi t \overline{H}}.
\end{equation}
\begin{corol}[Frequent unitary kicks]
\label{thm:Ukicks}
Given a bounded self-adjoint operator $H$ and a sequence of unitaries $(U_n)_{n\geq 1}$, consider the unitary product
\begin{equation}
W_n(t) =
U_{n+1}\rme^{-\rmi\frac{t}{n}H}U_n\rme^{-\rmi\frac{t}{n}H}\cdots \,U_2\rme^{-\rmi\frac{t}{n}H}U_1 .
\end{equation}
Let $V_n = U_n \cdots U_1$ for all integer $n\geq 1$, and suppose that 
\begin{equation}
\overline{H}_n = \frac{1}{n} \sum_{j=1}^n V_j^\dag H V_j \to \overline{H}, \quad \text{as}\quad n\to+\infty,
\label{eq:limkicks}
\end{equation}
for some bounded self-adjoint operator $\overline{H}$.  Then,
\begin{equation}
W_n(t)  - V_{n+1} \rme^{-\rmi t \overline{H}}\to 0,\quad \text{as}\quad n\to+\infty ,
\end{equation}
uniformly for $t$ in compact intervals.
Moreover, if one assumes that
\begin{equation}
\|\overline{H}_n-\overline{H}\| \leq \frac{\theta}{n^\alpha} \|H\|,
	\label{eq:convratemean}
\end{equation}
for some $\theta \ge 0$ and $\alpha>0$ and for all integer $n \geq 1$, then
\begin{equation}
\| W_n(t) - V_{n+1} \rme^{-\rmi t \overline{H}}\| \leq \left(\frac{\theta}{n^{\alpha_1}}+\frac{2}{n}\right) t \|H\| (1+ 2 t \|H\|) ,
\label{eq:thetabound2}
\end{equation}
with $\alpha_1 = \min\{\alpha,1\}$.
\end{corol}
\begin{proof}
Choosing $H_n=V_n^\dagger H V_n$ in Theorem~\ref{thm:ATrott}, we have $M=\|H\|$, and the statement follows immediately.
\end{proof}

\subsection{Dynamical Decoupling}
Let us assume that the sequence of unitaries $V_j$ is periodic,
\begin{equation}
V_{j+p} = V_{j}, \quad \text{for all $j\geq 1$,}
\end{equation}
for some $p\ge2$. This happens if
\begin{equation}
U_{p+1}= 
U_2^\dag U_3^\dag\cdots U_p^\dag,
\end{equation}
and
\begin{equation}
U_{j+p} = U_j, \quad \text{for all $j\geq 2$.}
\end{equation}
We say that $(V_1,\ldots,V_p)$ is a cycle of kicks. In such a case, the limit~\eqref{eq:limkicks} exists, \begin{equation}
\overline{H} = \overline{H}_p = \frac{1}{p} \sum_{j=1}^p V_j^\dag H V_j,
\end{equation}
and one has
\begin{equation}
\|\overline{H}_n - \overline{H}\|\leq 2\left\{\frac{n}{p}\right\} \frac{p}{n} \|H\|
	\leq \frac{2p}{n}\|H\|,
\end{equation}
which has the form~\eqref{eq:convratemean} with $\theta=2 p$ and $\alpha=1$.
Therefore, Corollary~\ref{thm:Ukicks} applies, and the evolution subjected to $n$ cycles of $p$ kicks converges to the evolution $\rme^{-\rmi t \overline{H}}$ as $n$ increases, with a rate
\begin{equation}
\Bigl\| \Bigl( V_p^\dag\rme^{-\rmi\frac{t}{n p}H} V_p\cdots 	V_1^\dag\rme^{-\rmi\frac{t}{n p}H}V_1 \Bigr)^n -  \rme^{-\rmi t \overline{H}}\Bigr\| \leq \frac{2}{n}\left(1+\frac{1}{p}\right)t\|H\|(1+ 2 t \|H\|).
\label{eq:DD}
\end{equation}
This bound is exponentially better in $t$ than a previous bound~\cite{ref:unity2}.

The method of dynamical decoupling~\cite{ref:DynamicalDecoupling-ViolaLloydPRA1998,ref:DynamicalDecoupling-ViolaKnillLloyd-PRL1999,ref:BangBang-VitaliTombesi-PRA1999,ref:DynamicalDecoupling-Zanardi-PLA1999,ref:BangBang-Duan-PLA1999,ref:DynamicalDecoupling-Viola-PRA2002,ref:BangBang-Uchiyama-PRA2002} fits in this formalism and has important applications in quantum control. Consider for instance a $d$-dimensional quantum system coupled to an environment that induces decoherence on the quantum system. The Hilbert space of the total system is a product $\mathcal{H}_S\otimes\mathcal{H}_E$, and the system-environment coupling is represented by the sum of interaction Hamiltonians of the form $h \otimes H_E$. The aim of the dynamical decoupling is to suppress such detrimental interactions by rapidly rotating the system around different axes with a cycle of unitary kicks $V_j=v_j\otimes1$ on the system with the property 
\begin{equation}
\frac{1}{p}\sum_{j=1}^p v_j^\dag h v_j = \frac{1}{d}\tr h
\label{eqn:GroupAverage}
\end{equation} 
for every operator $h$ on $\mathcal{H}_S$. Such cycles of unitaries average out the unwanted interactions ($\overline{h \otimes H_E}=0$ if $h$ is traceless, otherwise it yields just a Hamiltonian of the environment), and the evolutions of the system and of the environment are decoupled.

Finally, it is worth noticing that random dynamical decoupling\cite{ref:RandomDD-Viola2005,ref:RandomDD-Santos2006,ref:RandomDD-Viola2006,ref:RandomDD-Santos2008,ref:RandomDD-HillierArenzBurgarth2015,ref:RandomDD-Alex2022}, in which the unitaries $V_j$ are randomly chosen (without periodicity), can also be treated by our method, as an application of the random Trotter formula proved in Corollary~\ref{cor:RandomTrott}.

\subsection{Bang-Bang Control}\label{sec:BangBang}
The dynamical decoupling mentioned in the previous section intends to decouple the evolutions of systems interacting with each other, by applying a variety of unitaries to average out the interactions through~(\ref{eqn:GroupAverage}).
It can be regarded as a manifestation of the quantum Zeno dynamics~\cite{ref:BBZeno,ref:unity2}: the Hamiltonian of the coupled systems is projected by the group average~(\ref{eqn:GroupAverage}), and the systems evolve within subspaces (Zeno subspaces) according the projected Hamiltonian (Zeno Hamiltonian)~\cite{ref:QZS,ref:PaoloSaverio-QZEreview-JPA}.
To induce the quantum Zeno dynamics, it is actually enough to repeat a \emph{single} fixed unitary kick~\cite{ref:BBZeno}, instead of applying several different unitaries to cover a complete unitary group.
This situation is a special case of Corollary~\ref{thm:Ukicks}, and an explicit error bound is available (see also Ref.~\cite{ref:Bernad}).
The limit evolution is generated by a Zeno Hamiltonian $H_Z$ defined in the same way as the one~(\ref{eqn:ZenoHcont}) induced by the strong-coupling limit proved in Theorem~\ref{thm:StrongCouplingLimit} in Sec.~\ref{sec:strongCoupling}~\cite{ref:BBZeno,ref:ControlDecoZeno,ref:PaoloSaverio-QZEreview-JPA,ref:QZEExp-Schafer:2014aa,ref:unity2,ref:ProductFormulaTaylor}.
\begin{corol}[Fixed kick]
\label{cor:UnitaryKick}
Let $H$ be bounded and self-adjoint, and let $U$ be a unitary operator with the finite spectral representation
\begin{equation}
U=\sum_{\ell=1}^m\rme^{-\rmi\phi_\ell}P_\ell,
\end{equation}
where $\{\rme^{-\rmi\phi_\ell}\}$ are the $m$ distinct eigenvalues of $U$.
Let 
\begin{equation}
\eta = \mathop{\min_{k,\ell}}_{k\neq\ell}| \rme^{-\rmi\phi_k } - \rme^{-\rmi\phi_\ell }|= 2 \mathop{\min_{k,\ell}}_{k\neq\ell}\left|\sin\frac{\phi_k-\phi_\ell}{2}\right|
\end{equation}
be the minimal spectral gap of $U$, and \begin{equation}
H_Z= \sum_{\ell=1}^m
P_\ell HP_\ell
\end{equation}
be the Zeno Hamiltonian. Then,
\begin{equation}
\Bigl\|
\Bigl(U\rme^{-\rmi\frac{t}{n}H}\Bigr)^n
-
U^n \rme^{-\rmi tH_Z}
\Bigr\|
\le
\frac{2}{n}
\left(
\frac{
\sqrt{m}
}{\eta}+1 
\right) 
t
\|H\|
(
1
+
2t\|H\|
),
\label{eqn:BoundUnitaryKick}
\end{equation}
for all integer $n\geq 1$ and for all $t\geq0$.
\end{corol}
\begin{proof}
This is an application of Corollary~\ref{thm:Ukicks}
with $U_1=1$ and $U_j=U$ for $j=2,3,\ldots,n+1$, and thus $V_1=1$ and $V_j=U^{j-1}$ for $j=2,3,\ldots,n+1$. By Lemma~\ref{lem:DiscreteErgodic} in Appendix~\ref{app:Ergodic}, we get the ergodic mean
\begin{equation}
\overline{H}_n=\frac{1}{n}\sum_{j=0}^{n-1}
{U^\dag}^j H U^j \to\sum_\ell P_\ell  H P_\ell,\quad \text{as}\quad n\to+\infty, 
\end{equation}
that is $\overline{H} = H_Z$. Moreover,
\begin{equation}
\| \overline{H}_n - H_Z \| \leq \frac{2\sqrt{m}}{\eta n}
\|H\|,
\end{equation}
which has the form~\eqref{eq:convratemean} with $\theta=2 \sqrt{m}/\eta$ and $\alpha=1$,
and~\eqref{eq:thetabound2} gives the thesis.
\end{proof}
The bound~(\ref{eqn:BoundUnitaryKick}) is tighter by a factor $\log n$ than the bound derived in Ref.~\cite{ref:ProductFormulaTaylor}.

\subsection{A Generalized Trotter Formula}
The strong-coupling limit and the bang-bang control, analyzed in Secs.~\ref{sec:strongCoupling} and~\ref{sec:BangBang}, respectively, are both manifestations of the quantum Zeno dynamics. The connection between these two quite different situations can be made more explicit by writing the unitary kicks in the bang-bang evolution in the form $U=\rme^{-\rmi tH_0}$ (here, $t$ is some fixed time). Then, the strong-coupling limit~\eqref{eqn:StronCouplingBound} and the bang-bang limit~\eqref{eqn:BoundUnitaryKick} respectively represent the following approximations,
\begin{align}
\rme^{-\rmi t(\kappa H_0 +H_1)} &\approx \rme^{-\rmi t(\kappa H_0  + H_Z)},  \quad \text{as}\quad \kappa \to +\infty,
\\
\left(\rme^{-\rmi t H_0 }\rme^{-\rmi\frac{t}{n}H_1}\right)^n
&\approx
\rme^{-\rmi t(n  H_0 +H_Z)},  \quad \text{as}\quad n\to +\infty.
\end{align}
The right-hand sides of the these equations coincide by identifying $\kappa= n$, which implies that the left-hand sides must be approximately equal by this identification. In other words, we are interested in the validity of the following Trotter-like formula,
\begin{equation}
\left(\rme^{-\rmi \frac{t}{n}\kappa H_0}\rme^{-\rmi \frac{t}{n}H_1}\right)^n\approx\rme^{-\rmi t(\kappa H_0+H_1)}, \quad \text{as}\quad n\to +\infty,
\end{equation}
where the coupling strength $\kappa$ is allowed to grow with the number $n$ of Trotter steps. 
This relation was discussed in Ref.~\cite{ref:alphaTrot}, where it was found that the error can be bounded, but the bound grows with $\kappa$. Here, we improve this result using Corollary~\ref{cor:OneMoreThmbar}.
\begin{thm}[Generalized Trotter formula]
\label{GTF}
Let $H_0$ be a bounded self-adjoint operator with the finite spectral representation 
\begin{equation}
H_0=\sum_{\mu=1}^m \lambda_\mu P_\mu,
\label{eqn:SpectralRepHGTF}
\end{equation}  
and denote with
\begin{equation}
\eta = \mathop{\min_{\mu,\nu}}_{\mu\neq \nu} |\lambda_\mu-\lambda_\nu|
\end{equation}
the minimal spectral gap of $H_0$. Assume that 
\begin{equation}
\kappa \le   \frac{\theta}{\eta t} n, 
\label{eq:Kbound}
\end{equation}
for some $0<\theta < 2\pi$. Then, 
\begin{equation}
\label{eq:epsilon(n)0}
\left(\rme^{-\rmi \frac{t}{n}\kappa H_0}\rme^{-\rmi \frac{t}{n}H}\right)^n-\rme^{-\rmi t(\kappa H_0+H_1)}\rightarrow 0,
\end{equation}
as $n\to+\infty$. In particular, one has
\begin{equation}
\label{eq:epsilon(n)}
\left\|\left(
\rme^{-\rmi \frac{t}{n}\kappa H_0}\rme^{-\rmi \frac{t}{n}H}
\right)^n
-\rme^{-\rmi t(\kappa H_0+H_1)}\right\|
\le \frac{1}{n}[\sqrt{m}\,g(\theta)+2] t\|H_1\|(1+2t\|H_1\|),
\end{equation}
where 
$g(x)= 2 |1+\rmi x -\rme^{\rmi x}|/|(\rme^{\rmi x} -1)x|\leq 2/(2\pi-x) + (1-1/\pi)$. 
\end{thm}
\begin{remark}
Note that the product formula~(\ref{eq:epsilon(n)0}) holds for $\kappa =o(n)$, i.e.~for a coupling strength $\kappa$ growing sublinearly in the number of Trotter steps $n$, which was conjectured in Ref.~\cite{ref:alphaTrot}. Here, by using the universal bound, we show that the conjecture is in fact true up to $\kappa = O(n)$, as in~\eqref{eq:Kbound}.
\end{remark}
\begin{proof}[Proof of Theorem~\ref{GTF}]
As in the proof of the strong-coupling limit in Theorem~\ref{thm:StrongCouplingLimit}, it is convenient to look at the distance between the two evolutions in the frame rotating with $\kappa H_0$,
\begin{equation}\label{eq:rotatingGenProd}
\left(\rme^{-\rmi\frac{t}{n}\kappa H_0}\rme^{-\rmi\frac{t}{n}H_1}\right)^n-\rme^{-\rmi t(\kappa H_0+H_1)}=\rme^{-\rmi t \kappa H_0}[\hat{U}_{n,\kappa}(t)-\widetilde{U}_\kappa (t)],
\end{equation}
where
\begin{equation}
\widetilde{U}_\kappa (t)=\rme^{\rmi\kappa tH_0}\rme^{-\rmi t (\kappa H_0+H_1)}
=\T \exp\!\left(-\rmi \int_0^t \d s\,\widetilde{H}_\kappa (s)\right)
\end{equation}
is generated by the Hamiltonian in the rotating frame
\begin{equation}
\widetilde{H}_\kappa (s)=\rme^{\rmi\kappa sH_0} H_1 \rme^{-\rmi\kappa sH_0},
\end{equation}
while
\begin{align}
\hat{U}_{n,\kappa}(t)
&=\rme^{\rmi\kappa tH_0}\left(\rme^{-\rmi\frac{t}{n}\kappa H_0}e^{-\rmi\frac{t}{n}H_1}\right)^n
\nonumber\\
&=\Bigl(
\rme^{\rmi(n-1)\frac{t}{n}\kappa H_0}\rme^{-\rmi \frac{t}{n} H_1}\rme^{-\rmi (n-1)\frac{t}{n}\kappa H_0}
\Bigr)
\cdots
\Bigl(
\rme^{\rmi\frac{t}{n}\kappa H_0}\rme^{-\rmi \frac{t}{n} H_1}\rme^{-\rmi\frac{t}{n}\kappa H_0}
\Bigr)\,
\rme^{-\rmi \frac{t}{n} H_1}
\nonumber\\
&=\rme^{-\rmi \frac{t}{n}\widetilde{H}_\kappa((n-1)t/n)} \cdots\, \rme^{-\rmi \frac{t}{n}\widetilde{H}_\kappa(t/n)} \rme^{-\rmi \frac{t}{n}\widetilde{H}_\kappa(0)}
\nonumber
\displaybreak[0]
\\
&=\T \exp\!\left(-\rmi \int_0^t \d s\,\hat{H}_{n,\kappa}(s)\right),
\end{align}
with $\hat{H}_{n,\kappa}(s)$ being a ``discretized'' version of $\widetilde{H}_\kappa (s)$, i.e.,
\begin{equation}
\hat{H}_{n,\kappa}(s)=\widetilde{H}_\kappa\!\left( \left\lfloor s \frac{n}{t}\right\rfloor\frac{t}{n}\right).
\end{equation}
By applying Lemma~\ref{lemma:divergence}, one has
\begin{equation}
\|\hat{U}_{n,\kappa}(s)-\widetilde{U}_\kappa (s)\|\le \|S_{n,\kappa}\|_{\infty,t}(1+\|\hat{H}_{n,\kappa}\|_{1,t}+\|\widetilde{H}_\kappa \|_{1,t}),
\label{eqn:GTFbound0}
\end{equation}
for all $s\in[0,t]$, where in this case
\begin{equation}
S_{n,\kappa}(t)=\int_0^{t} \d s\,[\hat{H}_{n,\kappa}(s)-\widetilde{H}_\kappa (s)]
\end{equation}
represents essentially the integrated error associated with the ``discretization.''
Since $\|\hat{H}_{n,\kappa}(s)\|=\|\widetilde{H}_\kappa (s)\|=\|H_1\|$, the bound~(\ref{eqn:GTFbound0}) is reduced to
\begin{equation}\label{eq:GenBound1}
\|\hat{U}_{n,\kappa}(t)-\widetilde{U}_\kappa (s)\|\le \|S_{n,\kappa}\|_{\infty,t}  (1+2t\|H_1\|),
\end{equation}
for all $s\in[0,t]$. 
We now need to bound $\|S_{n,\kappa}\|_{\infty,t}$.

Let us first consider the action $S_{n,\kappa}(s)$ at $s=s_\ell=\ell\frac{t}{n}$ for $\ell=1,\dots,n$,
\begin{align}
S_{n,\kappa}(s_\ell)&=\int_0^{s_\ell} \d s\,[\hat{H}_{n,\kappa}(s)-\widetilde{H}_\kappa (s)]
\nonumber
\\
&=\frac{t}{n}\sum_{j=0}^{\ell-1}\widetilde{H}_\kappa \!\left(j\frac{t}{n}\right)-\int_0^{s_\ell} \d s\, \widetilde{H}_\kappa (s)
\nonumber\\
&=\tau_n \left(
\sum_{j=0}^{\ell-1}\widetilde{H}_\kappa (j\tau_n)-\int_0^{\ell} \d u \, \widetilde{H}_\kappa (u\tau_n)
\right),
\label{eq:Sn(tk)}
\end{align}
where $\tau_n=t/n$.
By using the spectral representation of $H_0$ in~(\ref{eqn:SpectralRepHGTF}), we have
\begin{equation}
\widetilde{H}_\kappa (s)
=\sum_{\mu,\nu=1}^m \rme^{\rmi\kappa\omega_{\mu\nu}s} P_\mu H_1 P_\nu,
\end{equation}
with $\omega_{\mu\nu}=\lambda_\mu-\lambda_\nu$. 
Then,
\begin{align}
&\sum_{j=0}^{\ell-1}\widetilde{H}_\kappa (j\tau_n)-\int_0^\ell \d u\, \widetilde{H}_\kappa (u\tau_n)
\nonumber\\
&\qquad
= \sum_{\mu,\nu} \left(\sum_{j=0}^{\ell-1} \rme^{\rmi j\kappa\omega_{\mu\nu}\tau_n}-\int_0^\ell \d u\,\rme^{\rmi u\kappa\omega_{\mu\nu}\tau_n} \right)P_\mu H_1 P_\nu
\nonumber\\
&\qquad
={\sum_{\mu,\nu}}' \left(\frac{\rme^{\rmi \ell\kappa\omega_{\mu\nu}\tau_n}-1}{\rme^{\rmi\kappa\omega_{\mu\nu}\tau_n}-1}-\frac{\rme^{\rmi \ell\kappa\omega_{\mu\nu}\tau_n}-1}{\rmi\kappa\omega_{\mu\nu}\tau_n}\right)P_\mu H_1P_\nu
\nonumber
\displaybreak[0]
\\
&\qquad
={\sum_{\mu, \nu}}'(\rme^{\rmi \ell\kappa\omega_{\mu\nu}\tau_n}-1)
\frac{1+\rmi\kappa\omega_{\mu\nu}\tau_n-\rme^{\rmi\kappa \omega_{\mu\nu}\tau_n}}{(\rme^{\rmi\kappa \omega_{\mu\nu}\tau_n}-1)\rmi\kappa\omega_{\mu\nu}\tau_n}
P_\mu H_1P_\nu.
\end{align}
Thus, by taking the norm of the action~\eqref{eq:Sn(tk)} using Lemma~\ref{lem:ProjectedSum} in Appendix~\ref{app:Ergodic}, we get the bound
\begin{equation}
\|S_{n,\kappa}(s_\ell)\|
\le\frac{t}{n}
\sqrt{m}\mathop{\max_{\mu,\nu}}_{\mu\neq\nu}g_\ell(\kappa\omega_{\mu\nu}\tau_n) \|H_1\|,
\label{eq:sumIntBound0}
\end{equation}
where
\begin{equation}\label{eq:gk}
g_\ell(x)=\left|(\rme^{\rmi \ell x}-1)\frac{1+\rmi x -\rme^{\rmi x}}{(\rme^{\rmi x} -1)x}\right|,
\end{equation}
which is bounded by
\begin{equation}\label{eq:gkBound}
g_\ell(x)\le g(x)
= 2 \left|\frac{1+\rmi x -\rme^{\rmi x}}{(\rme^{\rmi x} -1)x}\right| 
=\left|
\frac{\sinc(x/2)-\rme^{\rmi x/2}}{\sin(x/2)}
\right|.
\end{equation}
Notice that $g(x)$ diverges at points $x_\ell=2\pi\ell$ with $\ell=\pm1,\pm2,\dots$,  but can be suitably bounded far from these points.
In particular, since $g(x)$ is even and increasing in $(0,2\pi)$, by assuming~\eqref{eq:Kbound}, one has
\begin{equation}
0\le\kappa|\omega_{\mu\nu}| \tau_n =\kappa |\lambda_\mu-\lambda_\nu| \frac{t}{n}\le \theta,
\end{equation}
and $g(\kappa\omega_{\mu\nu}\tau_n) \le g(\theta)<+\infty$.
The action~\eqref{eq:sumIntBound0} is thus bounded by 
\begin{equation}
\|S_{n,\kappa}(s_\ell)\|\le \frac{t}{n} \sqrt{m}\, g(\theta) \|H_1\|,
\end{equation}
uniformly for $\ell=0,\dots, n$. The bound on $\|S_{n,\kappa}(s)\|$ for $s\in[0,t]$ can be easily obtained by this result: for each fixed $s$, take $s_\ell=\ell\frac{t}{n}$ for $\ell=\lfloor s\frac{n}{t}\rfloor$, and write
\begin{equation}\label{eq:Sn(t)bound}
\|S_{n,\kappa}(s)\|\le \|S_{n,\kappa}(s_\ell)\|+\|S_{n,\kappa}(s)-S_{n,\kappa}(s_\ell)\|\le\frac{t}{n}[\sqrt{m}\,g(\theta)+2] \|H_1\|,
\end{equation}
where we used
\begin{equation}
\|S_{n,\kappa}(s)-S_{n,\kappa}(s_\ell)\|\le \int_{s_\ell}^s \d \sigma\,\|\hat{H}_{n,\kappa}(\sigma)-\widetilde{H}_\kappa (\sigma)\|\le 2(s-s_\ell)\|H_1\|\le \frac{2t}{n} \|H_1\|.
\end{equation}
Equation~\eqref{eq:Sn(t)bound} yields the bound
\begin{equation}
\|S_{n,\kappa}\|_{\infty,t} \le \frac{t}{n}[\sqrt{m}\,g(\theta)+2 ] \|H_1\|,
\end{equation}
which is inserted into~\eqref{eq:GenBound1} to finally obtain
\begin{equation}
\|\hat{U}_{n,\kappa}(t)-\widetilde{U}_\kappa (t)\|\le \frac{1}{n}[\sqrt{m}\,g(\theta)+2] t\|H_1\|(1+2t\|H_1\|).
\end{equation}
This, used together with~\eqref{eq:rotatingGenProd}, proves~\eqref{eq:epsilon(n)}.
\end{proof}

\section{Conclusions}
\label{sec:Conclusions}
The evolutions generated by time-dependent Hamiltonians are notoriously complicated to characterize. In many cases of interest, one looks for a simple approximation of the exact evolution which captures its essential features. Here, we have derived simple bounds to estimate the errors of such approximations.
Our result allows one to estimate the error associated with these approximations in terms of the integral action of the difference between the approximated and the exact generator of the evolution. Thanks to the complete generality of the approach, this result can be used in a wide plethora of cases which may appear to be unrelated to each other, such as the RWA, the adiabatic theorems, strong-coupling limits, and even pulsed evolutions, which have the additional peculiarity of being generated by time-dependent discontinuous (usually piecewise-constant) Hamiltonians.
An obvious limitation of our work is that all bounds are given in terms of norms of some operators. Therefore, most of our results do not apply to unbounded operators (however, see Remark~\ref{rmk:UnboundedH}). 
In particular, with respect to the RWA, it would be interesting to develop convergence results for commonly used systems such as the Jaynes-Cummings model.

\section*{Acknowledgments}
DB acknowledges discussions with Robin Hillier, Arne Laucht, and Mauro Morales. This research was funded in part by the Australian Research Council (projects FT190100106, DP210101367, CE170100009). 
It was also supported in part by the Top Global University Project from the Ministry of Education, Culture, Sports, Science and Technology (MEXT), Japan.
KY was supported by the Grants-in-Aid for Scientific Research (C) (No.~18K03470) and for Fostering Joint International Research (B) (No.~18KK0073) both from the Japan Society for the Promotion of Science (JSPS). 
GG and PF were partially supported by Istituto Nazionale di Fisica Nucleare (INFN) through the project ``QUANTUM'' and by the Italian National Group of Mathematical Physics (GNFM-INdAM). PF was partially supported by Regione Puglia and by QuantERA ERA-NET Cofund in Quantum Technologies (GA No.~731473), project PACE-IN\@.

\appendix
\section{Locally Integrable Generators}
\label{app:locInt}
Our main tools, Lemma~\ref{lemma:divergence}, Theorem~\ref{thm:OneMoreThm}, and Corollary~\ref{cor:OneMoreThmbar}, based on integration by parts, are all valid for locally integrable generators.
In this appendix, we provide explicit derivations of our basic tools by just assuming the local integrability of the generators.
\begin{lemma}[Gronwall's inequality for $L^1_{\mathrm{loc}}$ functions]
\label{lemma:Gronwall}
Let $a\in L^1_{\mathrm{loc}}(\mathbb{R})$ be  locally integrable, with $a(t) \ge0$ for all $t\ge0$. Assume that $g\in L^\infty_{\mathrm{loc}}(\mathbb{R})$ satisfies
\begin{equation}
0\le g(t) \le \int_0^t \d s\, a(s) g(s) .
\label{eq:Gron}
\end{equation}
Then,
$g(t) \equiv 0$ for all $t\ge0$.
\end{lemma}
\begin{proof}
By iterating,
\begin{align}
0\le g(t) &\le \int_0^t \d s_1\, a(s_1) \int_0^{s_1} \d s_2\, a(s_2) g(s_2)
\nonumber\\
&\le \int_0^t \d s_1
\cdots \int_0^{s_{n-1}} \d s_n \int_0^{s_n} \d s\, a(s_1) \cdots a(s_n) a(s) g(s)
\nonumber
\displaybreak[0]
\\
&= \int_0^t \d s\, a(s) g(s) \int_s^{t} \d s_n \cdots \int_{s_2}^t \d s_1\, a(s_1) \cdots a(s_n) 
\nonumber
\displaybreak[0]
\\
&= \int_0^t \d s\, a(s) g(s) \frac{1}{n!} \left(
\int_{s}^t \d s_1\, a(s_1)
\right)^n 
\nonumber
\displaybreak[0]
\\
& \le \frac{1}{n!} \left( 
\int_0^t \d s_1\, a(s_1) 
\right)^n  \int_0^t \d s\, a(s) g(s) \to 0,
\end{align}
as $n\to+\infty$.
\end{proof}

\begin{prop}[Propagators generated by $L^1_{\mathrm{loc}}$ generators]
Let $t\in\mathbb{R} \mapsto H(t)\in B(\mathcal{H})$ be a locally integrable operator-valued map, $H\in L^1_{\mathrm{loc}}(\mathbb{R})$. 
Then, the integral equation
\begin{equation}
U(t) = 1 - \rmi \int_0^t \d s\, H(s) U(s)
\label{eq:integraleq}
\end{equation}
has a unique bounded solution $t\mapsto U(t)\in B(\mathcal{H})$.
Moreover, $U(t)$ is continuous.
\end{prop}
\begin{proof}
\mbox{}
\begin{enumerate}
\item(Existence)
The series
\begin{equation}
U(t) = 1 + \sum_{k\ge 1} \int_0^t  \d s_1 \cdots \int_0^{s_{k-1}} \d s_k\, H(s_1) \cdots H(s_k)
\end{equation}
is absolutely convergent for every $t\ge 0$. 
Indeed,
\begin{align}
\|U(t)\|
&\le
1 + \sum_{k\ge 1} \int_0^t  \d s_1 \cdots \int_0^{s_{k-1}} \d s_k\, \| H(s_1)\| \cdots \|H(s_k)\|
\nonumber\\
&=\sum_{k\ge 0} \frac{1}{k!} \left( 
\int_0^t  \d s\, \|H(s)\| 
\right)^k
\nonumber\\
&=\exp\!\left(\int_0^t  \d s\, \|H(s)\|\right)
\end{align}
is bounded, and it is immediate to show that $U(t)$ satisfies the integral equation~\eqref{eq:integraleq}.

\item(Uniqueness)
Let $U$ and $V$ be two bounded solutions of the integral equation~(\ref{eq:integraleq}). 
Then,
\begin{equation}
\| U(t) - V(t) \| = \left\| 
\int_0^t \d s\, H(s) [U(s)-V(s)]
\right\|.
\end{equation}
Let $g(t) = \| U(t) - V(t) \|$ and $a(t)= \|H(t)\|$. 
They satisfy Gronwall's inequality~\eqref{eq:Gron}, so that $g(t)= 0$.
Thus, $U=V$.

\item(Continuity)
One has
\begin{equation}
U(t_1) - U(t_2)  = -\rmi \int_{t_1}^{t_2} \d s\, H(s) U(s),
\end{equation}
so that
\begin{equation}
\|U(t_1) - U(t_2)\|   \le \int_{t_1}^{t_2} \d s\, \| H(s)\| \|U(s)\| \to 0,
\end{equation}
as $t_2\to t_1$.
\end{enumerate}
\end{proof}

\begin{lemma}[Integration by parts]
\label{lem:3}
Let $a(t)$, $b(t)$, and $c(t)$ be locally integrable bounded operator-valued functions, and let $A(t) = A(0) +\int_{0}^t \d s\, a(s)$,  $B(t) = B(0) +\int_0^t \d s\, b(s)$, and $C(t) = C(0) +\int_0^t \d s\, c(s)$. 
Then,
\begin{equation}
\int_0^t \d s\,[ a(s) B(s) + A(s) b(s)] = A(t) B(t) - A(0) B(0),
\label{eq:2termint}
\end{equation}
and
\begin{equation}
\int_0^t \d s\, [ a(s) B(s) C(s) + A(s) b(s) C(s) + A(s) B(s) c(s)] = A(t) B(t) C(t) - A(0) B(0) C(0).
\label{eq:3termint}
\end{equation}
\end{lemma}
\begin{proof}
The proof is a direct computation:
\begin{align}
&\int_0^t \d s\, [ a(s) B(s) + A(s) b(s)]
\nonumber\\
&\qquad
= \int_0^t \d s\, a(s) \left( B(0) +\int_0^s \d u\, b(u)  \right)
+ \int_0^t \d u\left( A(0) +\int_0^u \d s\, a(s)  \right) b(u)
\nonumber
\displaybreak[0]
\\
&\qquad
=[A(t)-A(0)] B(0) + \int_0^t \d s\, a(s) \int_0^s \d u\, b(u) 
\nonumber
\\
&\qquad\qquad
{}
+ A(0)[B(t)-B(0)] + \int_0^t \d s\, a(s) \int_s^t \d u\, b(u) 
\nonumber
\displaybreak[0]
\\
&\qquad
=[A(t)-A(0)] B(0) 
+ A(0)[B(t)-B(0)] 
+[A(t)-A(0)][B(t)-B(0)]
\vphantom{\int_0^t}
\nonumber
\displaybreak[0]
\\
&\qquad
= A(t) B(t)- A(0)B(0),
\vphantom{\int_0^t}
\end{align}
that is formula~\eqref{eq:2termint}. 
Formula~\eqref{eq:3termint} is obtained by setting
$D(t) = A(t) B(t)$, so that $D(t)= D(0) + \int_0^t \d s\, d(s)$ with $d(t)=a(t) B(t) + A(t) b(t)$ by~\eqref{eq:2termint}.
Thus,
\begin{equation}
\int_0^t \d s\, [ a(s) B(s) C(s) + A(s) b(s) C(s) + A(s) B(s) c(s)] 
=\int_0^t \d s\, [ d(s)  C(s)  +  D(s) c(s)],
\end{equation}
and applying again~\eqref{eq:2termint}, we are done.
\end{proof}

\begin{thm}
Let $U_j$ be the continuous propagators generated by $H_j \in L^1_{\mathrm{loc}}(\mathbb{R})$, for $j=1,2$.
Then,
\begin{equation}
\label{eq:U1dagU2}
U_1(t)^\dag U_2(t) = 1 -\rmi \int_0^t \d s\, U_1(s)^\dag [ H_2(s) - H_1(s)^\dag ] U_2(s).
\end{equation}
\end{thm}
\begin{proof}
From
\begin{equation}
U_2(t) = 1 -\rmi \int_0^t \d s\, H_2(s) U_2(s), 
\qquad  U_1(t)^\dag = 1 +\rmi \int_0^t \d s\, U_1(s)^\dag H_1(s)^\dag, 
\end{equation}
one gets~\eqref{eq:U1dagU2} by integration by parts~\eqref{eq:2termint}, for
$A(t) =  U_1(t)^\dag$ with $a(t)= \rmi U_1(t)^\dag H_1(t)^\dag$, and 
$B(t) = U_2(t)$ with $b(t) = -\rmi H_2(t) U_2(t)$.
\end{proof}

\begin{corol}
Let $H= H^\dag \in L^1_{\mathrm{loc}}(\mathbb{R})$ be the bounded self-adjoint generator of a continuous propagator $U$. Then, $U(t)$ is unitary for all $t$,
\begin{equation}
\label{eq:U(t)unitary}
U(t)^\dag U(t) =  U(t) U(t)^\dag = 1.
\end{equation}
\end{corol}
\begin{corol}
\label{cor:3}
Let $H_j= H_j^\dag \in L^1_{\mathrm{loc}}(\mathbb{R})$ be the bounded self-adjoint generators of continuous propagators $U_j$, for $j=1,2$. 
Then,
\begin{equation}
\label{eq:U(t)unitary2}
U_2(t) - U_1(t) =  -\rmi \int_0^t \d s\,  U_1(t) U_1(s)^\dag[ H_2(s) - H_1(s)] U_2(s).
\end{equation}
\end{corol}

Now, the next theorem provides an explicit derivation of the integration-by-part lemma (Lemma~\ref{lemma:divergence}) for locally integrable generators.
\begin{thm}[Integration-by-part lemma for locally integrable generators]
\label{thm:5}
Let $U_j$ be the continuous propagators generated by bounded self-adjoint $H_j \in L^1_{\mathrm{loc}}(\mathbb{R})$, for $j=1,2$. 
Let 
\begin{equation}
S_{21}(t) = \int_0^t \d s\, [H_2(s) - H_1(s)].
\end{equation}
Then,
\begin{equation}
U_2(t)-U_1(t) = -\rmi S_{21}(t) U_2(t) - \int_0^t \d s\, U_1(t) U_1(s)^\dag [ H_1(s) S_{21}(s) - S_{21}(s) H_2(s)] U_2(s).
\end{equation}
\end{thm}
\begin{proof}
By Corollary~\ref{cor:3}, $U_2(t) - U_1(t) = U_1(t) D(t)$, with  
\begin{equation}
D(t) = -\rmi \int_0^t \d s\,  U_1(s)^\dag[ H_2(s) - H_1(s)] U_2(s).
\end{equation}
Then, an application of~(\ref{eq:3termint}) of Lemma~\ref{lem:3} to
$A(t) =  U_1(t)^\dag$ with $a(t)= \rmi U_1(t)^\dag H_1(t)$, 
$B(t) = S_{21}(t)$ with $b(t) = H_2(t) -H_1(t)$, and 
$C(t) = U_2(t)$ with $c(t) = -\rmi H_2(t) U_2(t)$
gives
\begin{align}
D(t) =-\rmi U_1(t)^\dag S_{21}(t)U_2(t)
&{}+\rmi \int_0^t \d s\,[\rmi U_1(s)^\dag H_1(s)] S_{21}(s) U_2(s)
\nonumber\\
&{}+\rmi \int_0^t \d s\, U_1(s)^\dag  S_{21}(s) [ -\rmi H_2(s) U_2(s)],
\end{align}
and the theorem is proved for $U_2(t)-U_1(t) = U_1(t) D(t)$.
\end{proof}

\section{Ergodic Means}
\label{app:Ergodic}
In this appendix, we provide a few basic bounds on ergodic means, which are used in the main text. We use the  spectral norm.
\begin{lemma}
\label{lem:ProjectedSum}
Let $\{P_\ell\}$ be a set of $m$ Hermitian projections satisfying $P_\ell=P_\ell^\dag$ and $P_kP_\ell=\delta_{k\ell}P_\ell$ for all $k$ and $\ell$.
Then, for any bounded operators $\{A_\ell\}$, we have
\begin{equation}
\biggl\|
\sum_\ell A_\ell P_\ell 
\biggr\|
\le
\sqrt{
\sum_\ell 
\|
A_\ell 
\|^2
\vphantom{\biggr\|}
}.
\label{eq:B1}
\end{equation}
Moreover,  for any bounded operator $A$ and for any complex numbers $\{c_{k,\ell}\}$, we have
\begin{equation}\label{eq:boundSumProj}
\biggl\|
\sum_{k,\ell}
c_{k,\ell}
P_kAP_\ell
\biggr\|
\le
\sqrt{m}\max_{k,\ell}
|c_{k,\ell}|
\|A\| .
\end{equation}
\end{lemma}
\begin{proof}
The spectral norm $\|A\|$ gives the largest singular value of an operator $A$, and satisfies $\|A\|^2=\|A^\dag A\|=\|AA^\dag\|$.
Thus,
\begin{align}
\biggl\|
\sum_\ell A_\ell P_\ell 
\biggr\|^2
&=
\biggl\|
\biggl(\sum_k A_k P_k\biggr)
\biggl(\sum_\ell A_\ell P_\ell\biggr)^\dag
\biggr\|
=
\biggl\|
\sum_\ell A_\ell P_\ell A_\ell^\dag
\biggr\|
\nonumber
\\
&\le
\sum_\ell 
\|
A_\ell P_\ell A_\ell^\dag
\|
=
\sum_\ell 
\|
A_\ell P_\ell 
\|^2
\le
\sum_\ell 
\|
A_\ell 
\|^2,
\end{align}
that is~\eqref{eq:B1}.

The second inequality~(\ref{eq:boundSumProj}) derives from~\eqref{eq:B1},
\begin{align}
\biggl\|
\sum_{k,\ell}
c_{k,\ell}
P_kAP_\ell
\biggr\|^2
&=
\biggl\|
\sum_\ell
\,\biggl(
\sum_k
c_{k,\ell}
P_kA
\biggr)\,
P_\ell
\biggr\|^2
\le
\sum_\ell
\,\biggl\|
\sum_k
c_{k,\ell}
P_kA
\biggr\|^2
\nonumber\\
&\le
\sum_\ell
\,\biggl\|
\sum_k
c_{k,\ell}
P_k
\biggr\|^2
\|A\|^2
=
\sum_\ell
\max_k
|c_{k,\ell}|^2
\|A\|^2
\le
m\max_{k,\ell}
|c_{k,\ell}|^2
\|A\|^2.
\end{align}
\end{proof}

In the following two lemmas, $\sum'_{k,\ell}$ represents double summations over $k$ and $\ell$ excluding terms with $k=\ell$.
\begin{lemma}[Continuous ergodic mean]
\label{lem:ContinuousErgodic}
Let $H$ be a bounded self-adjoint operator with the finite spectral representation 
\begin{equation}
H=\sum_{\ell=1}^m E_\ell P_\ell ,
\label{eqn:SpectralRepHLemma5}
\end{equation}
where $\{E_\ell\}$ is the spectrum of $H$ and $\{P_\ell\}$ its spectral projections. 
Let 
\begin{equation}
\eta=\mathop{\min_{k,\ell}}_{k\neq\ell}|E_k-E_\ell|
\end{equation}
be the minimal spectral gap of $H$.
Then, for any bounded operator $A$ and for any $t>0$, we have
\begin{equation}
\biggl\|
\frac{1}{t}\int_0^t\d s\,
\rme^{\rmi sH}A\rme^{-\rmi sH}
-\sum_\ell P_\ell AP_\ell 
\biggr\|
\le
\frac{2\sqrt{m}}{\eta t}
\|A\|.
\end{equation}
\end{lemma}
\begin{proof}
By using the spectral representation of $H$ in~(\ref{eqn:SpectralRepHLemma5}), we have
\begin{align}
\frac{1}{t}\int_0^t\d s\,
\rme^{\rmi sH}A\rme^{-\rmi sH}
-\sum_\ell P_\ell AP_\ell 
&=
{\sum_{k,\ell}}'
\frac{1}{t}
\int_0^t\d s\,\rme^{\rmi s(E_k-E_\ell)}P_kAP_\ell
\nonumber\\
&=
{\sum_{k,\ell}}'
\frac{\rme^{\rmi t(E_k-E_\ell)}-1}{\rmi t(E_k-E_\ell)}P_kAP_\ell.
\end{align}
By using Lemma~\ref{lem:ProjectedSum}, this is bounded by
\begin{align}
\biggl\|
\frac{1}{t}\int_0^t\d s\,
\rme^{\rmi sH}A\rme^{-\rmi sH}
-\sum_\ell P_\ell AP_\ell 
\biggr\|
&=
\Biggl\|
{\sum_{k,\ell}}'
\frac{\rme^{\rmi t(E_k-E_\ell)}-1}{\rmi t(E_k-E_\ell)}P_kAP_\ell
\Biggr\|
\nonumber
\displaybreak[0]
\\
&\le
\sqrt{m}\mathop{\max_{k,\ell}}_{k\neq\ell}
\left|
\frac{\sin[t(E_k-E_\ell)/2]}{t(E_k-E_\ell)/2}
\right|\|A\|
\le
\frac{2\sqrt{m}}{\eta t}
\|A\|. 
\end{align}
\end{proof}

\begin{lemma}[Discrete ergodic mean]
\label{lem:DiscreteErgodic} 
Let $U$ be a unitary operator with the finite spectral representation
\begin{equation}
U=\sum_{\ell=1}^m \rme^{-\rmi\phi_\ell }P_\ell, 
\label{eqn:SpectralRepU}
\end{equation}
where $\{\rme^{-\rmi\phi_\ell }\}$ is the spectrum of $U$ and $\{P_\ell \}$ its spectral projections. 
Let
\begin{equation}
\eta = \mathop{\min_{k,\ell}}_{k\neq\ell}| \rme^{-\rmi\phi_k } - \rme^{-\rmi\phi_\ell }|= 2 \mathop{\min_{k,\ell}}_{k\neq\ell}\left|\sin\frac{\phi_k-\phi_\ell}{2}\right|
\end{equation}
be the minimal spectral gap of $U$.
Then, for any bounded operator $A$ and for any integer $n\geq1$, we have
\begin{equation}
\Biggl\|
\frac{1}{n}
\sum_{j=0}^{n-1}{U^\dag}^jAU^j
-\sum_\ell P_\ell AP_\ell 
\Biggr\|
\le
\frac{2\sqrt{m}}{\eta n}\|A\| .
\end{equation}
\end{lemma}
\begin{proof}
By using the spectral representation of $U$ in~(\ref{eqn:SpectralRepU}), we have
\begin{equation}
\frac{1}{n}\sum_{j=0}^{n-1}
{U^\dag}^jAU^j
-\sum_\ell P_\ell AP_\ell 
=
{\sum_{k,\ell}}'
\frac{1}{n}
\sum_{j=0}^{n-1}
\rme^{\rmi j(\phi_k-\phi_\ell)}P_kAP_\ell
=
\frac{1}{n}
{\sum_{k,\ell}}'
\frac{1-\rme^{\rmi n(\phi_k-\phi_\ell)}}{1-\rme^{\rmi(\phi_k-\phi_\ell)}}
P_kAP_\ell. 
\end{equation}
By using Lemma~\ref{lem:ProjectedSum}, this is bounded by
\begin{align}
\Biggl\|
\frac{1}{n}
\sum_{j=0}^{n-1}{U^\dag}^jAU^j
-\sum_\ell P_\ell AP_\ell 
\Biggr\|
&=
\frac{1}{n}
\,\Biggl\|
{\sum_{k,\ell}}'
\frac{1-\rme^{\rmi n(\phi_k-\phi_\ell)}}{1-\rme^{\rmi(\phi_k-\phi_\ell)}}
P_kAP_\ell
\Biggr\|
\nonumber
\displaybreak[0]
\\
&\le
\frac{\sqrt{m}}{n}
\mathop{\max_{k,\ell}}_{k\neq\ell}
\left|
\frac{\sin[n(\phi_k-\phi_\ell)/2]}{\sin[(\phi_k-\phi_\ell)/2]}
\right|
\|A\|
\le
\frac{2 \sqrt{m}}{\eta n}
\|A\|. 
\end{align}
\end{proof}

\section{Rotating Frames}
\label{app:canongauge}
In Remark~\ref{rmk:gauge}, the freedom in the choice of a rotating frame is mentioned.
Given $U_1(t)$ and $U_2(t)$ unitary propagators, the ``canonical'' gauge for a rotating frame is given by one of the two evolutions, say $U_1(t)$.
In this rotating frame, the distance between the two evolutions is given by the distance between the relative evolution $U(t)= U_1(t)^\dag U_2(t)$ and the identity (zero Hamiltonian),
\begin{equation}
	\|U_2(t)-U_1(t)\| = \|U(t)-1\|.
\end{equation}
The Hamiltonian of the relative evolution $U(t)$ is given by
\begin{equation}
	H(t) = U_1(t)^\dag [H_2(t)-H_1(t) ] U_1(t),
\end{equation}
and the integral action is simply
\begin{equation}
	S(t)= \int_0^t \rmd s\, H(s).
\end{equation}
The Hamiltonian $H(t)$ may have an involute form and such canonical frame may be unpractical. However, it has a conceptual merit as we are going to show.

One can show that the converse of~\eqref{eq:actioncontrol} holds in the canonical frame, under suitable assumptions on $H(t)$. Namely,
\begin{equation}
U_2\approx U_1  \quad \Rightarrow \quad S\approx 0.
\end{equation}
\begin{prop}
\label{prop:converse}
Let $t\mapsto U(t)$ be the unitary propagator generated by a locally integrable time-dependent Hamiltonian $H(t)$, with $H(t)$ bounded and self-adjoint for all $t\in \mathbb{R}$.
Then, the integral action
$S(t) = \int_0^t \d s\,H(s)$ is bounded by
\begin{equation}
	\|S\|_{\infty,t} \leq (1+ \|H\|_{1,t})\|U -1\|_{\infty,t}.
\end{equation}
\end{prop}
\begin{proof}
	We have
	\begin{align}
		U(t) -1 &= -\rmi \int_0^t \d s\, H(s) U(s) 
		\nonumber\\
		&= -\rmi \int_0^t \d s\, H(s)[U(s) - 1] -\rmi \int_0^t \d s\, H(s)
		\nonumber\\
		&= -\rmi \int_0^t \d s\, H(s)[U(s) - 1] -\rmi S(t),
	\end{align}
	that is,
	\begin{equation}
		S(t) = \rmi [U(t) - 1 ] - \int_0^t \d s\, H(s) [U(s) - 1 ].
	\end{equation}
	By taking the norm the result follows.	
\end{proof}

\section{Non-isospectral Hamiltonians}
\label{app:isospectral}
Here, we prove that, if the constant generators of two evolutions are not isospectral, the two evolutions eventually diverge, regardless of how close the generators may be.
\begin{prop}
\label{prop:isospectral}
Let $H$ and $G$ be two (not necessarily bounded) self-adjoint operators with pure point spectrum. If $H$ and $G$ are not unitarily equivalent, then
\begin{equation}
\sup_{t\in\mathbb{R}}\| \rme^{-\rmi t H} - \rme^{-\rmi t G}\| \ge \sqrt{2}.
\end{equation}
\end{prop}
\begin{proof}
\mbox{}
\begin{enumerate}
\item Let $g$ be an eigenvalue of $G$ and let $\varphi$ be its associated normalized eigenvector. One has
\begin{align}
\delta_\varphi^2(t)
&=\|(\rme^{-\rmi t H} - \rme^{-\rmi t G})\varphi\|^2 
\vphantom{\int}
\nonumber\\
&
= \|(\rme^{-\rmi t (H-g)} - 1)\varphi\|^2 
\vphantom{\int}
\nonumber\\
&= \int |\rme^{-\rmi t(\lambda-g)}-1|^2\,\rmd \mu_\varphi(\lambda)
\nonumber\\
&= 4\int \sin^2\!\left(\frac{\lambda-g}{2}t\right)\rmd \mu_\varphi(\lambda) \nonumber\\
&= 2 - 2 \int \cos[(\lambda-g)t]\,\rmd \mu_\varphi(\lambda), 
\end{align}
where $\mu_\varphi$ is the spectral measure of $H$ at $\varphi$. Thus, for any $T>0$,
\begin{equation}\label{eq:SupLowerBound}
\sup_{t\in\mathbb{R}} \delta_\varphi^2(t)
\geq \frac{1}{T} \int_0^T \delta_\varphi^2(t)\,\rmd t = 2 - 2  
\int \sinc[(\lambda-g)T]\,\rmd\mu_\varphi(\lambda)
\end{equation}
by Fubini's theorem. Notice that $\sinc(\lambda T) \to \mathbf{1}_{\{0\}}(\lambda)$  pointwise as $T\to+\infty$, where $\mathbf{1}_\Omega(\lambda)$ is the indicator function of the set $\Omega$. 
Therefore, by taking the limit $T\to+\infty$, by the dominated convergence theorem one gets
\begin{equation}
\label{eq:SupLowerBoundLimit}
\sup_{t\in\mathbb{R}} \delta_\varphi^2(t)   \geq 2 - 2\mu_\varphi(\{g\}).
\end{equation}
Notice that this result is valid for every self-adjoint $H$, irrespective of its spectrum.
If $H$ has a pure point spectrum,
with the spectral resolution
\begin{equation}
H= \sum_j h_j P_j,
\end{equation}
then the spectral measure at $\varphi$ of a set $\Omega\subset\mathbb{R}$ reads 
\begin{equation}
\mu_\varphi(\Omega) =
\sum_{j \,:\, h_j \in \Omega} \| P_j \varphi \|^2 ,
\end{equation}
which yields
\begin{equation}
\delta_\varphi^2(t)=2 - 2 \sum_j \| P_j \varphi \|^2 \cos[(h_j-g)t], 
\end{equation}
and	 
\begin{equation}
\mu_\varphi(\{g\}) = 
\sum_j  \delta_{h_j,g} \| P_j \varphi \|^2,
\end{equation}
with $\delta_{h,g}$ denoting the Kronecker delta.

\item Since by assumption $H$ and $G$ are not unitarily equivalent, either they have different spectra or they have the same spectrum with different multiplicities. In the first case, let $g$ be an eigenvalue of, say, $G$ with $g\notin\mathop{\mathrm{spec}}H$. Then, $\mu_\varphi(\{g\})=0$ and, by point 1 above,
\begin{equation}
\sup_{t\in\mathbb{R}}\| \rme^{-\rmi t H} - \rme^{-\rmi t G}\| \ge \sup_t \delta_\varphi(t) \geq \sqrt{2}.	
\label{eq:2.88}
\end{equation} 
In the second case, there is a common eigenvalue $g$ of $G$ and $H$ with different multiplicities. 
Let $P_{j_1}$ and $Q_{j_2}$, for some $j_1$ and $j_2$, be the associated eigenprojections of $H$ and $G$, respectively, i.e.\ $h_{j_1} = g_{j_2} =g$. We have, say, 
$\dim \mathop{\mathrm{ran}} Q_{j_2} > \dim \mathop{\mathrm{ran}} P_{j_1}$,
and thus there exists a unit vector $\varphi\in \mathop{\mathrm{ran}} Q_{j_2}$, with $P_{j_1} \varphi=0$.
Therefore,
\begin{equation}
\mu_\varphi(\{g\})= \| P_{j_1} \varphi \|^2 = 0,
\end{equation}
and~\eqref{eq:2.88} holds again.
\end{enumerate}
\end{proof}


\end{document}